\documentclass[11pt]{article}

\usepackage[margin=1in]{geometry}
\usepackage{amsmath,amssymb,amsfonts,amsthm}
\usepackage[numbers, sort]{natbib}
\usepackage{algorithm}
\usepackage{algorithmic}
\usepackage[dvipsnames,usenames]{xcolor}
\usepackage{graphicx}
\usepackage{enumitem}
\usepackage{subcaption}

\usepackage[colorlinks=true,urlcolor=blue,linkcolor=RoyalBlue,citecolor=OliveGreen]{hyperref}
\usepackage[nameinlink]{cleveref}
\Crefname{equation}{Eqn.}{Eqns.}
\crefname{equation}{}{}
\Crefname{section}{Section}{Sections}
\crefname{section}{Sec.}{Secs.}

\newtheorem{lemma}{Lemma}
\newtheorem{corollary}[lemma]{Corollary}
\newtheorem{theorem}{Theorem}
\newtheorem{definition}{Definition}

\newif\ifshowproof
\showprooftrue

\newcommand{\opt}{\textsc{Opt}}
\newcommand{\rev}{\textsc{Rev}}

\newcommand{\E}{\mathbf{E}}
\renewcommand{\Pr}{\mathbf{Pr}}

\newcommand{\kl}{D_\text{\rm{KL}}}
\newcommand{\skl}{D_\text{\rm{SKL}}}

\newcommand{\defeq}{\stackrel{\text{def}}{=}}

\newcommand{\shading}{s_{m, n, \delta}}
\newcommand{\doubleshading}{d_{m, n, \delta}}
\newcommand{\truncatetop}{t^{\max}}
\newcommand{\truncatebottom}{t^{\min}}

\title{Settling the Sample Complexity of Single-parameter Revenue Maximization}

\author{
	Chenghao Guo%
	\thanks{IIIS, Tsinghua University. Email: guoch16@mails.tsinghua.edu.cn}
	\and
	Zhiyi Huang%
	\thanks{The University of Hong Kong. Email: zhiyi@cs.hku.hk}
	\and
	Xinzhi Zhang%
	\thanks{IIIS, Tsinghua University. Email: zhang-xz16@mails.tsinghua.edu.cn}
}

\date{April, 2019}

\begin{document}
	
	\begin{titlepage}
		\thispagestyle{empty}
		\maketitle
		\begin{abstract}
	\thispagestyle{empty}
	This paper settles the sample complexity of single-parameter revenue maximization by showing matching upper and lower bounds, up to a poly-logarithmic factor, for all families of value distributions that have been considered in the literature.
	The upper bounds are unified under a novel framework, which builds on the strong revenue monotonicity by Devanur, Huang, and Psomas (STOC 2016), and an information theoretic argument.
	This is fundamentally different from the previous approaches that rely on either constructing an $\epsilon$-net of the mechanism space, explicitly or implicitly via statistical learning theory, or learning an approximately accurate version of the virtual values. 
	To our knowledge, it is the first time information theoretical arguments are used to show sample complexity upper bounds, instead of lower bounds. 
	Our lower bounds are also unified under a meta construction of hard instances.
\end{abstract}
	\end{titlepage}
	
	\section{Introduction}
\label{sec:introduction}

Suppose there is one item for sale and there are $n$ bidders.
Each bidder has a private value for the item that is independently, but not necessarily identically, drawn from the corresponding prior distribution, denoted as $\mathbf{D} = D_1 \times D_2 \times \dots \times D_n$.
What is the optimal mechanism in terms of the expected revenue?
This classic problem of revenue maximization was solved by \citet{Myerson/1981/MOR} analytically.
Given the prior distributions from which the bidders' values are drawn, in particular, given their cumulative distribution functions (cdf), denoted as $F_i$'s, and probability density functions (pdf), denoted as $f_i$'s, the optimal auction is characterized by the virtual value functions: 
\begin{equation}
\label{eqn:virtual_value}
\phi_i(v) = v - \frac{1 - F_i(v)}{f_i(v)}
~.
\end{equation}
Informally, the optimal auction lets the bidder with the largest non-negative virtual value win the item, and charges the winner a payment that equals the threshold value above which she wins.%
\footnote{In general, the optimal auction chooses the winner based on an ``ironed'' version of the virtual values. We will omit this in the introduction for simplicity of our discussions.}

From an algorithmic viewpoint, however, the problem is not fully settled because the cdf and pdf of the prior distributions are rarely given as input in practice.
\citet{ColeR/2014/STOC} initiated the study of the following sample complexity problem:
Suppose the algorithm has access to the prior distributions only in the form of i.i.d.\ samples, how many samples are sufficient and necessary for finding an approximately optimal auction?
In particular, \citet{ColeR/2014/STOC} considered the multiplicative $1 - \epsilon$ approximation, and regular and MHR distributions, and showed that the sample complexity is polynomial in the number of bidders $n$, and $\epsilon^{-1}$.
Subsequently, there is a long line of work in this direction, either to improve the sample complexity bounds~\cite{DevanurHP/2016/STOC, MorgensternR/2015/NIPS, Syrgkanis/2017/NIPS}, or to consider other families of distributions such as bounded support distributions, with both multiplicative and additive approximations~\cite{DevanurHP/2016/STOC, GonczarowskiN/2017/STOC}, or the special cases with a single bidder~\cite{HuangMR/2015/EC} or i.i.d.\ bidders~\cite{RoughgardenS/2016/EC}, or the generalization to multiple heterogeneous items~\cite{CaiD/2017/FOCS, BalcanSV/2018/EC, GonczarowskiW/2018/FOCS}.
Despite many efforts, we still cannot pin down the asymptotically optimal sample complexity for any of the families of distributions considered in the literature, other than the special case of a single bidder.
\Cref{tab:previous-multi-bidder-results} summarizes the state-of-the-art upper and lower bounds prior to this paper.

\subsection{Previous Approaches}
\label{sec:previous_approaches}

We first present a brief overview on the previous approaches for analyzing the sample complexity of revenue maximization, which can be categorized into two groups, and explain their limitations.

\paragraph{Statistical Learning Theory.}
The first approach relies on constructing an $\epsilon$-net of the mechanism space, namely, a subset of mechanisms such that for any distribution in the family, there always exists an approximately optimal mechanism in the subset.
Then, it remains to identify such an approximately optimal mechanism in the $\epsilon$-net.
This can be done via a standard concentration plus union bounds combo.
Informally, the resulting sample complexity will be:%
\footnote{This form relies on the assumption that $O(\epsilon^{-2})$ samples are sufficient for estimating the expected revenue of a mechanism up to an $\epsilon$ error, which need not be true in general especially with unbounded value distributions.}
\[
\frac{\log \big( \text{size of the $\epsilon$-net} \big)}{\epsilon^2}
~.
\]
The construction of the $\epsilon$-net can be either explicit (e.g., \cite{DevanurHP/2016/STOC, GonczarowskiN/2017/STOC,GonczarowskiW/2018/FOCS}), or implicit via various learning dimensions from statistical learning theory (e.g., \cite{MorgensternR/2015/NIPS, Syrgkanis/2017/NIPS}).

The main limitation of this approach is that the size of the $\epsilon$-net seems to have an unavoidable exponential dependence in $\epsilon^{-1}$ (see below for an example). Recall that the sample complexity upper bound will be $\log (\text{size of the }\epsilon-\text{net})/\epsilon^2$, this exponential dependence leads to an \emph{at least cubic dependence in $\epsilon^{-1}$} in the sample complexity upper bounds.
For example, we sketch below an explicit construction of the $\epsilon$-net by \citet{DevanurHP/2016/STOC}.
With an appropriate discretization, it suffices to consider $\epsilon^{-1}$ distinct values.
Further, since the optimal auction chooses the winner to maximize virtual value, it suffices to know the ordering of $0$ and $\phi_i(v)$'s for all $n$ bidders and all $\epsilon^{-1}$ values. 
Hence, the number of auctions that we need to consider is no more than the number of orderings over the $n \epsilon^{-1}$ virtual values $\phi_i(v)$'s and $0$, which equals $(n \epsilon^{-1} + 1)!$ and is singly exponential in both $n$ and $\epsilon^{-1}$.
Getting rid of the exponential dependence in $\epsilon^{-1}$ intuitively means that it suffices to consider a constant number of distinct values, which seems implausible.

\begin{table}
	\centering
	\renewcommand{\arraystretch}{1.2}
	\begin{tabular}{|c|c|c|}
		\hline
		Setting & Lower Bound & Upper Bound \\
		\hline
		Regular & $\Omega(\max \{ n\epsilon^{-1}, \epsilon^{-3} \} )$ \cite{ColeR/2014/STOC, HuangMR/2015/EC} & $\tilde{O}(n\epsilon^{-4})$ \cite{DevanurHP/2016/STOC} \\
		\hline
		MHR & $\Omega(\max \{n \epsilon^{-1/2}, \epsilon^{-3/2}\})$ \cite{ColeR/2014/STOC, HuangMR/2015/EC} & $\tilde{O}(n \epsilon^{-3})$ \cite{DevanurHP/2016/STOC, MorgensternR/2015/NIPS} \\
		\hline
		$[1,H]$ & $\Omega (H\epsilon^{-2})$ \cite{HuangMR/2015/EC} & $\tilde{O}(nH\epsilon^{-3})$ \cite{DevanurHP/2016/STOC} \\ 
		\hline
		$[0,1]$-additive & $\Omega(\epsilon^{-2})$ \cite{HuangMR/2015/EC} & $\tilde{O}(n\epsilon^{-3})$ \cite{DevanurHP/2016/STOC, GonczarowskiN/2017/STOC} \\
		\hline
	\end{tabular}
	\caption{Best known sample complexity bounds prior to this paper}
	\label{tab:previous-multi-bidder-results}
\end{table}

\paragraph{Learning the Virtual Values.}
An alternative approach (e.g., \cite{ColeR/2014/STOC, RoughgardenS/2016/EC}) is to learn the individual value distributions well enough to obtain enough approximately accurate information about the virtual values, which induces a mechanism.
Then, we analyze the revenue approximation using the connections between expected revenue and virtual values.
Importantly, this approach does not need to take a union bound over exponentially many candidate mechanisms, circumventing the bottleneck that introduces the undesirable cubic dependence in $\epsilon^{-1}$ in the learning theory approach. 
Indeed, for the special case of independently and identically distributed (i.i.d.) bidders with $[0, 1]$-bounded distributions and additive approximation, \citet{RoughgardenS/2016/EC} showed a sample complexity upper bound of $\tilde{O}(n^2 \epsilon^{-2})$, which is the only previous example, to our knowledge, with a sub-cubic dependence in $\epsilon^{-1}$.

The main limitation of this approach roots in the form of the virtual value as defined in \Cref{eqn:virtual_value}.
It involves three components, the value $v$, the complementary cumulative distribution function $1 - F_i(v)$, a.k.a., the quantile, and the pdf $f_i(v)$.
Here, the value $v$ is given as input;
the quantile $1 - F_i(v)$ is relatively easy to estimate accurately via standard concentration inequalities. 
It is, however, impossible to get an accurate estimation of the density function $f_i$ in general.
As a result, it is infeasible to learn the virtual values accurately point-wise.
This is a major technical hurdle that prevents existing works using this approach from getting tight sample complexity upper bounds; in particular, they all have \emph{super-linear dependence in $n$}.
Even for the special case of i.i.d.\ bidders, the bound is quadratic in $n$~\cite{RoughgardenS/2016/EC}; the dependence is at least $n^7$ for the general case~\cite{ColeR/2014/STOC}.
Note that a linear dependence in $n$ follows almost trivially from the learning theory approach (e.g., \cite{DevanurHP/2016/STOC}).

\paragraph{Prior Knowledge of the Distribution Family.}

Another limitation of the existing approaches is that they generally rely on knowing the family of distributions upfront.
Even for the special case of a single bidder, the best known algorithms are different for regular, MHR, and bound-support distributions (e.g., \cite{HuangMR/2015/EC}).
For MHR distributions, we may simply pick the optimal price with respect to (w.r.t.) the empirical distribution, i.e., the uniform distribution over the samples.
For regular and $[1,H]$ bounded support distributions, however, we need to introduce a threshold $\delta > 0$ and to choose the optimal price subject to having a sale probability at least $\delta$.
Further, the threshold is chosen differently for regular and $[1,H]$ bounded support distributions.
If we fail to introduce a threshold when it is an arbitrary regular distribution, the expected revenue may not converge to the optimal at all~\cite{DhangwatnotaiRY/2015/GEB}.
If we set the threshold under the belief that the distribution has a $[1,H]$ bounded support while it is in fact an arbitrary regular distribution, the convergence rate will be far from optimal.
See \Cref{app:prior_knowledge} for a concrete example.
It would definitely be nice to have a more robust algorithm.

\subsection{Our Contributions}

We introduce an algorithm that achieves the optimal sample complexity, up to a poly-logarithmic factor, simultaneously for all families of distributions that have been considered in the literature.
Our upper and lower bounds, summarized in \Cref{tab:multi-bidder-results}, improve the best known bounds in all cases.

\begin{table}
	\centering
	\renewcommand{\arraystretch}{1.2}
	\begin{tabular}{|c|c|c|}
		\hline
		Setting & Lower Bound (\cref{sec:lower_bound_multi_bidder}) & Upper Bound (\cref{sec:upper_bound_multi_bidder}) \\
		\hline
		Regular & $\Omega(n\epsilon^{-3})$ & $\tilde{O}(n\epsilon^{-3})$ \\
		\hline
		MHR & $\tilde{\Omega}(n \epsilon^{-2})$ & $\tilde{O}(n\epsilon^{-2})$ \\
		\hline
		$[1,H]$ & $\Omega(nH\epsilon^{-2})$ & $\tilde{O}(nH\epsilon^{-2})$ \\
		\hline
		$[0,1]$-additive & $\Omega(n\epsilon^{-2})$ & $\tilde{O}(n\epsilon^{-2})$ \\
		\hline
	\end{tabular}
	\caption{Sample complexity bounds in this paper}
	\label{tab:multi-bidder-results}
\end{table}

\paragraph{Our Algorithm.}
The algorithm constructs from the samples a dominated empirical distribution, denoted as $\mathbf{\tilde{E}} = \tilde{E}_1 \times \tilde{E}_2 \times \dots \times \tilde{E}_n$, which is dominated by the true value distribution $\mathbf{D}$ in the sense of first-order stochastic dominance, but is as close to $\mathbf{D}$ as possible.
Then, it chooses the optimal mechanism w.r.t.\ $\mathbf{\tilde{E}}$.
We call it the dominated empirical Myerson auction.

To construct the dominated empirical distribution, we first look at the estimation error by the empirical distribution, in terms of the difference between the empirical quantiles and the true quantiles.
This can be bounded using standard concentration inequalities. 
For example, suppose a value $v$ has quantile $q$.
Then, Bernstein inequality gives that, with high probability, its quantile in the empirical distribution is approximately equal to $q$, up to an additive error of:
\begin{equation}
\label{eqn:intro_estimation_error}
\tilde{O} \left( \sqrt{\frac{q(1-q)}{m}} \right)
~.
\end{equation}
To ensure that the error bound holds for all values, one can simply take a union bound at the cost of an extra logarithmic factor inside the square root. 
Intuitively, the dominated empirical distribution is obtained by subtracting this term from the quantile of each value $v$ in the empirical distribution.
See \Cref{sec:upper_bound_multi_bidder} for the formal definitions of the dominated empirical distribution and the algorithm.

Next we explain the main difference between our algorithm and those in previous works, with the exception of~\citet{RoughgardenS/2016/EC}.
Previous works generally pick the optimal auction w.r.t.\ the empirical distribution, with a distribution-family-dependent preprocessing on the sample values, in the form of truncating large but rare values and/or a discretization of the values.
The preprocessing is to avoid choosing the auction based on some rare but high values in the samples.
In contrast, our algorithm picks the optimal auction w.r.t.\ the dominated empirical distribution, without any preprocessing or any knowledge of the underlying family of distributions.
The conservative estimates of quantiles by the dominated empirical distribution implicitly tune down the impact of rare but high values, simultaneously for all families of distributions.

The algorithm by \citet{RoughgardenS/2016/EC} is the most similar one to ours.
They also constructed a dominated empirical distribution and picked the corresponding optimal auction.
A subtle difference is that they used the Dvoretzky-Kiefer-Wolfowitz (DKW) inequality~\cite{DvoretzkyKW/1956/AoMS} to bound the estimation error of the empirical distribution and to construct the dominated empirical distribution, which on one hand avoided losing a logarithmic factor from the union bound, but on the other hand did not get the better bounds for values with quantiles close to $0$ or $1$ as in \Cref{eqn:intro_estimation_error}.
The latter property is crucial for our analysis.
We leave as an interesting open question whether there is a strengthened version of the DKW inequality with quantile-dependent bounds.
Such an inequality will improve the logarithmic factor in the upper bounds of this paper.
We stress that while the algorithms are similar in spirit, our analysis is fundamentally different, as we will explain next. 
Importantly, our sample complexity upper bounds hold for the general non-i.i.d.\ case while the upper bound of \citet{RoughgardenS/2016/EC} holds only for the special case of i.i.d.\ bidders.

\paragraph{Analysis via Revenue Monotonicity.}
Our analysis consists of two components.
The first one is two inequalities that lower bound the expected revenue of the dominated empirical Myerson auction on the true distribution, where the inequalities are enabled by the strong revenue monotonicity of single-parameter problems by \citet{DevanurHP/2016/STOC}.
The strong revenue monotonicity states that the optimal auction w.r.t.\ a distribution that is dominated by the true distribution gets at least the optimal revenue of the dominated distribution.
In particular, running the dominated empirical Myerson on the true value distribution $\mathbf{D}$ gets at least the optimal revenue of the dominated empirical distribution $\mathbf{\tilde{E}}$.
Further, consider a doubly shaded version of the true distribution, denoted as $\mathbf{\tilde{D}}$, which intuitively is obtained by subtracting twice the error term in \Cref{eqn:intro_estimation_error} from the quantiles of the true distribution.
Then, $\mathbf{\tilde{D}}$ is dominated by $\mathbf{\tilde{E}}$ and, thus, its optimal revenue is at most that of $\mathbf{\tilde{E}}$.
This weaker notion of revenue monotonicity is folklore in the literature and follows as a direct corollary of the stronger notion.
Therefore, we conclude that the expected revenue of the dominated empirical Myerson auction is at least the optimal revenue of the doubly shaded distribution $\mathbf{\tilde{D}}$.
It remains to compare the optimal revenue of $\mathbf{D}$ and $\mathbf{\tilde{D}}$.

This idea is quite powerful on its own. 
The key observation is that $\mathbf{\tilde{D}}$ approximately preserves the probability density/mass of $\mathbf{D}$ almost \emph{point-wise}, except for a small subset of values that have little impact on the optimal revenue.
Intuitively, this is because it consistently underestimates the quantiles; in contrast, the empirical distribution has fluctuations in its estimations.  
Hence, $\mathbf{\tilde{D}}$ approximately preserves the virtual values of $\mathbf{D}$ almost point-wise, circumventing the technical hurdle faced by the second previous approach discussed in \Cref{sec:previous_approaches}.
By this idea and standard accounting arguments for the expected revenue, we can get the optimal sample complexity upper bound for regular distributions in \Cref{tab:multi-bidder-results}, and match the best previous upper bounds for the other three families of distributions in \Cref{tab:previous-multi-bidder-results}.
We present a formal discussion in \Cref{app:without_information_theory}.



\paragraph{Analysis via Information Theory.}
To get the optimal sample complexity upper bounds for all families of distributions under a unified framework, we need the second idea, namely, to bound the difference between the optimal revenues of $\mathbf{D}$ and $\mathbf{\tilde{D}}$ with an information theoretic argument.
The argument consists of two claims: 1) the distributions $\mathbf{D}$ and $\mathbf{\tilde{D}}$ are similar in the information theoretic sense so that it takes many samples to distinguish them, and 2) we can estimate the expected revenue of any given mechanism on $\mathbf{D}$ and $\mathbf{\tilde{D}}$ with a small number of samples.
Concretely, we will show that the Kullback-Leibler (KL) divergence between $\mathbf{D}$ and $\mathbf{\tilde{D}}$ is at most $\tilde{O}(\frac{n}{m})$, omitting some caveats which we will explain in details in \Cref{sec:upper_bound_multi_bidder}.
By standard information theoretic arguments, it implies that one needs at least $\tilde{\Omega}(\frac{m}{n})$ samples to distinguish these two distributions.
For example, consider a $[0, 1]$-bounded distribution $\mathbf{D}$ and an additive $\epsilon$ approximation. 
Suppose $m$ is at least $\tilde{O}(n \epsilon^{-2})$ as in \Cref{tab:multi-bidder-results}. 
Then, we get that it takes at least $C \cdot \epsilon^{-2}$ samples to distinguish $\mathbf{D}$ and $\mathbf{\tilde{D}}$ for some sufficiently large constant $C > 0$.
On the other hand, it takes less than $C \cdot \epsilon^{-2}$ samples to estimate the expected revenue of any mechanism on both $\mathbf{D}$ and $\mathbf{\tilde{D}}$ up to an additive $\epsilon$ factor.
Thus, the expected revenue of any mechanism differs by at most $\epsilon$ on the two distributions; otherwise, we can distinguish them with less than $C \cdot \epsilon^{-2}$ samples by estimating the expected revenue of the mechanism.
As a result, the optimal revenues of $\mathbf{D}$ and $\mathbf{\tilde{D}}$ differ by at most $\epsilon$.

To our knowledge, this is the first time information theory is used to show sample complexity upper bounds for revenue maximization.
Previously, it was used only for lower bounds (e.g., \cite{HuangMR/2015/EC}).
We believe it will find further applications in studying the sample complexity of multi-parameter revenue maximization and other learning problems.
We stress that our algorithm is constructive and, in fact, can be implemented in quasi-linear time;%
\footnote{For each bidder, it takes $O(m \log m)$ time to sort the samples and to compute the quantiles of the empirical distribution, and $O(m)$ time to compute the quantiles of the dominated empirical distribution, and $O(m \log m)$ times to compute the convex hull of the corresponding revenue curve, which characterizes the optimal auction.}
both the doubly shaded distribution $\mathbf{\tilde{D}}$ and the information theoretic arguments are used only in the analysis.

\paragraph{Lower Bound Constructions.}
Our lower bounds are unified under a meta construction, with some components chosen based on the family of distributions.
We briefly sketch the construction below.
Let the first bidder's value distribution be a point mass.
She will serve as the default winner in the optimal auction.
The value distribution of each of the other $n-1$ bidders will be either $D^h$ or $D^\ell$.
These two distributions satisfy that there is a value interval such that for any value in it, the corresponding virtual value wins over bidder $1$ if and only if the distribution is $D^h$.
Both $D^h$ and $D^\ell$ will have an $O(\frac{1}{n})$ chance of realizing a value in this interval.
Intuitively, to find a near optimal mechanism we must be able to distinguish the bidders with distribution $D^h$ from those with distribution $D^\ell$.
Finally, we will construct $D^h$ and $D^\ell$ to be similar so that it takes many samples to distinguish them.
The meta construction, inspired by the hard instances by \citet{ColeR/2014/STOC}, can be viewed as a non-trivial generalization of the lower bound framework by \citet{HuangMR/2015/EC} for the special case of single bidder.

\subsection{Other Related Works}

Prior to \citet{ColeR/2014/STOC}, there were a few sporadic works that had the flavor of learning the optimal price/auction from samples (e.g., \cite{Elkind/2007/SODA, DhangwatnotaiRY/2015/GEB}).

The learning theory approach has also been used to learn approximately optimal auction among a restricted family of simple auctions, both for single-parameter problems~\cite{MorgensternR/2016/COLT}, and for multi-parameter problems~\cite{MorgensternR/2015/NIPS, CaiD/2017/FOCS, BalcanSV/2016/NIPS, BalcanSV/2018/EC, Syrgkanis/2017/NIPS}.
To learn an approximately optimal auction without restrictions in multi-parameter problems, \citet{DughmiHN/2014/WINE} showed that it needed exponentially many samples in general;

\citet{GonczarowskiW/2018/FOCS} proved a polynomial sample complexity upper bound for the special case when bidders' valuations were additive, if we allowed approximate truthfulness.

The online learning version has also been considered, both in the full information setting, i.e., the seller runs a direction revelation auction and observes the bidder's valuation, and in the bandit setting, i.e., the seller runs a posted price auction and only observes if the bidder buys the item.
\citet{BlumH/2005/SODA} introduced the optimal algorithm in terms of a regret bound that scaled with $H$, the upper bound on bidders' values.
\citet{BubeckDHN/2017/EC} further improved the regret bound to scale with the optimal price instead of $H$, and their algorithm matched the optimal sample complexity bounds when the bidder's values in different rounds were i.i.d.\ from a prior distribution.

Intriguingly, even the weaker notion of revenue monotonicity ceases to hold in multi-parameter problems~\cite{HartR/2015/TE}, while approximate versions are showed for restricted families of valuations~\cite{RubinsteinW/2015/EC, Yao/2018/SAGT}.

	\section{Preliminaries}
\label{sec:prelim}

\subsection{Model}
Let there be a single item for sale, and let there be $n$ bidders.
Each bidder $i$ has a private valuation $v_i \ge 0$ for getting the item, where $v_i$ is independently drawn from the corresponding prior distribution $D_i$.
Thus, the value profile $\mathbf{v} = (v_1, v_2, \dots, v_n)$ follows a product distribution $\mathbf{D} = D_1 \times D_2 \times \dots \times D_n$.
We consider direct revelation mechanisms, each of which consists of an allocation function $\mathbf{x}$ and a payment function $\mathbf{p}$. 
First, each bidder submits a bid $b_i \ge 0$.
Then, $x_i(\mathbf{b})$ denotes the probability that bidder $i$ gets the item, and $p_i(\mathbf{b})$ denotes the expected payment by bidder $i$.
Since there is only one item, we have $\sum_{i = 1}^n x_i(\mathbf{b}) \le 1$ for all $\mathbf{b}$.
Each bidder $i$'s utility is $v_i \cdot x_i(\mathbf{b}) - p_i(\mathbf{b})$.
The seller seeks to maximize the expectation of the revenue, which is the sum of bidders' payments, $\sum_{i = 1}^n p_i(\mathbf{b})$.

We remark that our algorithm and the framework for proving sample complexity upper and lower bounds apply to more general single-parameter problems under matroid constraints.
We defer such extensions to Appendix~\ref{app:matroid_ub} and Appendix~\ref{app:matroid_lb}.

By the revelation principle, we focus on \emph{Bayesian incentive compatible} (BIC) mechanisms, which mean that for any bidder $i$ and any value $v_i$, conditioned on the other bidders bidding truthfully, i.e., $\mathbf{b}_{-i} = \mathbf{v}_{-i}$, bidding $b_i = v_i$ maximizes bidder $i$'s expected utility over the randomness of other bidders' values, and guarantees non-negative expected utility.
A stronger notion is \emph{dominant strategy incentive compatible} (DSIC) mechanisms, which means that bidding $b_i = v_i$ always maximizes bidder $i$'s utility, and guarantees it is non-negative, no matter what other bidders bid.


\paragraph{Myerson's Optimal Auction.}
If the prior distribution $\mathbf{D}$ is given as input, the revenue maximizing mechanisms is fully characterized by \citet{Myerson/1981/MOR}.
Interestingly, Myerson's optimal auction is DSIC but is optimal among all BIC mechanisms. The characterization relies on the following notion of \emph{virtual values}.
We first explain this notion assuming the distributions are continuous and have positive densities as in Myeron's original paper.
For any bidder $i$, let $F_i$ and $f_i$ denote the cdf and pdf of the value distributions, the virtual value of bidder $i$ when her value is $i$ is $\phi_i(v_i)=v_i-\frac{1-F_i(v_i)}{f_i(v_i)}$.
Let $q_i(v_i) = \Pr_{D_i} [ v > v_i]$ be the \emph{quantile} of $v_i$.
We have $q_i(v_i) = 1 - F_i(v_i)$ if $D_i$ is continuous.

If for all $i$, the virtual value $\phi_i(v_i)$ is monotonically non-decreasing in $v_i$, the distribution $\mathbf{D}$ is said to be \emph{regular}. 
If $\phi_i(v_i)$ further has derivatives at least $1$ point-wise, $\mathbf{D}$ is said to have \emph{monotone hazard rate} (MHR).
Discrete versions of regular and MHR distributions over non-negative integers are also considered in the literature~\cite{Elkind/2007/SODA, BarlowMP/1963/AoMS}, where $f_i(v_i)$ is replaced with the probability mass of $v_i$.
The optimal auction is simple if the value distributions are MHR or even regular.
It lets the bidder with the largest non-negative virtual value win the item, breaking ties arbitrarily;
if no bidder has a non-negative virtual value, no one gets the item. 
The winner pays the threshold value at or above which she wins.

For general distributions, virtual values may not be monotone.
We need an extra step that defines an \emph{ironed} version of the virtual value that is monotone.
We will use the following definition of \emph{ironed virtual values} so that it generalizes to general distributions that may be a mixture of continuous and discrete distributions.
Define the mapping from quantiles to values as $v_i(q) = \sup \{ v : 1 - F_i(v) \le q \}$.
Define the \emph{revenue curve} over the quantile space as $R_i(q) = q \cdot v_i(q)$.
Let the \emph{ironed revenue curve} $\bar{R}_i(q)$ be the convex hull of $R_i(q)$.
The \emph{ironed virtual value} $\bar{\phi}_i(v_i)$ is the right derivative of $\bar{R}_i \big( 1 - F_i(v_i) \big)$. 
Then, Myerson's optimal auction picks a winner based on the ironed virtual value instead of the virtual value, and charges the threshold value accordingly.

For any mechanism $M$ and any distribution $\mathbf{D}$, we let $\rev(M, \mathbf{D})$ denote the expected revenue of running $M$ on $\mathbf{D}$.
Let $M_\mathbf{D}$ denote Myerson's optimal auction for $\mathbf{D}$. 
For concreteness, assume $M_\mathbf{D}$ breaks ties over bidders with the same virtual values in the lexicographical order.
Let $\opt(\mathbf{D}) = \rev(M_\mathbf{D}, \mathbf{D})$ denote the optimal revenue, which is given by Myerson's optimal auction. 

\paragraph{Sample Complexity.}
Now suppose we can access the prior distribution $\mathbf{D}$ only in the form of $m$ i.i.d.\ samples. 
For a give family of distributions $\mathcal{D}$ (e.g., regular, MHR, bounded support, etc.), the sample complexity of the revenue maximization problem is defined to be the (asymptotically) smallest number $m$ so that there is an algorithm satisfying that for any distribution $\mathbf{D} \in \mathcal{D}$, given $m$ i.i.d.\ samples from $\mathbf{D}$, it learns a mechanism that is a $1 - \epsilon$ multiplicative approximation in revenue with high probability.
We are also interested in an $\epsilon$ additive approximation in some cases.


\subsection{Technical Preliminaries}

\paragraph{Bernstein Inequality.}
Our algorithm and analysis will make use of the standard concentration bound by \citet{Bernstein/1924}, as stated in the next lemma.
\begin{lemma}\label{lem:bernstein}
	Let $X_1, X_2, \dots, X_m$ be i.i.d.\ random variables such that $\E [ X_i ] = 0$, $\E [ X_i^2 ] = \sigma^2$, and $|X_i| \le M$ for some constant $M > 0$.
	Then, for all positive $t$, we have:
	\[
		\textstyle
		\Pr \big[ ~ \big|\sum_{i=1}^m X_i \big| > t ~ \big] \le 2 \exp \big( -\frac{t^2}{ 2m \sigma^2 + (2/3) Mt} \big)
		~.
	\]
\end{lemma}

\paragraph{Strong Revenue Monotonicity.}
A distribution $\mathbf{D}$ first-order stochastically dominates another distribution $\mathbf{\tilde{D}}$, or simply $\mathbf{D}$ dominates $\mathbf{\tilde{D}}$ for brevity, if for every $i \in [n]$, $D_i$ dominates $\tilde{D}_i$ in that for every value $v$, its quantile in $D_i$ is weakly larger than that in $\tilde{D}_i$.
We denote this by $\mathbf{D} \succeq \mathbf{\tilde{D}}$.

\citet{DevanurHP/2016/STOC} showed a strong notion of revenue monotonicity as follows:

\begin{lemma}[Strong Revenue Monotonicity~\cite{DevanurHP/2016/STOC}]
	\label{lem:strong_revenue_monotonicity}
	Let $\mathbf{D}$ and $\mathbf{\tilde{D}}$ be two product value distributions such that $\mathbf{D} \succeq \mathbf{\tilde{D}}$.
	Recall that $M_{\mathbf{\tilde{D}}}$ is the optimal auction for $\mathbf{\tilde{D}}$.
	Then, we have:
	\[
		\rev(M_\mathbf{\tilde{D}}, \mathbf{D}) \ge \rev(M_\mathbf{\tilde{D}}, \mathbf{\tilde{D}})
		~.
	\]
\end{lemma}

The weaker notion of revenue monotonicity that is folklore in the literature follows as a corollary.

\begin{lemma}[Weak Revenue Monotonicity]
	\label{lem:weak_revenue_monotonicity}
	Let $\mathbf{D}$ and $\mathbf{\tilde{D}}$ be two product value distributions such that $\mathbf{D} \succeq \mathbf{\tilde{D}}$.
	Then, we have:
	\[
	\opt(\mathbf{D}) \ge \opt(\mathbf{\tilde{D}})
	~.
	\]
\end{lemma}

\paragraph{Information Theory.}
Consider two probability measure $P$ and $Q$ over a sample
space~$\Omega$.
The {\em Kullback-Leibler (KL) divergence} is defined as follows:
\[ 
\kl(P \| Q) = \int_{\Omega} \ln \left( \frac{dP}{dQ} \right) dP
~. 
\]

We further consider the following symmetric version:
\[
\skl(P, Q) = \kl(P \| Q) + \kl(Q \| P)
~.
\]

A classification algorithm $A : \Omega^m \mapsto \{ P, Q \}$ distinguishes $P$ and $Q$ correctly with $m$ samples if for any $D \in \{P, Q\}$, $A(\omega_1, \omega_2, \dots, \omega_m) = D$ with probability at least $\frac{2}{3}$, where $\omega_1, \omega_2, \dots, \omega_m$ are i.i.d.\ samples from $D$.
The upper and lower bounds in this paper both use the following connection between the number of samples needed to distinguish two distributions and their KL divergence.

\begin{lemma}[e.g., see \cite{HuangMR/2015/EC}]
	\label{lem:kl_divergence}
	Suppose there is a classification algorithm that distinguishes $P$ and $Q$ correctly with $m$ samples.
	Then, the number of samples $m$ is at least:
	\[
	\Omega \big( \skl(P, Q)^{-1} \big)
	~.
	\]
\end{lemma}

	\section{Upper Bounds}
\label{sec:upper_bound_multi_bidder}

We present in this section an algorithm and its analysis that achieve the optimal sample complexity, up to a poly-logarithmic factor, simultaneously for all families of distributions in the literature.
The proofs of some lemmas that are relatively standard are deferred to \Cref{app:missing_proofs_upper_bounds}.
We also include a discussion on the optimality of our algorithm in the special case of a single bidder in \Cref{app:single_ub}.

\subsection{Dominated Empirical Myerson}
\label{sec:dominated_empirical_myerson}

We first define the following function:
\begin{equation}
\label{eqn:shading}
\shading(q) \defeq \max \left\{ 0, q - \sqrt{\frac{2 q(1-q) \ln(2mn\delta^{-1})}{m}} - \frac{4 \ln(2mn\delta^{-1})}{m} \right\}
~.
\end{equation}

For a value distribution $D$, we abuse notation by letting $\shading(D)$ denote a distribution such that for any value $v > 0$ with quantile $q$ in $D$, its quantile in $\shading(D)$ is $\shading(q)$.
Let $\shading(\mathbf{D})$ denote the product distribution obtained by applying $\shading$ to each coordinate of $\mathbf{D}$.

Given this function, we now present our algorithm below as \Cref{alg:dominated_empirical_myerson}.

%
%
%
%

\begin{algorithm}
	\begin{algorithmic}[1]
		\REQUIRE $m$ i.i.d.\ samples from the value distribution $\mathbf{D} = D_1 \times D_2 \times \dots \times D_n$
		\ENSURE a mechanism that decides the allocation and payment given bids from $n$ bidders
		\STATE Let $E_i$ be the empirical distribution, i.e., the uniform distribution over the samples of bidder $i$.\\
		\STATE Let $\mathbf{E} = E_1 \times E_2 \times \dots \times E_n$. 
		\STATE Let $\mathbf{\tilde{E}} = \shading(\mathbf{E})$. 
		That is, let $q^{E_i}(v)$ be the quantile of $E_i$; the quantile of $\tilde{E}_i$ is as follows:
		\begin{equation*}
		q^{\tilde{E}_i}(v)=
		\begin{cases}
		\max \left\{ 0, q^{E_i}(v) - \sqrt{\frac{2 q^{E_i}(v) \left( 1-q^{E_i}(v) \right) \ln (2mn \delta^{-1})}{m}} - \frac{4 \ln (2mn \delta^{-1})}{m} \right\} & \text{if $v > 0$} \\
		1 & \text{if $v = 0$} 
		\end{cases}
		\end{equation*}
		\STATE Output Myerson's optimal auction $M_{\mathbf{\tilde{E}}}$ w.r.t.\ $\mathbf{\tilde{E}} = \tilde{E}_1 \times \tilde{E}_2 \times \dots \times \tilde{E}_n$.
	\end{algorithmic}
	\caption{Dominated Empirical Myerson}
	\label{alg:dominated_empirical_myerson}	
\end{algorithm}

Our algorithm relies on constructing from the samples a distribution $\mathbf{\tilde{E}}$ dominated by the true value distribution but is as close to it as possible in a sense.
We will refer to $\mathbf{\tilde{E}}$ as the dominated empirical distribution, which is intuitively a shaded version of the empirical distribution via function $\shading$.
This is formalized by the following two lemmas.

\begin{lemma}
	\label{lem:empirical_error_bound}
	With probability at least $1 - \delta$, for any value $v \ge 0$, for any bidder $i \in [n]$, $v$'s quantiles in $D_i$ and $E_i$ satisfy that:
	\[
	\left| q^{E_i}(v) - q^{D_i}(v) \right| \le \sqrt{\frac{2 q^{D_i}(v) \big( 1-q^{D_i}(v) \big) \ln(2mn\delta^{-1})}{m} } + \frac{\ln(2mn\delta^{-1})}{m}
	~.
	\]
\end{lemma}

\begin{lemma}
	\label{lem:D_and_tilde_E}
	Assuming the bounds in \Cref{lem:empirical_error_bound}, we have:
	\[
	\mathbf{D} \succeq \mathbf{\tilde{E}}
	~.
	\]
\end{lemma}

We show that \Cref{alg:dominated_empirical_myerson} is ``universally'' optimal in the sense that it achieves the following sample complexity upper bounds simultaneously for all families of distributions in the literature.
We will establish their optimality, up to a poly-logarithmic factor, with the lower bounds in \Cref{sec:lower_bound_multi_bidder}.

\begin{theorem}
	\label{thm:main_upper_bounds}
	For any $0 < \epsilon < 1$ and any $n$-bidder product value distribution $\mathbf{D}$, \Cref{alg:dominated_empirical_myerson} returns a mechanism with an expected revenue at least $(1 - \epsilon) \opt(\mathbf{D})$, with probability at least $1-\delta$, if:
	\begin{enumerate}
		\item $m$ is at least $O(n\epsilon^{-3} \ln(n\epsilon^{-1}) \ln(n\delta^{-1}\epsilon^{-1})) = \tilde{O}(n \epsilon^{-3})$ and $\mathbf{D}$ is regular; or
		\item $m$ is at least $O(n\epsilon^{-2} \ln(\epsilon^{-1}) \ln(n\epsilon^{-1}) \ln(n\delta^{-1}\epsilon^{-1})) = \tilde{O}(n \epsilon^{-2})$ and $\mathbf{D}$ is MHR; or
		\item $m$ is at least $O(nH\epsilon^{-2} \ln(nH\epsilon^{-1}) \ln(n H\epsilon^{-1} \delta^{-1})) = \tilde{O}(n H \epsilon^{-2})$ and $\mathbf{D}$ has a bounded support in $[1, H]$ coordinate-wise.
	\end{enumerate}
	We also have that its expected revenue is at least $\opt(\mathbf{D}) - \epsilon$ if:
	\begin{enumerate}
		\setcounter{enumi}{3}
		\item $m$ is at least $O(n\epsilon^{-2} \ln(n\epsilon^{-1}) \ln(n\epsilon^{-1}\delta^{-1})) = \tilde{O}(n \epsilon^{-2})$ and $\mathbf{D}$ has a bounded support in $[0, 1]$ coordinate-wise.
	\end{enumerate}
\end{theorem}

%

\subsection{Meta Analysis}
\label{sec:meta_analysis}

\subsubsection*{Step 1: Analysis via Revenue Monotonicity}

The first idea in our analysis is to lower bound the expected revenue of the dominated empirical Myerson auction with inequalities enabled by revenue monotonicity.

We start by defining an auxiliary distribution $\mathbf{\tilde{D}}$ which intuitively is a doubly shaded version of the original distribution $\mathbf{D}$ such that it is dominated by $\mathbf{\tilde{E}}$.
Consider the following function:
\begin{equation}
\label{eqn:double_shading}
\doubleshading(q) \defeq \max \left\{ 0, q - \sqrt{\frac{8 q(1-q) \ln(2mn\delta^{-1})}{m}} - \frac{7 \ln(2mn\delta^{-1})}{m} \right\}
\end{equation}
We further allow it to operate on distributions the same way as the previous function $\shading$.
Then, let $\mathbf{\tilde{D}} = \doubleshading(\mathbf{D})$ be the auxiliary distribution. 

\begin{lemma}
	\label{lem:tilde_E_and_tilde_D}
	Assuming the bounds in \Cref{lem:empirical_error_bound}, we have that:
	\[
	\mathbf{\tilde{E}} \succeq \mathbf{\tilde{D}}
	~.
	\]
\end{lemma}

Next, we lower bound the expected revenue of the dominated empirical Myerson auction by the optimal revenue of the auxiliary distribution using revenue monotonicity.

\begin{lemma}
	\label{lem:analysis_revenue_monotonicity}
	With probability at least $1 - \delta$, we have:
	\[
	\rev(M_{\mathbf{\tilde{E}}}, \mathbf{D}) \ge \opt(\mathbf{\tilde{D}})
	~.
	\]
\end{lemma}

\begin{proof}
	We will prove the inequality when the bounds in \Cref{lem:empirical_error_bound} hold, which happens with probability at least $1 - \delta$.
	It follows from the following sequence of inequalities:
	\begin{align*}
	\rev(M_{\mathbf{\tilde{E}}}, \mathbf{D}) & \ge \rev(M_{\mathbf{\tilde{E}}}, \mathbf{\tilde{E}}) && \text{($\mathbf{D} \succeq \mathbf{\tilde{E}}$ by \Cref{lem:D_and_tilde_E}, strong revenue monotonicity by \Cref{lem:strong_revenue_monotonicity})} \\
	& = \opt(\mathbf{\tilde{E}}) \\
	& \ge \opt(\mathbf{\tilde{D}}) && \text{($\mathbf{\tilde{E}} \succeq \mathbf{\tilde{D}}$ by \Cref{lem:tilde_E_and_tilde_D}, weak revenue monotonicity by \Cref{lem:weak_revenue_monotonicity})}
	\end{align*}
\end{proof}


Given the above inequality, it suffices to show $\opt(\mathbf{\tilde{D}}) \ge (1 - \epsilon) \opt(\mathbf{D})$ for the first three cases which consider multiplicative approximation in \Cref{thm:main_upper_bounds}, and to show $\opt(\mathbf{\tilde{D}}) \ge \opt(\mathbf{D}) - \epsilon$ for the last case which considers additive approximation.

As we remarked in \Cref{sec:introduction}, this idea of bounding the expected revenue of the dominated empirical Myerson auction via revenue monotonicity, instead of concentration inequalities as in previous works, is quite powerful on its own.
In particular, the auxiliary distribution $\mathbf{\tilde{D}}$ approximately preserves the density and virtual value of the original distribution $\mathbf{D}$ almost \emph{point-wise}.
We will explain in \Cref{app:without_information_theory} how to build on this observation and standard accounting techniques for expected revenue to show the optimal sample complexity upper bounds for the case of regular distributions as stated in \Cref{thm:main_upper_bounds}, and for the other three cases weaker upper bounds which nevertheless match the best previous bounds already.


\subsubsection*{Step 2: Analysis via Information Theory}

Our second idea is to use an information theoretic argument to show that the optimal revenue of $\mathbf{\tilde{D}}$ is a $1 - \epsilon$ approximation (or an $\epsilon$ additive approximation) to that of $\mathbf{D}$.
Let us first explain what the analysis looks like in an idealized world, using the $[0, 1]$-bounded case as a running example.
Suppose that $m \ge \tilde{O}(n \epsilon^{-2})$ as stated in \Cref{thm:main_upper_bounds}.
The analysis builds on two observations:
\begin{enumerate}
	\item 
	$\mathbf{D}$ and $\mathbf{\tilde{D}}$ are close in KL divergence:
	\[
	\skl \big( \mathbf{D}, \mathbf{\tilde{D}} \big) \le \tilde{O} \bigg( \frac{n}{m} \bigg) = O \big( \epsilon^2 \big)
	~,
	\] 
	for a sufficiently small constant in the second asymptotic notation.
	This is the main technical component behind the information theoretic argument.
	See \Cref{lem:KL_upper_bound} for details.
	\item 
	For any mechanism and any $[0, 1]$-bounded value distribution, $O( \epsilon^{-2} )$ samples are sufficient for estimating an $\epsilon$ additive approximation of the expected revenue.
	This part follows directly from standard concentration bounds such as the Bernstein inequality (\Cref{lem:bernstein}).
\end{enumerate}

Then, we claim that the optimal revenue of $\mathbf{D}$ and $\mathbf{\tilde{D}}$ must be within an additive factor of $\epsilon$ from each other.
In particular, we claim that:
\[
\rev(M_{\mathbf{D}}, \mathbf{\tilde{D}}) \ge \opt(\mathbf{D}) - \epsilon
~.
\]
Otherwise, we would be able to distinguish these two distributions with $O( \epsilon^{-2} )$ samples by estimating the expected revenue of $M_{\mathbf{D}}$.
This contradicts the assumption that the KL divergence of the two distributions is small.

\paragraph{Formal Analysis.}

More generally, we show the following lemma.

\begin{lemma}
	\label{lem:difference_ori_shaded}
	If two distributions $\mathbf{D'}$ and $\mathbf{\tilde{D}'}$ satisfy that for some $N > 0$, and some $\alpha > 0$:
	\begin{enumerate}
		\item They are close in KL-divergence:
		\[
		\skl \big( \mathbf{D'}, \mathbf{\tilde{D}'} \big) \le c N^{-1} 
		~,
		\]
		for some sufficiently small constant $c > 0$.
		\item For any mechanism, and any of these two distributions, $N$ samples runs are sufficient to estimate the expected revenue up to an additive $\alpha$ factor with probability at least $\frac{2}{3}$.
	\end{enumerate}
	Then, we have:
	\[
	\opt(\mathbf{\tilde{D'}}) \ge \opt(\mathbf{D}') - 2 \alpha
	\]
\end{lemma}

\ifshowproof
%

\begin{proof}
	We will show a stronger claim that for any mechanism $M$, it holds that:
	\[
	\rev(M, \mathbf{\tilde{D}'}) \ge \rev(M, \mathbf{D'}) - 2 \alpha
	~.
	\]
	
	Suppose not. 
	Consider the following classification algorithm that takes $N$ i.i.d.\ samples from an unknown distribution that is either $\mathbf{D'}$ or $\mathbf{\tilde{D}'}$, and identifies which one it is correctly with probability at least $\frac{2}{3}$.
	
	\begin{enumerate}
		\item Run $M$ on the $N$ samples from the unknown distribution to estimate the expected revenue up to an additive $\alpha$ factor.
		\item Return $\mathbf{D'}$ if the estimate is at least $\rev(M, \mathbf{D'}) - \alpha$;
		return $\mathbf{\tilde{D}'}$ otherwise.
	\end{enumerate}
	
	The correctness of the algorithm follows by condition 2 in the lemma statement and the assumption (for contrary) that $\rev(M, \mathbf{\tilde{D}'}) < \rev(M, \mathbf{D'}) - 2 \alpha$.
	Hence, there exists an algorithm that distinguishes the two distributions using $N$ samples.
	This, however, contradicts \Cref{lem:kl_divergence} and condition 1 in the statement of this lemma, because they together indicate that no algorithm can distinguish $\mathbf{D'}$ and $\mathbf{\tilde{D}'}$ correctly using $N$ samples.
	
	We stress that to get the contradiction it suffices to show the existence of the algorithm. 
	How one can acquire the necessary information, in particular, the value of $\rev(M, \mathbf{D'})$, to implement the algorithm is not important.
	%
	%
	%
	%
\end{proof}
\fi

Intuitively, we would like to let $\mathbf{D'} = \mathbf{D}$, $\mathbf{\tilde{D}'} = \mathbf{\tilde{D}}$, and $\alpha = \epsilon \opt(\mathbf{D})$ (or $\alpha = \epsilon$ in the case of additive approximation) in the above lemma to finish the analysis.
However, the two conditions in \Cref{lem:difference_ori_shaded} need not hold for distributions $\mathbf{D}$ and $\mathbf{\tilde{D}}$ in general.
The first condition may not hold, for example, if some large values with tiny quantiles in $\mathbf{D}$ is not in the support of $\mathbf{\tilde{D}}$ as a result of the double shading by $\doubleshading$.
The KL divergence will be infinity in this case.
The second condition may also fail, when the value distribution $\mathbf{D}$ is unbounded, as in the regular and MHR case.

To circumvent these obstacles, we will construct surrogate versions of $\mathbf{D}$ and $\mathbf{\tilde{D}}$, denoted as $\mathbf{D'}$ and $\mathbf{\tilde{D}'}$ respectively, which do satisfy the two conditions in statement of \Cref{lem:difference_ori_shaded}, and will relate their optimal revenues with those of $\mathbf{D}$ and $\mathbf{\tilde{D}}$ respectively to finish the analysis.

We first present in the next lemma some sufficient conditions under which we can bound the KL divergence of a distribution and its doubly shaded version.

\begin{lemma}
	\label{lem:KL_upper_bound}
	Suppose a distribution $D'$ has a bounded support in $[0, u]$ such that $\ell = 0$ and $u$ are point masses, whose probability masses, denoted as $p_{\ell}$ and $p_u$, are at least $\frac{64 \ln(2mn\delta^{-1})}{m}$. 
	Further, suppose $\tilde{D}' = \doubleshading(D')$ is the doubly shaded version of $D'$. 
	Then, we have:
	\[
	\skl \big( D', \tilde{D}' \big) = O\bigg(\frac{\ln(mn\delta^{-1})}{m} \ln \big( p_{\ell}^{-1} p_u^{-1} \big)\bigg) = \tilde{O} \bigg( \frac{1}{m} \ln \big( p_{\ell}^{-1} p_u^{-1} \big) \bigg )
	~.
	\]
\end{lemma}

\ifshowproof
\begin{proof}
	We will first prove the claim assuming there are no point masses other than $u$ and $\ell$.
	Then, the KL divergence can be written as:
	\begin{align}
	\skl (D', \tilde{D}') & = \left( p_u \ln \frac{p_u}{\doubleshading(p_u)} + \doubleshading(p_u) \ln \frac{\doubleshading(p_u)}{p_u} \right) 
	\label{eqn:KL_upper_bound_1} \\
	& \qquad + \left( p_\ell \ln \frac{p_\ell}{1 - \doubleshading(1 - p_\ell)} + \big( 1 - \doubleshading(1 - p_\ell) \big) \ln \frac{1 - \doubleshading(1 - p_\ell)}{p_\ell} \right)
	\label{eqn:KL_upper_bound_2} \\
	& \qquad + \int_{\ell < v < u} \left( \ln \frac{dD'}{d\tilde{D}'} + \frac{d\tilde{D}'}{dD'} \ln \frac{d\tilde{D}'}{dD'} \right) dD'
	\label{eqn:KL_upper_bound_3} 	
	~.
	\end{align}
	
	Next, we bound each of these three terms separately.
	
	\paragraph{First Term.}
	Consider the first term.
	By our assumption that $p_u \ge \frac{64 \ln (2mn\delta^{-1})}{m}$, and noting that $p_u \le 1 - p_\ell \le 1-\frac{64 \ln (2mn\delta^{-1})}{m}$, we have:
	\begin{align}
	p_u - \doubleshading(p_u) & = \sqrt{\frac{8 p_u(1-p_u) \ln(2mn\delta^{-1})}{m}} + \frac{7 \ln(2mn\delta^{-1})}{m} 
	\notag \\
	& \le \sqrt{\frac{16 p_u(1-p_u) \ln(2mn\delta^{-1})}{m}} 
	\label{eqn:KL_upper_bound_4} \\
	& \le \frac{1}{2} p_u
	\label{eqn:KL_upper_bound_5}
	~.
	\end{align}
	
	Then, the first term can be bounded with the following sequence of inequalities:
	\begin{align*}
	\cref{eqn:KL_upper_bound_1} & = \big( p_u - \doubleshading(p_u) \big) \ln \left( 1 + \frac{p_u - \doubleshading(p_u)}{\doubleshading(p_u)} \right) \\
	& \le \frac{ \big( p_u - \doubleshading(p_u) \big)^2 }{\doubleshading(p_u)} && \text{($\ln(1+x) \le x$)} \\
	& \le \frac{ 2 \big( p_u - \doubleshading(p_u) \big)^2 }{p_u} && \text{(\Cref{eqn:KL_upper_bound_5})} \\
	& \le \frac{32 (1-p_u) \ln(2mn\delta^{-1})}{m} && \text{(\Cref{eqn:KL_upper_bound_4})} \\
	& \le \frac{32 \ln(2mn\delta^{-1})}{m} 
	~. 
	\end{align*}
	
	\paragraph{Second Term.}
	The way that we bound the second term is similar.
	We first establish the following inequality bounding the difference in the mass of $\ell$ in the two distributions:
	\begin{align}
	(1 - p_\ell) - \doubleshading(1 - p_\ell) & = \sqrt{\frac{8 p_\ell(1-p_\ell) \ln(2mn\delta^{-1})}{m}} + \frac{7 \ln(2mn\delta^{-1})}{m} 
	\notag \\
	& \le \sqrt{\frac{16 p_\ell(1-p_\ell) \ln(2mn\delta^{-1})}{m}} 
	\label{eqn:KL_upper_bound_6}
	~,
	\end{align}
	where the inequality follows by $\frac{64 \ln (2mn\delta^{-1})}{m}\le p_\ell \le 1 - p_u \le 1-\frac{64 \ln (2mn\delta^{-1})}{m}$.

	Then, the second term can bounded with the following sequence of inequalities:
	\begin{align*}
	\cref{eqn:KL_upper_bound_2} & = \big( (1 - p_\ell) - \doubleshading(1 - p_\ell) \big) \ln \left( 1 + \frac{(1 - p_\ell) - \doubleshading(1 - p_\ell)}{p_\ell} \right) \\
	& \le \frac{\left( (1 - p_\ell) - \doubleshading(1 - p_\ell) \right)^2}{p_\ell} && \text{($\ln(1+x) \le x$)} \\
	& \le \frac{16 (1-p_\ell) \ln(2mn\delta^{-1})}{m} && \text{(\Cref{eqn:KL_upper_bound_6})} \\
	& \le \frac{16 \ln(2mn\delta^{-1})}{m} 
	~. 
	\end{align*}
	
	\paragraph{Third Term.}
	Consider any $\ell < v < u$, which by our assumption is not a point mass.
	Let $q = q^{D'}(v)$ denote its quantile in $D'$, and let $\doubleshading'$ denote the derivative of $\doubleshading$.
	We have:
	\[
	\frac{d \tilde{D}'}{d D'}(v) = \doubleshading' \big( q \big) 
	~.
	\]

	Hence, we can rewrite the third term as:
	\begin{equation}
	\label{eqn:KL_upper_bound_7}
	\int_{\ell < v < u} \left( \ln \frac{dD'}{d\tilde{D}'} + \frac{d\tilde{D}'}{dD'} \ln \frac{d\tilde{D}'}{dD'} \right) dD' = \int_{p_u}^{1-p_\ell} \big( \doubleshading'(q) - 1 \big) \ln \big( \doubleshading'(q) \big) dq
	~.
	\end{equation}
	
	%
	
	Next, note that:
	%
	%
	%
	\begin{align*}
	\left| \doubleshading' \big( q \big) - 1 \right| & = \sqrt{\frac{2 \ln(2mn\delta^{-1})}{q(1-q)m}} |2q - 1| \\
	& \le \sqrt{\frac{2 \ln(2mn\delta^{-1})}{q(1-q)m}} && \text{($0 \le q \le 1$)} \\
	& \le \frac{1}{2} 
	~.
	&& \text{($q_u \le q \le 1-q_\ell$, and $q_u, q_\ell \ge \tfrac{64 \ln(2mn\delta^{-1})}{m}$)} 
	\end{align*}
	
	Further note that $x \ln (1 + x) \le x^2$ for $x \ge 0$ because $\ln (1 + x) \le x$, and $x \ln (1 + x) < 2x^2$ for $-\frac{1}{2} < x < 0$ because $\ln (1 + x) \ge - \ln (1 - 2x) \ge 2x$. 
	The third term is bounded by:
	\begin{align*}
	\cref{eqn:KL_upper_bound_3} & \le \int_{p_u}^{1-p_\ell} 2 \big( \doubleshading' ( q ) - 1 \big)^2 dq \\
	& \le \int_{p_u}^{1-p_\ell} \frac{16 \ln(2mn\delta^{-1})}{q \big( 1-q \big)m} dq \\
	& = \frac{4 \ln(2mn\delta^{-1})}{m} \int_{p_u}^{1-p_\ell} \left( \frac{1}{q} + \frac{1}{1 - q} \right) dq \\
	& = \frac{4 \ln(2mn\delta^{-1})}{m} \left( \ln \frac{1 - q_\ell}{q_u} + \ln \frac{1 - q_u}{q_\ell} \right) \\
	& \le \frac{4 \ln(2mn\delta^{-1})}{m} \ln \big( p_{\ell}^{-1} p_u^{-1} \big)
	~.
	\end{align*}
	
	\paragraph{Point Masses.}
	Next consider a general distribution $D'$ that potentially has point masses other than $u$ and $\ell = 0$.
	An important observation is that $D'$ and $\tilde{D}'$ have the same set of point masses.
	To see this, note that the point masses are precisely the discontinuous points of the quantile function.
	Further, the quantiles of $\tilde{D}'$ are obtained by applying function $\doubleshading$ to the quantiles of $D'$.
	Finally, by the definition of $\doubleshading$, and the assumptions on $p_u$ and $p_\ell$, $\doubleshading$ is continuous and strictly increasing on quantiles between $q^{D'}(u) = p_u$ and $q^{D'}(\ell^+) = \lim_{v \rightarrow \ell^+} q^{D'}(v) = 1 - p_\ell$.
	Hence, there is a one to one mapping between the discontinuous points of $q^D$ and those of $q^{\tilde{D}'}$.
	
	Suppose $\ell < v < u$ is one such point mass.
	Let $q_v = q^{D'}(v)$ denote its quantile and let $p_v$ denote its probability mass.
	Then, it corresponds to quantiles between $q_v - p_v$ and $q_v$.
	That is, the integration from $q_v-p_v$ to $q_v$ has to be subtracted from \Cref{eqn:KL_upper_bound_7}.
	Instead, the contribution by $v$ to the KL divergence is:
	\begin{align*}
	& p_v \ln \frac{p_v}{\doubleshading(q_v) - \doubleshading(q_v-p_v)} + \big( \doubleshading(q_v) - \doubleshading(q_v-p_v) \big) \ln \frac{\doubleshading(q_v) - \doubleshading(q_v-p_v)}{p_v} \\
	& \qquad = p_v \left( \frac{\doubleshading(q_v) - \doubleshading(q_v-p_v)}{p_v} - 1 \right) \ln \frac{\doubleshading(q_v) - \doubleshading(q_v-p_v)}{p_v}
	~.
	\end{align*}
	
	Note that $(x-1)\ln x$ is convex, and $\doubleshading(q_v) - \doubleshading(q_v-p_v) = \int_{q_v-p_v}^{q_v} \doubleshading'(q) dq$.
	By Jensen's inequality, the above is at most:
	\[
	\int_{q_v-p_v}^{q_v} \big( \doubleshading'(q_v) - 1 \big) \ln \big( \doubleshading'(q_v) \big) dq
	~.
	\]
	
	That is, the contribution by $v$ to the KL divergence is at most the integration over corresponding quantile interval of $v$ that is subtracted from \Cref{eqn:KL_upper_bound_7}.
	Applying this argument to all point masses, where there are at most countably infinitely many, we prove the lemma for general distributions that may have point masses other than $u$ and $\ell$.
\end{proof}
\fi

In light of the conditions in \Cref{lem:KL_upper_bound} under which we can upper bound the KL divergence of a distribution and its doubly shaded version given by function $\doubleshading$, we will construct the surrogate distribution $\mathbf{D'}$ by truncating both the top and the bottom ends of the original distribution $\mathbf{D}$, and let $\mathbf{\tilde{D}'} = \doubleshading({\mathbf{D'}})$ be the other surrogate distribution.

Let us first consider the truncation in the bottom end, which is easier.
Define a function $\truncatebottom_{\epsilon}$ that takes a value distribution, say, $D_i$ for some bidder $1 \le i \le n$, as input, and returns a distribution obtained by truncating the lowest $\epsilon$ fraction of values in $D_i$ down to $0$.
More precisely, the quantile of any value $v$ in $\truncatebottom_{\epsilon}(D_i)$ is defined as follows:
\[
q^{\truncatebottom_{\epsilon}(D_i)}(v) \defeq 
\begin{cases}
\min \{ q^{D_i}(v), 1 - \epsilon \} & \text{if $v > 0$} \\
1 & \text{if $v = 0$}
\end{cases}
\]

Further, for any product value distribution $\mathbf{D}$, define:
\[
\truncatebottom_{\epsilon}(\mathbf{D}) \defeq \truncatebottom_{\epsilon}(D_1) \times \truncatebottom_{\epsilon}(D_2) \times \dots \times \truncatebottom_{\epsilon}(D_n)
~.
\]


The truncated version now has $0$ as the smallest value in its support.
Further, the probability mass of $0$ is at least $\epsilon$, which is good enough for the purpose of using \Cref{lem:KL_upper_bound}.
On the other hand, we want to make sure the optimal revenue after the truncation, namely, that of $\truncatebottom_{\epsilon}(\mathbf{D})$, is close to the optimal revenue of the original distribution $\mathbf{D}$. 
This is established in the next lemma.

\begin{lemma}
	\label{lem:truncate_bottom}
	For any product value distribution $\mathbf{D}$, we have:
	\[
	\opt\big(\truncatebottom_{\epsilon}(\mathbf{D})\big) \ge (1-\epsilon) \opt(\mathbf{D})
	~.
	\]
\end{lemma}

Next, let us turn to the truncation in the top end.
This part is more subtle because it serves two purposes, to satisfy the conditions in \Cref{lem:KL_upper_bound} for bounding the KL divergence of the surrogate distributions, and to satisfy the second condition of \Cref{lem:difference_ori_shaded}, namely, to ensure that one can estimate the expected revenue of any mechanism on the auxiliary distribution with standard concentration bounds.
For the latter, we would intuitively like to truncate values that are too large, in particular, those that are much larger than the optimal revenue of the surrogate distribution $\mathbf{D'}$.
Therefore, given an appropriate vector of value upper bounds $\mathbf{\bar{v}}$, we introduce the following function $\truncatetop_{\mathbf{\bar{v}}}$ for truncating the top end of the value distributions.
For every bidder $i$, let $\truncatetop_{\bar{v}_i}(D_i)$ be the distribution obtained by truncating values larger than $\bar{v}_i$ down to $\bar{v}_i$.
In other words, the quantiles of the truncated distribution is defined as:
\[
q^{\truncatetop_{\bar{v}_i}(D_i)}(v) \defeq 
\begin{cases}
q^{D_i}(v) & \text{if $v_i \le \bar{v}_i$} \\
0 & \text{if $v_i > \bar{v}_i$}
\end{cases}
\]

Further, for any $\mathbf{\bar{v}} = (\bar{v}_1, \bar{v}_2, \dots, \bar{v}_n)$, define:
\[
\truncatetop_{\mathbf{\bar{v}}}(\mathbf{D}) = \truncatetop_{\bar{v}_1}(D_1) \times \truncatetop_{\bar{v}_2}(D_2) \times \dots \times \truncatetop_{\bar{v}_n}(D_n)
~.
\]

Informally, we will choose the value upper bounds $\mathbf{\bar{v}}$ such that 1) for any $i$, $\bar{v}_i$ is upper bounded by the optimal revenue $\opt(\mathbf{D})$ multiplied by a factor that depends on the family of distributions, and 2) for any $i$, $q^{D_i}(\bar{v}_i)$ is at least $\tilde{\Omega}(\frac{1}{m})$.
The first property is to satisfy the conditions in \Cref{lem:difference_ori_shaded}, and the second property is the satisfy the conditions in \Cref{lem:KL_upper_bound}.
%

We summarize the meta construction of the surrogate distributions and some of their properties in the following lemma.
%
%
%

\begin{lemma}
	\label{lem:auxiliary_distributions}
	For any product value distribution $\mathbf{D}$, suppose there exist $\mathbf{\bar{v}}$, $\beta \ge \opt(\mathbf{D})$, and $p > 0$ such that:
	\begin{enumerate}
		\item
		$\beta \ge p \bar{v}_i$ for all $i \in [n]$.
		\item $q^{D_i}(\bar{v}_i)$ is at least $p \epsilon^2n^{-1}$ for all $i \in [n]$.
		\item  $\opt(\truncatetop_{\mathbf{\bar{v}}}(\mathbf{D})) \ge \opt(\mathbf{D}) - \epsilon \beta$.
	\end{enumerate}
	Then, there exist distributions $\mathbf{D'}$ and $\mathbf{\Tilde{D}'} = \doubleshading(\mathbf{D'})$ such that for any\\ $m \ge O\big(n p^{-1}\epsilon^{-2}\ln(mn\delta^{-1})\ln(np^{-1}\epsilon^{-1})\big) $, i.e., $m\ge \tilde{O}(n p^{-1}\epsilon^{-2})$:
	\begin{enumerate}
		\item[a)] $\mathbf{D'}$ and $\mathbf{\tilde{D}'}$ have bounded supports in $[0, p^{-1} \beta]$ coordinate-wise.
		\item[b)] $\opt(\mathbf{D'}) \ge \opt(\mathbf{D}) - 2 \epsilon \beta$.
		\item[c)] $\mathbf{\tilde{D}} \succeq \mathbf{\Tilde{D}'}$.
		\item[d)] $\skl(\mathbf{D}',\mathbf{\Tilde{D}'}) = O(p\epsilon^2)$.
	\end{enumerate}
\end{lemma}

\ifshowproof
\begin{proof}
	Given such a vector $\mathbf{\bar{v}}$, define the surrogate distributions as:
	\[
	\mathbf{D'} = \truncatebottom_{\epsilon} \circ \truncatetop_{\mathbf{\bar{v}}} (\mathbf{D})
	~,
	\]
	and
	\[
	\mathbf{\Tilde{D}'} = \doubleshading(\mathbf{D'})
	~.
	\]
	
	\paragraph{Part a)} 
	This is true by the above definition of $\mathbf{D'}$, the definition of $\mathbf{\tilde{D}'} = \doubleshading(\mathbf{D'})$, and the first condition in this lemma which implies $\bar{v}_i \le p^{-1} \beta$ for all $i \in [n]$.
	
	\paragraph{Part b)}
	Note that $\mathbf{D} \succeq \truncatetop_{\mathbf{\bar{v}}} (\mathbf{D})$.
	By weak revenue monotonicity (\Cref{lem:weak_revenue_monotonicity}) and the third condition of this lemma, we have:
	\[
	\opt(\mathbf{D}) \ge \opt \big( \truncatetop_{\mathbf{\bar{v}}}(\mathbf{D}) \big) \ge \opt(\mathbf{D}) - \epsilon \beta
	~.
	\]
	
	Similarly, note that $\truncatetop_{\mathbf{\bar{v}}} (\mathbf{D}) \succeq \truncatebottom_{\epsilon} \circ \truncatetop_{\mathbf{\bar{v}}} (\mathbf{D})$.
	We have:
	%
	\begin{align*}
	\opt\big( \truncatetop_{\mathbf{\bar{v}}}(\mathbf{D}) \big) 
	& \ge 
	\opt \big( \truncatebottom_{\epsilon} \circ \truncatetop_{\mathbf{\bar{v}}}(\mathbf{D}) \big) && \text{(weak revenue monotonicity, i.e., \Cref{lem:weak_revenue_monotonicity})} \\ 
	& \ge 
	(1 - \epsilon) \opt\big( \truncatetop_{\mathbf{\bar{v}}}(\mathbf{D}) \big) && \text{(\Cref{lem:truncate_bottom})} \\
	& \ge \opt\big( \truncatetop_{\mathbf{\bar{v}}}(\mathbf{D}) \big) - \epsilon \opt(\mathbf{D}) && \text{(weak revenue monotonicity, i.e., \Cref{lem:weak_revenue_monotonicity})} \\
	& \ge
	\opt\big( \truncatetop_{\mathbf{\bar{v}}}(\mathbf{D}) \big) - \epsilon \beta
	~.
	\end{align*}	
	
	Putting together we have:
	\[
	\opt(\mathbf{D'}) 
	\ge 
	\opt \big( \truncatebottom_{\epsilon} \circ \truncatetop_{\mathbf{\bar{v}}}(\mathbf{D}) \big) 
	\ge 
	\opt(\mathbf{D}) - 2 \epsilon \beta
	~.
	\]
	
	Hence, we have proved part b) of this lemma.
	
	\paragraph{Part c)}
	Note that $\mathbf{D} \succeq \mathbf{D'}$ since $\mathbf{D'}$ is obtained by truncating both the top and bottom ends of distribution $\mathbf{D}$.
	Further, $\mathbf{\tilde{D}} = \doubleshading(\mathbf{D})$ and $\mathbf{\tilde{D}'} = \doubleshading(\mathbf{D'})$.
	This part of the lemma now follows because $\doubleshading$ is a monotone function.
	
	\paragraph{Part d)}
	For every bidder $i$, our construction ensures that $D'_i$ and $\tilde{D}'_i = \doubleshading(D'_i)$ satisfy the conditions of \Cref{lem:KL_upper_bound} with $u = \bar{v}_i$, $\ell = 0$, $p_u \ge p \epsilon^2 n^{-1} = \tilde{\Omega}(m^{-1})$ (due to the second condition in this lemma), and $p_\ell \ge \epsilon = \tilde{\Omega}(m^{-1})$ (due to the definition of $\truncatebottom_\epsilon$). 
	Hence, by \Cref{lem:KL_upper_bound} we have:
	\[
	\skl(D'_i, \tilde{D}'_i) = O\left(\frac{\ln(mn\delta^{-1})\ln(n\epsilon^{-1}p^{-1})}{m}\right) = O \left( \frac{p \epsilon^2}{n} \right)
	~.
	\]
	
	Then, this part of the lemma follows because $\skl(\mathbf{D'}, \mathbf{\tilde{D}'}) = \sum_{i = 1}^n \skl(D'_i, \tilde{D}'_i)$.
\end{proof}
\fi
%

As a corollary, we lower bound the optimal revenue of the auxiliary distribution $\mathbf{\tilde{D}}$ by that of the original distribution $\mathbf{D}$ when $m$ is sufficiently large, where the bounds depend on the parameters $\beta$ and $p$ in the conditions of \Cref{lem:auxiliary_distributions}.

\begin{corollary}
	\label{cor:analysis_information_theory}
	Suppose there exist $\mathbf{\bar{v}}$, $\beta$, and $p$ satisfying the conditions in \Cref{lem:auxiliary_distributions}.
	Then, we have the following lower bound on the optimal revenue of the auxiliary distribution $\mathbf{\tilde{D}}$:
	\[
	\opt(\mathbf{\tilde{D}}) \ge \opt(\mathbf{D}) - 4 \epsilon \beta
	~,
	\]
	provided that the number of samples $m$ is at least:
	\[
	O\big(np^{-1}\epsilon^{-2}\ln(np^{-1}\epsilon^{-1}\delta^{-1})\ln(np^{-1}\epsilon^{-1})\big)=\tilde{O} \big(n p^{-1} \epsilon^{-2} \big)
	~.
	\]
\end{corollary}

\ifshowproof
\begin{proof}
	Note that $\mathbf{D'}$ and $\mathbf{\tilde{D}'}$ have supports upper bounded by $p^{-1} \beta$ (part a) of \Cref{lem:auxiliary_distributions}), and the expected revenue of any mechanism on any of these two distributions is at most $\opt(\mathbf{D}) \le \beta$ ($\mathbf{D} \succeq \mathbf{D'}, \mathbf{\tilde{D}'}$ and weak revenue monotonicity by \Cref{lem:weak_revenue_monotonicity}).
	By Bernstein inequality (\Cref{lem:bernstein}), $O(p^{-1} \epsilon^{-2})$ sample runs are sufficient to estimate the expected revenue of any mechanism on any of these two distributions up to an additive factor of $\epsilon \beta$.
	Using \Cref{lem:difference_ori_shaded} with $N = p^{-1} \epsilon^{-2}$, and $\alpha = \epsilon \beta$ we get that:
	\[
	\opt(\mathbf{\tilde{D}'}) \ge \opt(\mathbf{D'}) - 2 \epsilon \beta
	~.
	\]
	
	By part c) of \Cref{lem:auxiliary_distributions}, and weak revenue monotonicity (\Cref{lem:weak_revenue_monotonicity}), we have:
	\[
	\opt(\mathbf{\tilde{D}}) \ge \opt(\mathbf{\tilde{D}'}) 
	~.
	\]
	
	Finally, by part b) of \Cref{lem:auxiliary_distributions}, we have:
	\[
	\opt(\mathbf{D'}) \ge \opt(\mathbf{D}) - 2 \epsilon \beta
	~.
	\]
	
	Putting together these three inequalities proves the corollary.
\end{proof}
\fi

Finally, combining \Cref{lem:analysis_revenue_monotonicity} and \Cref{cor:analysis_information_theory}, we get the following corollaries which finish the meta analysis.
We remark that a direct combination of \Cref{lem:analysis_revenue_monotonicity} and \Cref{cor:analysis_information_theory} gives $1 - 4 \epsilon$ multiplicative approximation or $4 \epsilon$ additive approximation.
However, reducing the approximation parameter by a factor of $4$ increases the number of samples needed by at most a constant factor because the bounds are polynomial in $\epsilon^{-1}$.

\begin{corollary}
	\label{cor:meta_analysis_multiplicative}
	Suppose there exist $\mathbf{\bar{v}}$ and $p$ satisfying the conditions in \Cref{lem:auxiliary_distributions} with $\beta = \opt(\mathbf{D})$.
	Then, with probability at least $1 - \delta$, we have:
	\[
	\rev(M_{\mathbf{\tilde{E}}}, \mathbf{D}) \ge (1 - \epsilon) \opt(\mathbf{D})
	~.
	\]
	provided that the number of samples $m$ is at least:
	\[
	O\big(np^{-1}\epsilon^{-2}\ln(np^{-1}\epsilon^{-1}\delta^{-1})\ln(np^{-1}\epsilon^{-1})\big)=\tilde{O} \big(n p^{-1} \epsilon^{-2} \big)
	~.
	\]
\end{corollary}

\begin{corollary}
	\label{cor:meta_analysis_additive}
	Suppose there exist $\mathbf{\bar{v}}$ and $p$ satisfying the conditions in \Cref{lem:auxiliary_distributions} with $\beta = 1$.
	Then, with probability at least $1 - \delta$, we have:
	\[
	\rev(M_{\mathbf{\tilde{E}}}, \mathbf{D}) \ge \opt(\mathbf{D}) - \epsilon
	~.
	\]
	as long as the number of samples $m$ is at least:
	\[
	O\big(np^{-1}\epsilon^{-2}\ln(np^{-1}\epsilon^{-1}\delta^{-1})\ln(np^{-1}\epsilon^{-1})\big)=\tilde{O} \big(n p^{-1} \epsilon^{-2} \big)
	~.
	\]
\end{corollary}


\subsection{Proof of \texorpdfstring{\Cref{thm:main_upper_bounds}}{Theorem~\ref{thm:main_upper_bounds}}}

Finally, we will prove the sample complexity upper bounds for specific families of distributions as stated in \Cref{thm:main_upper_bounds}.
Given \Cref{cor:meta_analysis_multiplicative} and \Cref{cor:meta_analysis_additive}, it remains to find, for each family of distributions, an appropriate vector of value upper bounds $\mathbf{\bar{v}}$, together with parameters $\beta \ge \opt(\mathbf{D})$ and $p > 0$, that satisfy the conditions in \Cref{lem:auxiliary_distributions}, which we restate below:
\begin{enumerate}
	\item
	$\beta \ge p \bar{v}_i$ for all $i \in [n]$.
	\item $q^{D_i}(\bar{v}_i)$ is at least $p \epsilon^2n^{-1}$ for all $i \in [n]$.
	\item  $\opt(\truncatetop_{\mathbf{\bar{v}}}(\mathbf{D})) \ge \opt(\mathbf{D}) - \epsilon \beta$.
\end{enumerate}


\subsubsection*{$[1, H]$-Bounded Support Distributions}

\label{sec:specific_type_ub}
Let $\beta = \opt(\mathbf{D})$, which is at least $1$ because the values are lower bounded by $1$.
Hence, by letting $p = \frac{1}{H}$, the first condition holds for any choice of $\mathbf{\bar{v}}$, because the values are upper bounded by $H$.
Hence, we only need to choose $\mathbf{\bar{v}}$ to satisfy the other two conditions.
Clearly, the larger $\mathbf{\bar{v}}$ is, the more likely the third condition will hold.
Hence, we will pick $\bar{v}_i$ for each $i \in [n]$ greedily to be the largest value that satisfies the second condition.
That is, define $\mathbf{\bar{v}}$ such that for all $i \in [n]$:
\[
\bar{v}_i = \sup \left\{ v:  q^{D_i}(v) \ge \frac{\epsilon^2}{nH} \right\}
~.
\]

It remains to verify the third condition, which holds due the following sequence of inequalities:
\begin{align*}
\opt(\mathbf{D}) - \opt(\truncatetop_{\mathbf{\bar{v}}}(\mathbf{D})) & \le H \cdot \Pr \big[ \exists i \in [n] : v_i > \bar{v}_i \big] && \text{(values bounded by $H$)} \\[1.5ex]
& \le H \cdot \sum_{i = 1}^n \Pr \big[ v_i > \bar{v}_i \big] && \text{(union bound)} \\
& \le H \cdot \sum_{i = 1}^n \frac{\epsilon^2}{nH} && \text{(definition of $\bar{v}_i$'s)} \\[1.5ex]
& = \epsilon^2 \le \epsilon^2 \beta 
~.
&& \text{($\beta = \opt(\mathbf{D}) \ge 1$)}
\end{align*}

Then, by \Cref{cor:meta_analysis_multiplicative}, we get the sample complexity upper bound of $\tilde{O}(nH\epsilon^{-2})$ as stated in \Cref{thm:main_upper_bounds} for $[1, H]$-bounded support distributions.

\subsubsection*{$[0, 1]$-Bounded Support Distributions}

Since this case considers additive approximation, we will rely on \Cref{cor:meta_analysis_additive} and thus let $\beta = 1$.
Then, the first condition holds trivially with $p = 1$.
Similar to the previous case, we will choose $\bar{v}_i$ for each $i \in [n]$ greedily to be the largest value that satisfies the second condition.
That is, define $\mathbf{\bar{v}}$ such that for all $i \in [n]$:
\[
\bar{v}_i = \sup \left\{ v:  q^{D_i}(v) \ge \frac{\epsilon^2}{n} \right\}
~.
\]

It remains to verify the third condition, which holds due the following sequence of inequalities:
\begin{align*}
\opt(\mathbf{D}) - \opt(\truncatetop_{\mathbf{\bar{v}}}(\mathbf{D})) & \le \Pr \big[ \exists i \in [n] : v_i > \bar{v}_i \big] && \text{(values bounded by $1$)} \\[1.5ex]
& \le \sum_{i = 1}^n \Pr \big[ v_i > \bar{v}_i \big] && \text{(union bound)} \\
& \le \sum_{i = 1}^n \frac{\epsilon^2}{n} && \text{(definition of $\bar{v}_i$'s)} \\[1.5ex]
& = \epsilon^2 
~.
\end{align*}

Then, by \Cref{cor:meta_analysis_additive}, we get the sample complexity upper bound of $\tilde{O}(n\epsilon^{-2})$ as stated in \Cref{thm:main_upper_bounds} for $[0, 1]$-bounded support distributions.

\subsubsection*{Regular Distributions}

Let $\beta = \opt(\mathbf{D})$, and $p = \frac{\epsilon}{8}$.
Consider two vectors of value upper bounds, $\mathbf{\bar{v}^1}$ and $\mathbf{\bar{v}^2}$. 
The former is used to truncate values that are much larger than $\beta = \opt(\mathbf{D})$ to satisfy the first condition.
The latter is used to truncate values with tiny quantiles to satisfy the second condition. 

Concretely, define $\mathbf{\bar{v}^1}$ such that for all $i \in [n]$:
\[
\bar{v}^1_i = p^{-1} \beta = p^{-1} \opt(\mathbf{D})
~.
\]

Define $\mathbf{\bar{v}^2}$ such that for all $i \in [n]$:
\[
\bar{v}^2_i = \sup \left\{ v:  q^{D_i}(v) \ge \frac{p \epsilon^2}{n} \right\}
~.
\]

Finally, let $\mathbf{\bar{v}}$ be the coordinate-wise minimum of $\mathbf{\bar{v}^1}$ and $\mathbf{\bar{v}^2}$, i.e., for all $i \in [n]$:
\[
\bar{v}_i = \min \left\{ \bar{v}^1_i, \bar{v}^2_i \right\}
~.
\]

Clearly, we have $\truncatetop_{\mathbf{\bar{v}}} = \truncatetop_{\mathbf{\bar{v}^2}} \circ \truncatetop_{\mathbf{\bar{v}^1}}$, and the first two conditions hold by our choice of $\mathbf{\bar{v}}$.
It remains to verify the last condition.
We start by bounding the revenue loss due to $\truncatetop_{\mathbf{\bar{v}^2}}$ with an argument similar to those in the previous two cases:
\begin{align*}
\opt \big( \truncatetop_{\mathbf{\bar{v}^1}}(\mathbf{D}) \big) & - \opt \big( \truncatetop_{\mathbf{\bar{v}^2}} \circ \truncatetop_{\mathbf{\bar{v}^1}}(\mathbf{D}) \big) \\[2ex]
& \le p^{-1} \beta \cdot \Pr_{\truncatetop_{\mathbf{\bar{v}^1}}(\mathbf{D})} \big[ \exists i \in [n] : v_i > \bar{v}^2_i \big] && \text{(values bounded by $\bar{v}^1_i = p^{-1} \beta$)} \\[1ex]
& \le p^{-1} \beta \cdot \sum_{i = 1}^n \Pr_{\truncatetop_{\bar{v}^1_i}(D_i)} \big[ v_i > \bar{v}^2_i \big] && \text{(union bound)} \\
& \le p^{-1} \beta \cdot \sum_{i = 1}^n \frac{p \epsilon^2}{n} && \text{(definition of $\bar{v}^2_i$'s)} \\[1.5ex]
& = \epsilon^2 \beta
~.
\end{align*}


Finally, we bound the revenue loss due to $\truncatetop_{\mathbf{\bar{v}^1}}$.
To do that, we need the following lemma by \citet{DevanurHP/2016/STOC} that bounds the tail contribution of regular distributions.

\begin{lemma}[\citet{DevanurHP/2016/STOC}, Lemma 2]
	\label{lem:regular_tail_bound}
	For any product regular distribution $\mathbf{D}$, any $\frac{1}{4} \ge p > 0$, suppose $\mathbf{\bar{v}^1}$ satisfies that $\bar{v}^1_i \ge p^{-1} \opt(\mathbf{D})$ for all $i \in [n]$.
	Then, we have:
	\[
	\opt \big( \truncatetop_{\mathbf{\bar{v}^1}}(\mathbf{D}) \big) \ge (1 - 4p) \opt(\mathbf{D})
	\]
\end{lemma}


By this lemma and our choice of $p = \frac{\epsilon}{8}$ and $\beta = \opt(\mathbf{D})$, we get that:
\[
\opt \big( \truncatetop_{\mathbf{\bar{v}^1}}(\mathbf{D}) \big) \ge \left( 1 - \frac{\epsilon}{2} \right) \opt ( \mathbf{D} ) = \opt(\mathbf{D}) - \frac{\epsilon}{2} \beta
~.
\]

Putting together we have:
\[
\opt \big( \truncatetop_{\mathbf{\bar{v}^2}} \circ \truncatetop_{\mathbf{\bar{v}^1}}(\mathbf{D}) \big) 
\ge 
\opt ( \mathbf{D} ) - \left( \frac{\epsilon}{2} + \epsilon^2 \right) \beta 
\ge 
\opt(\mathbf{D}) - \epsilon \beta
~.
\]

Then, by \Cref{cor:meta_analysis_multiplicative}, we get the sample complexity upper bound of $\tilde{O}(n \epsilon^{-3})$ as stated in \Cref{thm:main_upper_bounds} for the regular distributions.

\subsubsection*{MHR Distributions}

Let $\beta = \opt(\mathbf{D})$, and $p = \frac{1}{c \log(2/\epsilon)}$ for some sufficiently large constant $c > 0$.
Similar to the regular case, we consider two vectors of value upper bounds to satisfy the first and the second conditions respectively. 
Concretely, define $\mathbf{\bar{v}^1}$ such that for all $i \in [n]$:
\[
\bar{v}^1_i = p^{-1} \beta = p^{-1} \opt(\mathbf{D})
~.
\]

Define $\mathbf{\bar{v}^2}$ such that for all $i \in [n]$:
\[
\bar{v}^2_i = \sup \left\{ v:  q^{D_i}(v) \ge \frac{p \epsilon^2}{n} \right\}
~.
\]

Finally, let $\mathbf{\bar{v}}$ be the coordinate-wise minimum of $\mathbf{\bar{v}^1}$ and $\mathbf{\bar{v}^2}$, i.e., for all $i \in [n]$:
\[
\bar{v}_i = \min \left\{ \bar{v}^1_i, \bar{v}^2_i \right\}
~.
\]

Again, we have $\truncatetop_{\mathbf{\bar{v}}} = \truncatetop_{\mathbf{\bar{v}^2}} \circ \truncatetop_{\mathbf{\bar{v}^1}}$, and the first two conditions hold by our choice of $\mathbf{\bar{v}}$.
It remains to verify the last condition.
Bounding the revenue loss due to $\truncatetop_{\mathbf{\bar{v}^2}}$ is verbatim to the regular case; we include it as follows for completeness.
\begin{align*}
\opt \big( \truncatetop_{\mathbf{\bar{v}^1}}(\mathbf{D}) \big) & - \opt \big( \truncatetop_{\mathbf{\bar{v}^2}} \circ \truncatetop_{\mathbf{\bar{v}^1}}(\mathbf{D}) \big) \\[2ex]
& \le p^{-1} \beta \cdot \Pr_{\mathbf{v} \sim \truncatetop_{\mathbf{\bar{v}^1}}(\mathbf{D})} \big[ \exists i \in [n] : v_i > \bar{v}^2_i \big] && \text{(value bounded by $\bar{v}^1_i = p^{-1} \beta$)} \\[1ex]
& \le p^{-1} \beta \cdot \sum_{i = 1}^n \Pr_{v_i \sim \truncatetop_{\bar{v}^1_i}(D_i)} \big[ v_i > \bar{v}^2_i \big] && \text{(union bound)} \\
& \le p^{-1} \beta \cdot \sum_{i = 1}^n \frac{p \epsilon^2}{n} && \text{(definition of $\bar{v}^2_i$'s)} \\[1.5ex]
& = \epsilon^2 \beta
~.
\end{align*}

Finally, we bound the revenue loss due to $\truncatetop_{\mathbf{\bar{v}^1}}$.
To do that, we need the extreme value theorem by \citet{CaiD/2011/FOCS}.
It was originally proved only for continuous MHR distributions~\cite{CaiD/2011/FOCS} but in fact holds for discrete MHR distributions as well~\cite{Cai/2018/communication}.
Below we restate an interpretation of the extreme value theorem by \citet{DevanurHP/2016/STOC} and \citet{MorgensternR/2015/NIPS}, which is most convenient for our analysis.

\begin{lemma}[\citet{CaiD/2011/FOCS}]
	\label{lem:MHR_extreme_value}
	For any product MHR distribution $\mathbf{D}$, and any $\frac{1}{4}\ge \epsilon\ge 0$, suppose $\mathbf{\bar{v}^1}$ satisfies that $\bar{v}^1_i\ge c \log (\frac{1}{\epsilon}) \opt (\mathbf{D})$ for all $i \in [n]$ for a sufficiently large constant $c$. 
	Then, we have:
	\[
	\opt \big( \truncatetop_{\mathbf{\bar{v}^1}}(\mathbf{D}) \big) \ge (1-\epsilon)\opt(\mathbf{D})
	~.
	\]
\end{lemma}

By this lemma and our choice of $p = \frac{1}{c \log(2/\epsilon)}$ and $\beta = \opt(\mathbf{D})$, we get that:
\[
\opt \big( \truncatetop_{\mathbf{\bar{v}^1}}(\mathbf{D}) \big) \ge \left( 1 - \frac{\epsilon}{2} \right) \opt(\mathbf{D}) = \opt(\mathbf{D}) - \frac{\epsilon}{2} \beta
~.
\]

Putting together we have:
\[
\opt \big( \truncatetop_{\mathbf{\bar{v}^2}} \circ \truncatetop_{\mathbf{\bar{v}^1}}(\mathbf{D}) \big) 
\ge 
\opt(\mathbf{D}) - \left( \frac{\epsilon}{2} + \epsilon^2 \right) \beta 
\ge 
\opt(\mathbf{D}) - \epsilon \beta
~.
\]

Then, by \Cref{cor:meta_analysis_multiplicative}, we get the sample complexity upper bound of $\tilde{O}(n \epsilon^{-2})$ as stated in \Cref{thm:main_upper_bounds} for the MHR distributions.

	\newpage
	
	\section{Lower Bounds}
\label{sec:lower_bound_multi_bidder}

In this section, we will prove the following sample complexity lower bounds for all four families of distributions considered in this paper.
Each of them matches the corresponding upper bound in \Cref{sec:upper_bound_multi_bidder} up to a poly-logarithmic factor.

\begin{theorem}
	\label{thm:lower_bounds}
	Suppose an algorithm, given $m$ samples, returns a mechanism that is a $1 - \epsilon$ approximation, with probability at least $0.99$, for a given family of distribution.
	Then, we have:
	\begin{enumerate}
		\item $m$ is at least $\Omega(n \epsilon^{-3})$ if it is the family of regular distributions; or
		\item $m$ is at least $\Omega(n\epsilon^{-2}\ln^{-2}n)=\tilde{\Omega}(n \epsilon^{-2})$ if it is the family of MHR distributions; or
		\item $m$ is at least $\Omega(n H \epsilon^{-2})$ if the family is $[1, H]$-bounded distributions.
	\end{enumerate}
	Suppose it is an $\epsilon$ additive approximation with probability at least $0.99$.
	Then, we have:
	\begin{enumerate}
		\setcounter{enumi}{3}
		\item $m$ is at least $\Omega(n \epsilon^{-2})$ if it is the family of $[0, 1]$-bounded distributions.
	\end{enumerate}
\end{theorem}

We will present the meta construction of hard instances in \Cref{sec:meta_hard_instance}, and the corresponding meta analysis in \Cref{sec:meta_hardness_analysis}.
Finally, we will explain in \Cref{sec:meta_hardness_instantiation} how to use them to prove the sample complexity lower bounds stated in \Cref{thm:lower_bounds}.

\subsection{Meta Hard Instance}
\label{sec:meta_hard_instance}

Our construction of hard instances relies on finding three distributions in the family:
\begin{enumerate}
	\item $D^b$, a base distribution that is a point mass. The bidder with this distribution will serve as the default winner unless another bidder realizes an extremely high value.
	\item $D^h$, a distribution that has a relatively higher chance to win over $D^b$.
	\item $D^\ell$, a distribution that has a relatively lower chance to win over $D^b$.
\end{enumerate}

Given these distributions, consider the following family of hard instances.
Let bidder $1$'s value follows the base distribution $D^b$.
For any other bidder $2 \le i \le n$, let her value distribution be either $D^h$ or $D^\ell$.
Formally, let:
\[
\mathcal{H} = \big\{ \mathbf{D} : D_1 = D^b \text{, and } D_i = D^h \text{ or } D^\ell \text{ for all } 2 \le i \le n \}
~.
\]
Our plan is to show that any algorithm that gets a good enough approximation on all distributions in $\mathcal{H}$ must take a lot of samples.

What properties do we need from these three distributions $D^b$, $D^h$, and $D^\ell$ in order to show a sample complexity lower bound?
Let $\phi^b, \phi^h, \phi^\ell$ denote the corresponding virtual value functions.
For some positive $v_0, v_1, v_2$, where $v_1$ may be $+\infty$, and parameters $0 < p \le \frac{1}{n}$ and $\Delta > 0$, the three distributions shall satisfy the following conditions:
\begin{enumerate}[label=\alph*)]
	\item
	$D^b$ is a point mass at $v_0$.
	\label{property:base_pointmass}
	%
	%
	\item
	The probability of $v \ge v_2$ is at most $1/n$ for both $D^h$ and $D^\ell$.	
	\label{property:prob_high_values}
	\item
	The probability of $v_1 > v \ge v_2$ is at least $p$ for both $D^h$ and $D^\ell$.
	\label{property:prob_critical_interval}
	\item
	For any value $v_1 > v \ge v_2$, we have $\phi^\ell(v) + \Delta \le v_0 \le \phi^h(v) - \Delta$.
	\label{property:virtual_value_gap}
	\item
	For any value $v < v_2$, we have $\phi^h(v), \phi^\ell(v) \le v_0$.
	Our instances in fact satisfy that $D^h$ and $D^\ell$ are identical for values $v < v_2$, but this is not necessary in the meta analysis below. 
	\label{property:virtual_value_bound_small_values}
\end{enumerate}

The above conditions are essential for our construction of the meta hard instance.
The following three, on the other hand, are for the convenience of our argument.

\begin{enumerate}[label=\alph*)]
	\setcounter{enumi}{5}
	\item
	For any value $v_1 > v \ge v_2$, we have $\sqrt{2} \ge \frac{dD^\ell}{dD^h}(v) \ge \frac{1}{\sqrt{2}}$. Here, the factor $\sqrt{2}$ can be replaced by any other constants. Our instances will in fact satisfy this condition up to $1 \pm \epsilon$.
	\label{property:density_gap}
	\item
	$D^h$ is regular.
	\label{property:regular_high_distribution}
	\item
	Either $v_1 = +\infty$, or $v_1$ is a point mass and an upper bound of values in both $D^h$ and $D^\ell$.
	\label{property:v1_point_mass}
\end{enumerate}

Intuitively, to construct a mechanism that gives an approximately optimal expected revenue w.r.t.\ an unknown product value distribution in $\mathcal{H}$, the algorithm must be able to distinguish bidders with value distribution $D^h$, and those with value distribution $D^\ell$, from the samples. 
Otherwise, when there was exactly a bidder with a value between $v_1$ and $v_2$, the algorithm could not correctly decide whether to pick her to be the winner over the default winner, namely, bidder $1$.
We formalize this intuition with the following lemma and its proof in the next subsection.

%

\begin{lemma}
	\label{lem:meta_hardness}
	If an algorithm takes $m$ samples from an arbitrary product value distribution $\mathbf{D} \in \mathcal{H}$ and returns, with probability at least $0.99$, a mechanism whose expected revenue is at least:
	\[
	\opt(\mathbf{D}) - O( n p \Delta) 
	~.
	\]
	Then, the number of samples $m$ is at least:
	\[
	\Omega \left(\skl(D^h, D^\ell)^{-1} \right) 
	~.
	\]
\end{lemma}

\subsection{Meta Analysis: Proof of \texorpdfstring{\Cref{lem:meta_hardness}}{Lemma~\ref{lem:meta_hardness}}}
\label{sec:meta_hardness_analysis}

For some sufficiently small constant $c$, suppose for contrary that the algorithm, denoted as $A$, takes $m <c \cdot \skl(D^h, D^\ell)^{-1}$ samples.
We will account for the revenue loss due to the mistakes made by the mechanism chosen by algorithm $A$ on a bidder by bidder basis.
Concretely, for every bidder $2 \le i \le n$, define a subset of value vector $\mathcal{V}_i$ as follows:
\[
\mathcal{V}_i= \bigg\{ \mathbf{b} = (b_1, b_2, \dots, b_n) : \text{$b_1=v_0$, $v_1 > b_i \ge v_2$, and $b_j < v_2$ for all $j \ne 1, i$} \bigg\}
~.  
\]
Note that $\mathcal{V}_2, \mathcal{V}_3, \dots, \mathcal{V}_n$ are disjoint.
We will account for the revenue loss due to the mistakes made on the value vectors in each subset $\mathcal{V}_i$ separately.

\paragraph{Proof Sketch.}
The plan is to prove the lemma by showing there is a distribution $\mathbf{D} \in \mathcal{H}$ such that for at least $\Theta(n)$ different bidders $2 \le i \le n$, the mechanism chosen by the algorithm with $m <c \cdot \skl(D^h, D^\ell)^{-1}$ samples will make a lot of mistakes on the value vectors in $\mathcal{V}_i$ and, as a result, will have a revenue loss of $\Omega(p \Delta)$.
This will be formalized in \Cref{cor:meta_hardness_bad_instance} and \Cref{lem:meta_hardness_bad_instance_count_const_prob}.

How do we prove that? 
We will do so by showing that for any $2 \le i \le n$, and for a randomly chosen pair of distributions in $\mathcal{H}$ that differ only in the $i$-th coordinate, the algorithm $A$ must have an $\Omega(p \Delta)$ revenue loss due to value vectors in $\mathcal{V}_i$ on at least one of the two distributions.
This will be formally proved in \Cref{lem:meta_hardness_revenue_loss_probability}.
Then, the aforementioned claim follows by a simple counting argument, via \Cref{lem:meta_hardness_counting}.

Intuitively, this is because the algorithm $A$ cannot distinguish such a pair of distributions with so few samples but the winner must be selected differently for the two distributions for value vectors in $\mathcal{V}_i$.
This intuition is formalized with a sequence of claims in \Cref{lem:meta_hardness_similar_decisions}, \Cref{cor:meta_hardness_similar_decisions}, and
\Cref{lem:meta_hardness_mistakes}.

\paragraph{Formal Proof.}
Next, we instantiate the above proof sketch with a formal argument.
Fix any $2 \le i\le n$, and any $\mathbf{D}_{-i}= \times_{j \neq i} D_j$ such that $D_1=D^b$ and $D_j \in \{D^h,D^{\ell}\}$ for all $j \ne 1, i$.
Let $\mathbf{D}^1 = (\mathbf{D}_{-i}, D_i = D^h) \in \mathcal{H}$ and $\mathbf{D}^2 = (\mathbf{D}_{-i}, D_i = D^\ell) \in \mathcal{H}$ be a pair of distributions that differ only in the $i$-th coordinate. 
Then, we have:
\[
\skl \big( \mathbf{D}^1, \mathbf{D}^2 \big) = \skl(D^h, D^\ell)
~.
\]
Then, since algorithm $A$ takes $m < c \cdot \skl(D^h, D^\ell)^{-1} = c \cdot \skl(\mathbf{D}^1, \mathbf{D}^2)^{-1}$ samples for some sufficiently small constant $c$, by \Cref{lem:kl_divergence}, it cannot distinguish whether the underlying distribution is $\mathbf{D}^1$ or $\mathbf{D}^2$ correctly, and as a result will choose a mechanism from essentially the same distribution in both cases.

On the other hand, the optimal auctions w.r.t.\ $\mathbf{D}^1$ and $\mathbf{D}^2$ pick different bidders as the winner for value vectors in $\mathcal{V}_i$: 
the one w.r.t.\ $\mathbf{D}^1$ allocates the item to bidder $i$, while the one w.r.t.\ $\mathbf{D}^2$ allocates the item to the default winner, i.e., bidder $1$.
To instantiate this intuition, we first formally show that the subsets of mechanisms that are close to optimal for $\mathbf{D}^1$ and $\mathbf{D}^2$ respectively, in terms of their choices of winners when the value vector is in $\mathcal{V}_i$, are disjoint. 
We start with the following technical lemma.

\begin{lemma}
	\label{lem:meta_hardness_similar_decisions}
	For any mechanism $M$, the probability that $M$ picks bidder $i$ as the winner, conditioned on the value vector $\mathbf{v}$ is in $\mathcal{V}_i$, differs by at most a factor of $2$ whether $\mathbf{v}$ is drawn from $\mathbf{D}^1$ or $\mathbf{D}^2$.
\end{lemma}

In the proof of this lemma the the rest of the subsection, we will use $\Pr_{\mathbf{v} \sim \mathbf{D} : \mathbf{v} \in \mathcal{V}}$ to denote the conditional probability when $\mathbf{v}$ is drawn from $\mathbf{D}$ conditioned on $\mathbf{v} \in \mathcal{V}$ for some subset $\mathcal{V}$ of value vectors.
Similarly, we will use $\E_{\mathbf{v} \sim \mathbf{D} : \mathbf{v} \in \mathcal{V}}$ to denote the conditional expectation.

\begin{proof}
	We have:
	\begin{align*}
	& \Pr_{\mathbf{v} \sim \mathbf{D}^1 : \mathbf{v} \in \mathcal{V}_i} \big[ \text{$M$ picks $i$ as the winnner} \big] \cdot \Pr_{\mathbf{v} \sim \mathbf{D}^1} \big[ \mathbf{v} \in \mathcal{V}_i \big] \\[1ex]
	& \qquad\qquad = \int_{\mathbf{v} \in \mathcal{V}_i} \mathbf{1}(\text{$M$ picks $i$ as the winnner}) ~ d \mathbf{D}^1 \\
	& \qquad\qquad = \int_{\mathbf{v} \in \mathcal{V}_i} \mathbf{1}(\text{$M$ picks $i$ as the winnner}) \frac{d \mathbf{D}^1}{d \mathbf{D}^2}(\mathbf{v}) ~ d \mathbf{D}^2 \\
	& \qquad\qquad = \int_{\mathbf{v} \in \mathcal{V}_i} \mathbf{1}(\text{$M$ picks $i$ as the winnner}) \frac{d D^h}{d D^\ell}(v_i) ~ d \mathbf{D}^2
	~.
	\end{align*}

	Here, $\mathbf{1}(\cdot)$ is the indicator function that equals $1$ if the event is true. 
	The last equality is due to the definition $\mathbf{D}^1$ and $\mathbf{D}^2$, which differ only in the $i$-th coordinate.
	
	Similarly, we also have:
	\begin{align*}
	& \Pr_{\mathbf{v} \sim \mathbf{D}^2 : \mathbf{v} \in \mathcal{V}_i} \big[ \text{$M$ picks $i$ as the winner} \big] \cdot \Pr_{\mathbf{v} \sim \mathbf{D}^2} \big[ \mathbf{v} \in \mathcal{V}_i \big] \\
	& \qquad\qquad = \int_{\mathbf{v} \in \mathcal{V}_i} \mathbf{1}(\text{$M$ picks $i$ as the winner}) ~ d \mathbf{D}^2
	~.
	\end{align*}
	
	Noting that $v_1 > v_i > v_2$ for any $\mathbf{v} \in \mathcal{V}_i$, the claim now follows by condition \ref{property:density_gap}, which indicates that $\sqrt{2} \ge \frac{d \mathbf{D}^1}{d \mathbf{D}^2}(\mathbf{v}) \ge \frac{1}{\sqrt{2}}$ and that the probability of $\mathbf{v} \in \mathcal{V}_i$ differs by at most a factor of $\sqrt{2}$.
\end{proof}

Next, we define the following partitions of mechanisms:
\begin{align*}
\mathcal{M}^1 & = \left\{ M : \Pr_{\mathbf{v} \sim \mathbf{D}^1 : \mathbf{v} \in \mathcal{V}_i} \big[ \text{$M$ picks $i$ as the winnner} \big] \ge \frac{2}{3} \right\} 
~, \\
\mathcal{M}^2 & = \left\{ M : \Pr_{\mathbf{v} \sim \mathbf{D}^1 : \mathbf{v} \in \mathcal{V}_i} \big[ \text{$M$ picks $i$ as the winnner} \big] < \frac{2}{3} \right\}
~.
\end{align*}

We have the following as a corollary of \Cref{lem:meta_hardness_similar_decisions}.

\begin{corollary}
	\label{cor:meta_hardness_similar_decisions}
	For any $M \in \mathcal{M}^1$, we have that:
	\[
	\Pr_{\mathbf{v} \sim \mathbf{D}^2 : \mathbf{v} \in \mathcal{V}_i} \big[ \text{\rm $M$ picks $i$ as the winnner} \big] \ge \frac{1}{3}
	~.
	\]
\end{corollary}

For any value vector $\mathbf{v} \in \mathcal{V}_i$, recall that the optimal auction w.r.t.\ $\mathbf{D}^1$ will allocate the item to bidder $i$ while the one w.r.t.\ $\mathbf{D}^2$ will allocate the item to the default winner, i.e., bidder $1$.
Therefore, if the underlying distribution is $\mathbf{D}^1$ and the value vector is in $\mathcal{V}_i$, the mechanisms in $\mathcal{M}^2$ will pick a wrong winner with probability at least $\frac{1}{3}$ by definition.
Similarly, if the underlying distribution is $\mathbf{D}^2$ and the value vector is in $\mathcal{V}_i$, the mechanisms in $\mathcal{M}^1$ will pick a wrong winner with probability at least $\frac{1}{3}$ by \Cref{cor:meta_hardness_similar_decisions}.
Informally, the algorithm shall return a mechanism in $\mathcal{M}^j$ most of the time if the underlying distribution is $\mathbf{D}^j$ for $j \in \{1, 2\}$, in order to ensure that the expected revenue is close to optimal.


Next, we formalize the intuition that $A$ is taking too few samples to make different decisions on $\mathbf{D}^1$ and $\mathbf{D}^2$ with the following lemma.

\begin{lemma}
	\label{lem:meta_hardness_mistakes}
	For either $j = 1$ or $j = 2$ (or both), we have:
	\[
	\Pr \big[ A(\mathbf{D}^j) \in \mathcal{M}^{3-j} \big] > \frac{1}{3}
	~.
	\]
\end{lemma}

\begin{proof}
	Consider the following algorithm for distinguishing the two distributions $\mathbf{D}^1$ and $\mathbf{D}^2$. 
	Given an unknown distribution $\mathbf{D} \in \{ \mathbf{D}^1, \mathbf{D}^2 \}$, run algorithm $A$ with $m$ samples from $\mathbf{D}$.
	If the mechanism returned by $A$, i.e., $A(\mathbf{D})$, is in $\mathcal{M}^1$, return $\mathbf{D}^1$; otherwise, return $\mathbf{D}^2$.
	By our assumption for contrary that $A$ takes less than $c \cdot \skl(\mathbf{D}^1, \mathbf{D}^2)^{-1}$ samples, it cannot distinguish the two distributions correctly (\Cref{lem:kl_divergence}).
	That is, we have either:
	\[
	\Pr \big[ A(\mathbf{D^1}) \in \mathcal{M}^1 \big] < \frac{2}{3} ~,
	\]
	or:
	\[
	\Pr \big[ A(\mathbf{D^2}) \in \mathcal{M}^2 \big] < \frac{2}{3} ~,
	\]
	or both.
	The lemma now follows by that $\mathcal{M}^1$ and $\mathcal{M}^2$ form a partition of the mechanism space.
\end{proof}

Next, we will account for the revenue loss when the algorithm $A$ makes a mistake in the sense that it chooses a mechanism in $\mathcal{M}^{3-j}$ when the underlying distribution is $\mathbf{D}^j$ as in the statement of the previous lemma.
We first need to show a technical lemma.

\begin{lemma}
	\label{lem:meta_hardness_virtual_value_maximizer}
	For any value distribution $\mathbf{D} \in \mathcal{H}$, the optimal mechanism w.r.t.\ $\mathbf{D}$ always chooses the bidder with highest virtual value as winner.
\end{lemma}

\begin{proof}
	Recall that the optimal auction picks the bidder with the highest non-negative ironed virtual value.
	By our construction, bidder $1$'s value is a point mass and always has virtual value equals her value $v_0$.
	Hence, non-negativity holds trivially.
	It remains to show that the highest ironed virtual value coincides with the highest virtual value. 
	
	Suppose $v_1 < +\infty$ and there is at least one bidder $2 \le i \le n$ with value equals $v_1$.
	Then, her virtual value is also $v_1$ due to condition \ref{property:v1_point_mass}.
	As a result, the highest virtual value, ironed or not, equals $\max \{v_0, v_1\}$, because the ironed virtual value of any bidder $2 \le j \ne i \le n$ cannot exceed her value, which is upper bounded by $v_1$.
	
	Next, suppose no bidder $2 \le i \le n$ has a value equals $v_1$.
	In this case, any bidder whose value distribution is $D^\ell$ cannot have an ironed virtual value higher than that of bidder $1$.
	To see this, first note that condition \ref{property:v1_point_mass} indicates $v_1$ is not in any ironed interval.
	Further, by conditions \ref{property:virtual_value_gap} and \ref{property:virtual_value_bound_small_values}, the virtual value in $D^\ell$ is at most $v_0$ for any value other than $v_1$.
	The ironed value is simply the average over the corresponding ironed interval, and therefore cannot be larger than $v_0$.
	 
	It remains to consider bidders whose value distribution is $D^h$.
	The lemma now follows by condition \ref{property:regular_high_distribution}, which states that $D^h$ is regular.
\end{proof}

In the following discussions, let $\phi_{A(\mathbf{D}^j)}(\mathbf{v})$ denote the virtual value of the winner chosen by $A(\mathbf{D}^j)$ when the value vector is $\mathbf{v}$.
By the connection between expected revenue and virtual values showed by \citet{Myerson/1981/MOR}, and \Cref{lem:meta_hardness_virtual_value_maximizer}, we have:
\begin{equation}
\label{eqn:optimal_accounting}
\opt(\mathbf{D}) 
= 
\E_{\mathbf{v} \sim \mathbf{D}} \left[ \max_{k \in [n]} \phi_k(v_k) \right] 
= 
\int_{\mathbf{v}} \max_{k \in [n]} \phi_k(v_k) d \mathbf{D} 
~,
\end{equation}
and
\begin{equation}
\label{eqn:algorithm_accounting}
\rev(A(\mathbf{D}), \mathbf{D}) 
= 
\E_{\mathbf{v} \sim \mathbf{D}} \left[ \phi_{A(\mathbf{D})}(\mathbf{v}) \right] 
= 
\int_{\mathbf{v}} \phi_{A(\mathbf{D})}(\mathbf{v}) d \mathbf{D}  ~.
\end{equation}

To account for the revenue loss due to value vectors in $\mathcal{V}_i$, we will consider the following quantity:
\[
\int_{\mathbf{v} \in \mathcal{V}_i} \left( \max_{k \in [n]} \phi_k(v_k) - \phi_{A(\mathbf{D})}(\mathbf{v}) \right) d \mathbf{D} 
~.
\]

We will prove that the above quantity is at least $\Omega(p \Delta)$ for either $\mathbf{D} = \mathbf{D}^1$, or $\mathbf{D} = \mathbf{D}^2$, or both, with the following lemmas.

\begin{lemma}
	\label{lem:meta_hardness_revenue_loss_probability}
	For either $j = 1$ or $j = 2$ (or both), we have:
	\[
	\Pr_{A(\mathbf{D}^j)} \left[ \E_{\mathbf{v} \sim \mathbf{D}^j : \mathbf{v} \in \mathcal{V}_i} \left[ \max_{k \in [n]} \phi_k(v_k) -  \phi_{A(\mathbf{D}^j)}(\mathbf{v}) \right] \ge \frac{\Delta}{3} \right] \ge \frac{1}{3}
	~.
	\]
\end{lemma}

\begin{proof}
	Let  $j \in \{1, 2\}$ be the superscript for which the conclusion of \Cref{lem:meta_hardness_mistakes} holds.
	Then, it suffices to show that:
	\[
	A(\mathbf{D}^j) \in \mathcal{M}^{3-j}
	\quad \Rightarrow \quad
	\E_{\mathbf{v} \sim \mathbf{D}^j : \mathbf{v} \in \mathcal{V}_i} \left[ \max_{k \in [n]} \phi_k(v_k) -  \phi_{A(\mathbf{D}^j)}(\mathbf{v}) \right] \ge \frac{\Delta}{3} 
	~.
	\]
	
	\paragraph{Case 1: $j = 1$.}
	By $A(\mathbf{D}^1) \in \mathcal{M}^{2}$, we have that:
	\[
	\Pr_{\mathbf{v} \sim \mathbf{D}^1 : \mathbf{v} \in \mathcal{V}_i} \big[ \text{$M$ picks $i$ as the winnner} \big] < \frac{2}{3}
	~.
	\]
	
	By conditions \ref{property:virtual_value_gap} and \ref{property:virtual_value_bound_small_values}, the definition of $\mathcal{V}_i$, and that $D_i = D^h$ in $\mathbf{D}^1$, whenever the mechanism picks anyone other than bidder $i$ as the winner when the value vector $\mathbf{v}$ is in $\mathcal{V}_i$, we have:
	\[
	\max_{k \in [n]} \phi_k(v_k) -  \phi_{A(\mathbf{D}^1)}(\mathbf{v}) \ge \Delta
	~.
	\]
	
	So the claim follows.
	
	\paragraph{Case 2: $j = 2$.}
	By $A(\mathbf{D}^2) \in \mathcal{M}^{1}$, and \Cref{cor:meta_hardness_similar_decisions}, we have that:
	\[
	\Pr_{\mathbf{v} \sim \mathbf{D}^2 : \mathbf{v} \in \mathcal{V}_i} \big[ \text{$M$ picks $i$ as the winnner} \big] \ge \frac{1}{3}
	~.
	\]
	
	By condition \ref{property:virtual_value_gap}, the definition of $\mathcal{V}_i$, and that $D_i = D^\ell$ in $\mathbf{D}^2$, whenever the mechanism picks bidder $i$, instead of bidder $1$, as the winner when the value vector $\mathbf{v}$ is in $\mathcal{V}_i$, we have:
	\[
	\max_{k \in [n]} \phi_k(v_k) -  \phi_{A(\mathbf{D}^2)}(\mathbf{v}) \ge \Delta
	~.
	\]
	
	So the claim follows.
\end{proof}

Let $\mathcal{B}_{\mathbf{D}}$ denote the set of bidders for which algorithm $A$ performs badly in the sense that the mechanism returned by $A$ suffers from a revenue loss of at least $\frac{\Delta}{3}$ conditioned on $\mathcal{V}_i$, with probability at least $\frac{1}{3}$, as stated in \Cref{lem:meta_hardness_revenue_loss_probability}:

\[
\mathcal{B}_{\mathbf{D}} 
= 
\left\{ 
i : \Pr_{A(\mathbf{D})} \left[ \E_{\mathbf{v} \sim \mathbf{D} : \mathbf{v} \in \mathcal{V}_i} \left[ \max_{k \in [n]} \phi_k(v_k) -  \phi_{A(\mathbf{D})}(\mathbf{v}) \right] \ge \frac{\Delta}{3} \right] \ge \frac{1}{3}
\right\}
~.
\]

\begin{lemma}
	\label{lem:meta_hardness_counting}
	Suppose a distribution $\mathbf{D}$ is drawn uniformly at random from $\mathcal{H}$.
	Then, for any $2 \le i \le n$, we have:
	\[
	\Pr \big[ i \in \mathcal{B}_{\mathbf{D}} \big] \ge \frac{1}{2}
	~.
	\]
\end{lemma}

\begin{proof}
	Fix any $2 \le i \le n$.
	Enumerating over all possible $\mathbf{D}_{-i}$, $\mathbf{D}^1$ and $\mathbf{D}^2$ together enumerate over all distributions $\mathbf{D} \in \mathcal{H}$.
	Note that \Cref{lem:meta_hardness_revenue_loss_probability} holds for any $\mathbf{D}_{-i}$.
	We get that at least half of the distributions $\mathbf{D} \in \mathcal{H}$ satisfy that $i \in \mathcal{B}_{\mathbf{D}}$.	
\end{proof}

As a direct corollary, we have the following.

\begin{corollary}
	\label{cor:meta_hardness_bad_instance}
	There exists $\mathbf{D} \in \mathcal{H}$ such that:
	\[
	\big| \mathcal{B}_{\mathbf{D}} \big| \ge \frac{n-1}{2}
	~.
	\]
\end{corollary}

In the rest of the analysis, we will focus on the distribution $\mathbf{D} \in \mathcal{H}$ for which the conclusion of the above corollary holds.
The above corollary is already good enough for proving a weaker claim that the \emph{expected} revenue loss is at least $\Theta(n p \Delta)$, noting that the probability of having a value vector in $\mathcal{V}_i$ for each $i \in \mathcal{B}_\mathbf{D}$ is $\Theta(p)$ due to conditions c) and d).

To get the stronger claim in \Cref{lem:meta_hardness} that we have the stated revenue loss with a (small) constant probability, we need to further discuss the the number of bidders for which the \emph{realized mechanism} $A(\mathbf{D})$ performs poorly.
For any realization of the mechanism $A(\mathbf{D})$ returned by the algorithm, further let $\mathcal{B}_{\mathcal{D}, A(\mathbf{D})}$ denote the set of bidders for which the returned mechanism $A(\mathbf{D})$ performs poorly in the sense that it suffers from a revenue loss at least $\frac{\Delta}{3}$ on $\mathcal{V}_i$:

\[
\mathcal{B}_{\mathbf{D}, A(\mathbf{D})} 
= 
\left\{ 
i : \E_{\mathbf{v} \sim \mathbf{D} : \mathbf{v} \in \mathcal{V}_i} \left[ \max_{k \in [n]} \phi_k(v_k) -  \phi_{A(\mathbf{D})}(\mathbf{v}) \right] \ge \frac{\Delta}{3}
\right\}
~.
\]

We will show in the next lemma there are at least $\Theta(n)$ such bidders with constant probability.

\begin{lemma}
	\label{lem:meta_hardness_bad_instance_count_const_prob}
	For the distribution $\mathbf{D}\in\mathcal{H}$ in \Cref{cor:meta_hardness_bad_instance}, with probability at least $0.01$, we have:
	\begin{equation*}
	|\mathcal{B}_{\mathbf{D},A(\mathbf{D})}|\ge \frac{n}{15}
	~.	
	\end{equation*}
\end{lemma}

\begin{proof}
	Let $p^* = \Pr_{A(\mathbf{D})} \big[ |\mathcal{B}_{\mathbf{D},A(\mathbf{D})}|\ge\frac{n}{15} \big]$ denote the probability that the conclusion of the lemma holds. 
	We need to show that $p^* \ge 0.01$.
	
	On one hand, each bidder $i \in \mathcal{B}_{\mathbf{D}}$ is in $\mathcal{B}_{\mathbf{D}, A(\mathbf{D})}$ with probability at least $\frac{1}{3}$ by definition.
	The expected size $\mathcal{B}_{\mathbf{D},A(\mathbf{D})}$ is therefore lower bounded as follows:
	\[
	\E_{A(\mathbf{D})} \big[ |\mathcal{B}_{\mathbf{D},A(\mathbf{D})}| \big] \ge \frac{1}{3} \cdot \frac{n - 1}{2} = \frac{n-1}{6} \ge \frac{n}{12}
	~.
	\]
	
	On the other hand, we have:
	\begin{align*}
	\E_{A(\mathbf{D})} \big[ |\mathcal{B}_{\mathbf{D},A(\mathbf{D})}| \big] 
	& \le 
	\frac{n}{15} \cdot \Pr_{A(\mathbf{D})} \left[ |\mathcal{B}_{\mathbf{D},A(\mathbf{D})}| < \frac{n}{15} \right] 
	+ 
	n \cdot \Pr_{A(\mathbf{D})} \left[ |\mathcal{B}_{\mathbf{D},A(\mathbf{D})}| \ge \frac{n}{15} \right] \\ 
	& = \frac{n}{15} \cdot ( 1 - p^* ) + n \cdot p^* \\
	& \le \frac{n}{15} + n \cdot p^*
	~.
	\end{align*} 
	
	Putting together gives $p^* \ge \frac{1}{60} > 0.01$.
%
%
\end{proof}

We now complete the proof of \Cref{lem:meta_hardness} by arguing that the algorithm must suffer from the stated revenue loss on the distribution $\mathbf{D}$ in \Cref{cor:meta_hardness_bad_instance} and \Cref{lem:meta_hardness_bad_instance_count_const_prob}.
In particular, when the conclusion of \Cref{lem:meta_hardness_bad_instance_count_const_prob} is true, which happens with probability at least $0.01$, we have the following sequence of inequalities:
\begin{align*}
& \opt(\mathbf{D}) -  \rev(A(\mathbf{D}), \mathbf{D})  \\[1.5ex]
& \qquad =  \int_{\mathbf{v}} \left( \max_{k \in [n]} \phi_k(v_k) - \phi_{A(\mathbf{D})}(\mathbf{v}) \right) d \mathbf{D}  && \text{(Eqn.~\eqref{eqn:optimal_accounting} and Eqn.~\eqref{eqn:algorithm_accounting})} \\
& \qquad \ge \sum_{i \in \mathcal{B}_{\mathbf{D},A(\mathbf{D})}} \int_{\mathbf{v} \in \mathcal{V}_i} \left( \max_{k \in [n]} \phi_k(v_k) - \phi_{A(\mathbf{D})}(\mathbf{v}) \right) d \mathbf{D}  && \text{($\mathcal{V}_i$'s are disjoint)} \\
& \qquad = \sum_{i \in \mathcal{B}_{\mathbf{D},A(\mathbf{D})}} \E_{\mathbf{v} \sim \mathbf{D} : \mathbf{v} \in \mathcal{V}_i} \left[ \max_{k \in [n]} \phi_k(v_k) -  \phi_{A(\mathbf{D})}(\mathbf{v}) \right] \cdot \Pr_{\mathbf{v} \sim \mathbf{D}} \big[ \mathbf{v} \in \mathcal{V}_i \big] \\
& \qquad \ge \sum_{i \in \mathcal{B}_{\mathbf{D},A(\mathbf{D})}} \frac{\Delta}{3} \cdot \Pr_{\mathbf{v} \sim \mathbf{D}} \big[ \mathbf{v} \in \mathcal{V}_i \big] && \text{(definition of $\mathcal{B}_{\mathbf{D},A(\mathbf{D})}$)} \\
& \qquad \ge \sum_{i \in \mathcal{B}_{\mathbf{D},A(\mathbf{D})}} \frac{\Delta}{3} \cdot p \cdot \left(1 - \frac{1}{n}\right)^{n-2} && \text{(conditions c) and d))} \\
& \qquad \ge \sum_{i \in \mathcal{B}_{\mathbf{D},A(\mathbf{D})}} \frac{p \Delta}{3e} \\
& \qquad \ge \frac{n p \Delta}{45 e} ~. && \text{(\Cref{lem:meta_hardness_bad_instance_count_const_prob})}
\end{align*}

\subsection{Proof of \texorpdfstring{\Cref{thm:lower_bounds}}{Theorem~\ref{thm:lower_bounds}}}
\label{sec:meta_hardness_instantiation}

Given \Cref{lem:meta_hardness}, it remains to construct three distributions $D^b$, $D^h$, and $D^\ell$ from each family of distributions that satisfy the conditions listed in \Cref{sec:meta_hard_instance}, and that:
\begin{enumerate}[label=\alph*)]
	\setcounter{enumi}{8}
	\item 
	$n p \Delta = \Omega \big( \epsilon\opt(\mathbf{D}) \big)$ (or $n p \Delta = \Omega(\epsilon)$ in the case of additive approximation).
	\label{property:revenue_gap}
	\item 
	$\skl(D^h, D^\ell)^{-1}$ matches the corresponding sample complexity lower bound in \Cref{thm:lower_bounds}.
	\label{property:kl}
\end{enumerate}

To make the argument of condition \ref{property:kl} easier, we will use a slight generalization of a technical lemma by \citet{HuangMR/2015/EC} as follows.

\begin{lemma}[\citet{HuangMR/2015/EC}, Lemma 4.4 and Lemma 4.5]
	\label{lem:dptrick}
	Suppose two distributions $P$ and $Q$ over a sample space $\Omega$, and $\Omega_1, \Omega_2, \dots, \Omega_k$ form a partition of $\Omega$.
	Further, suppose for every $1 \le i \le k$, there exists $0 \le \epsilon_i < 1$ such that:
	\[
	(1 + \epsilon_i)^{-1} \le \frac{dP}{dQ}(\omega) \le (1 + \epsilon_i)
	~,
	\]
	for every $\omega \in \Omega_i$.
	Then, we have:
	\[ 
	\skl(P, Q) \le \sum_{i=1}^k P_i(\Omega_i) \cdot \epsilon_i^2 
	~. 
	\]
\end{lemma}

We present in \Cref{fig:hardness_revenue_curves} the revenue curves, in the quantile space, of the distributions $D^h$ and $D^\ell$ that we use in the sample complexity lower bounds for different families of distributions. 
Readers who are familiar with this interpretation of value distributions may find the revenue curves more intuitive than the formal definitions of these distributions in the proof.

\begin{figure}
	\centering
	\begin{subfigure}{\textwidth}
		\includegraphics[width=0.48\textwidth]{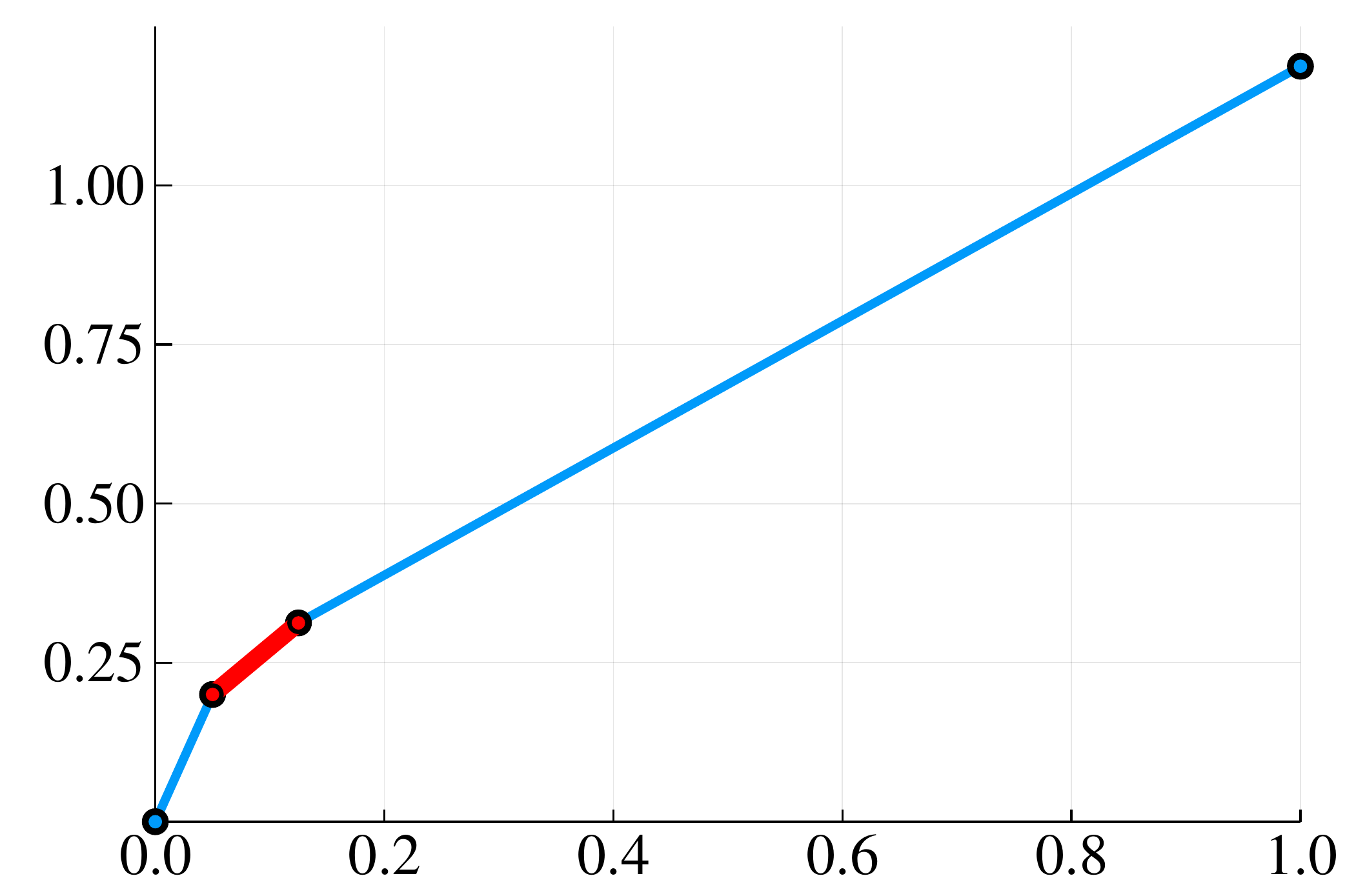}
		\includegraphics[width=0.48\textwidth]{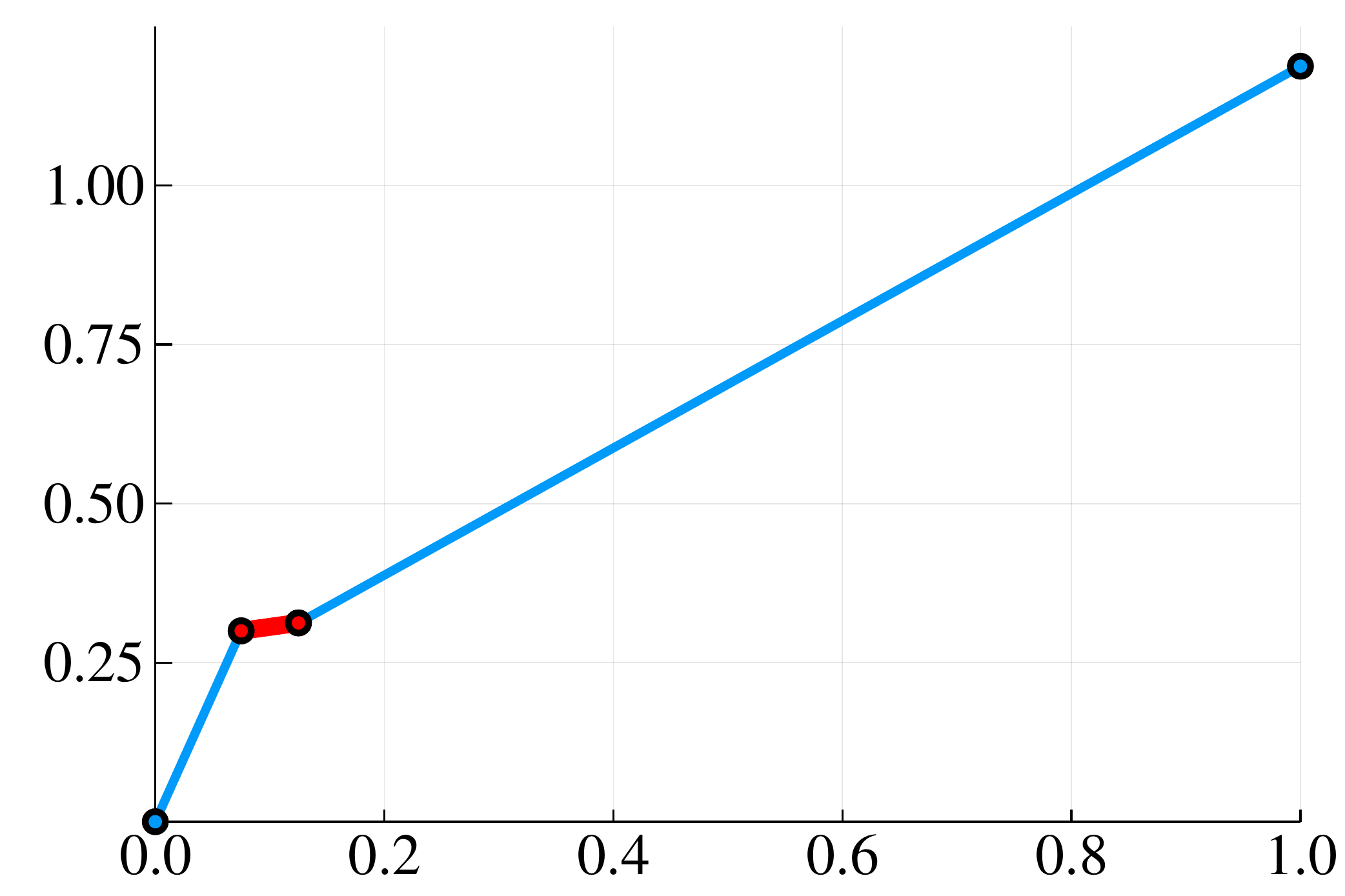}
		\caption{$[1, H]$-bounded support distributions with $n = 4$, $H = 4$, and $\epsilon = 0.2$.}
		\label{fig:bounded-support}
	\end{subfigure}
	
	\bigskip
	
	\begin{subfigure}{\textwidth}
		\includegraphics[width=0.48\textwidth]{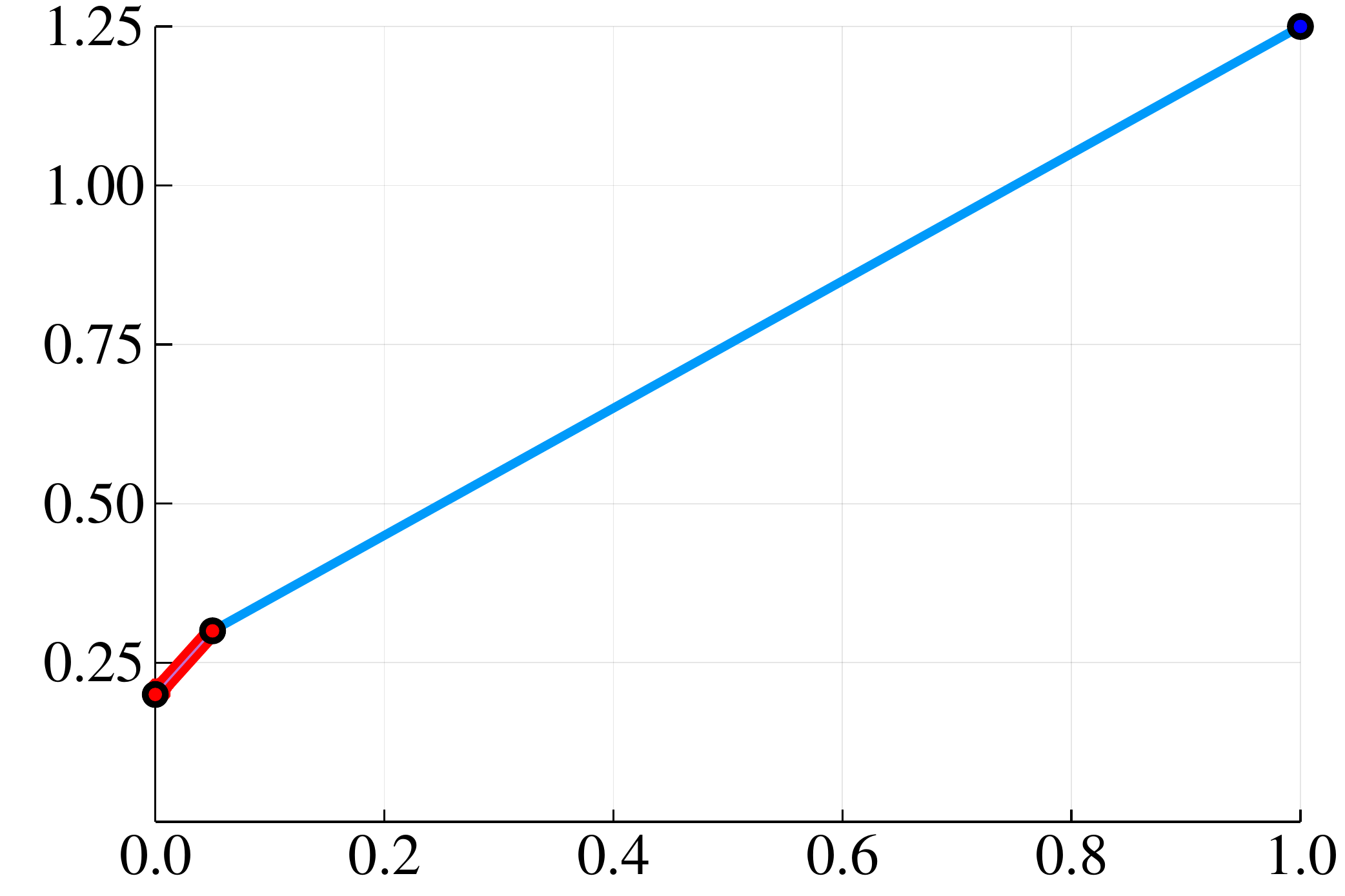}
		\includegraphics[width=0.48\textwidth]{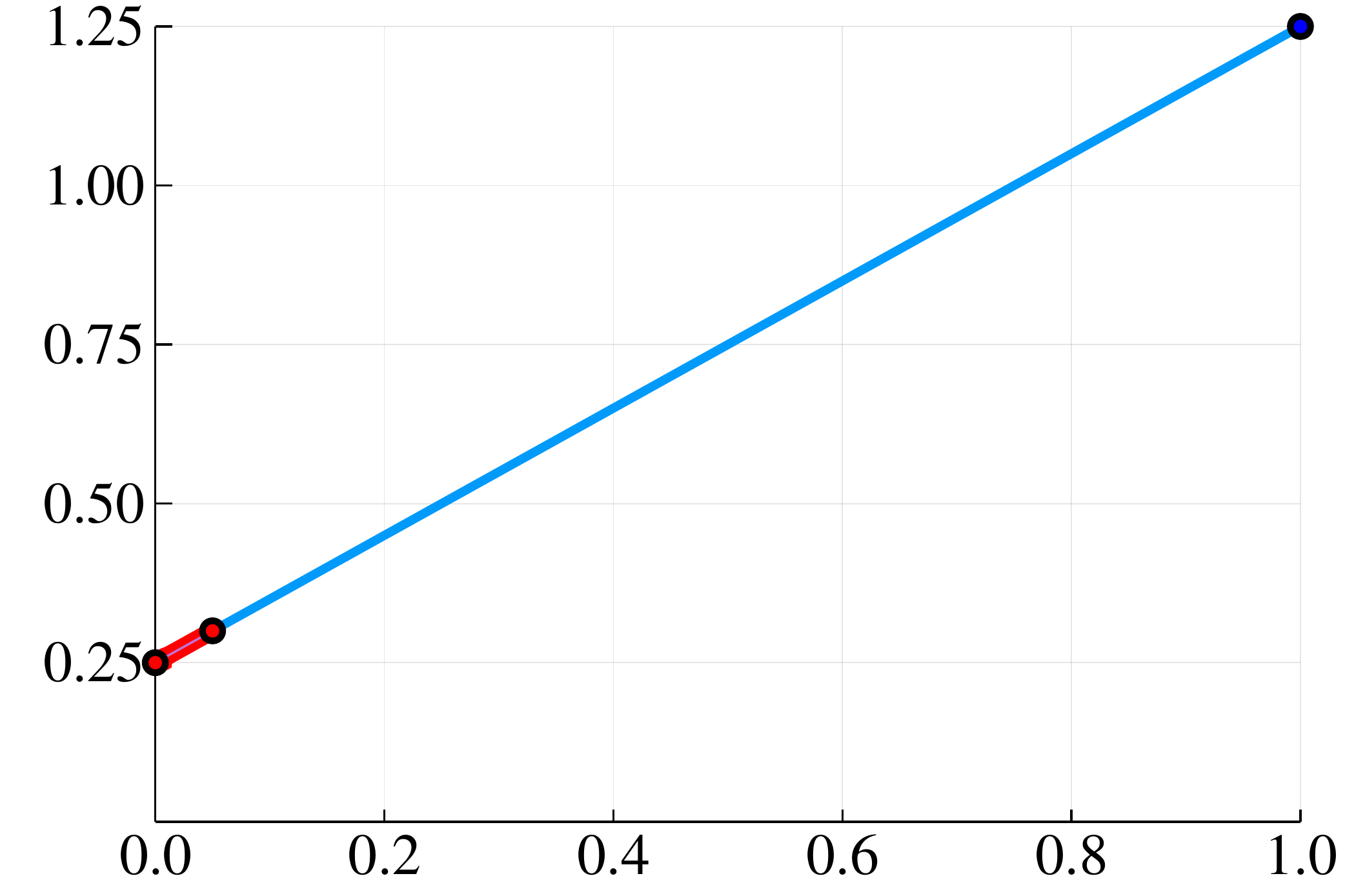}
		\caption{Regular distributions with $n = 4$, and $\epsilon = 0.2$.}
		\label{fig:regular}
	\end{subfigure}
	
	\bigskip
	
	\begin{subfigure}{\textwidth}
		\includegraphics[width=0.48\textwidth]{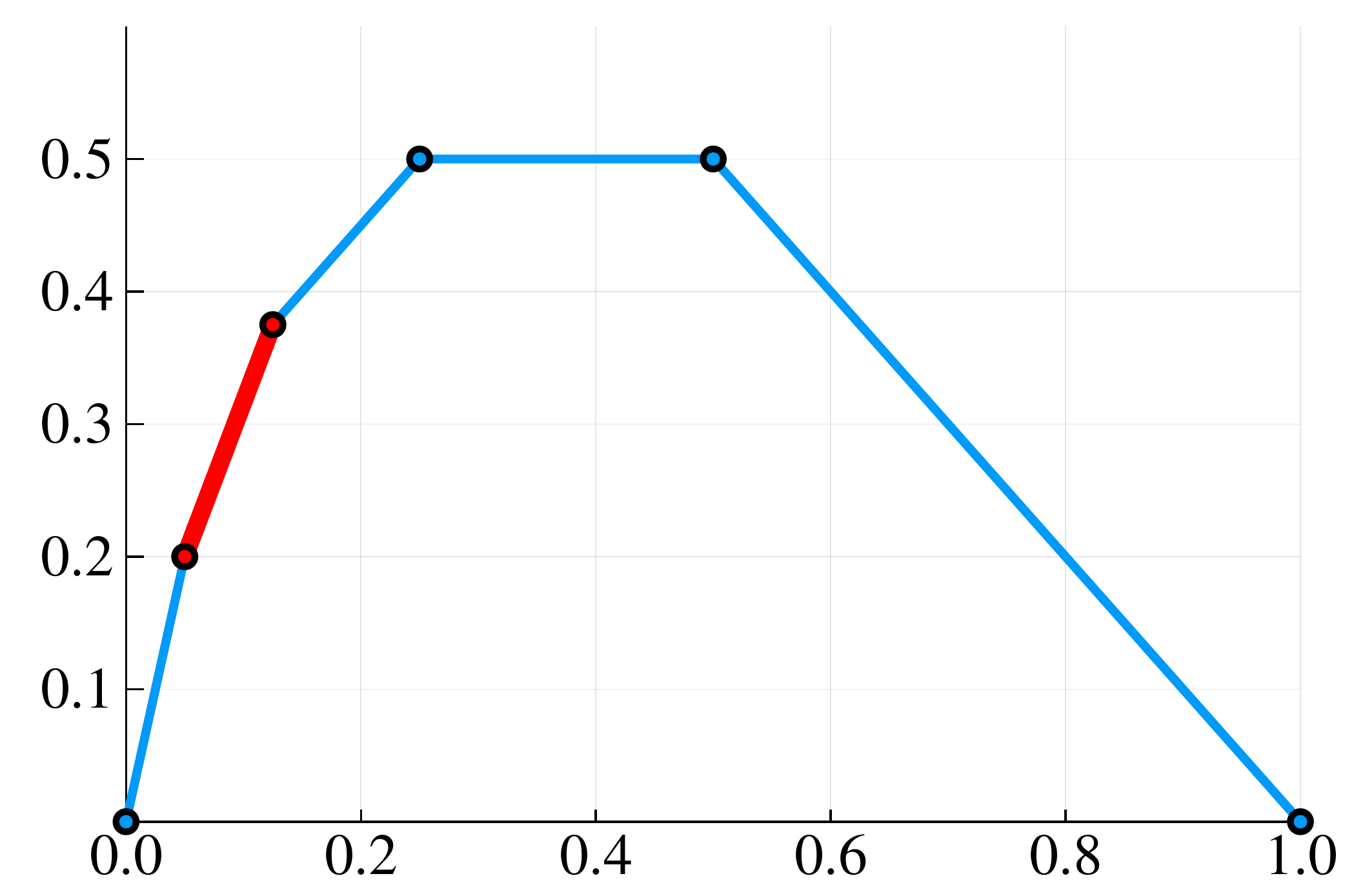}
		\includegraphics[width=0.48\textwidth]{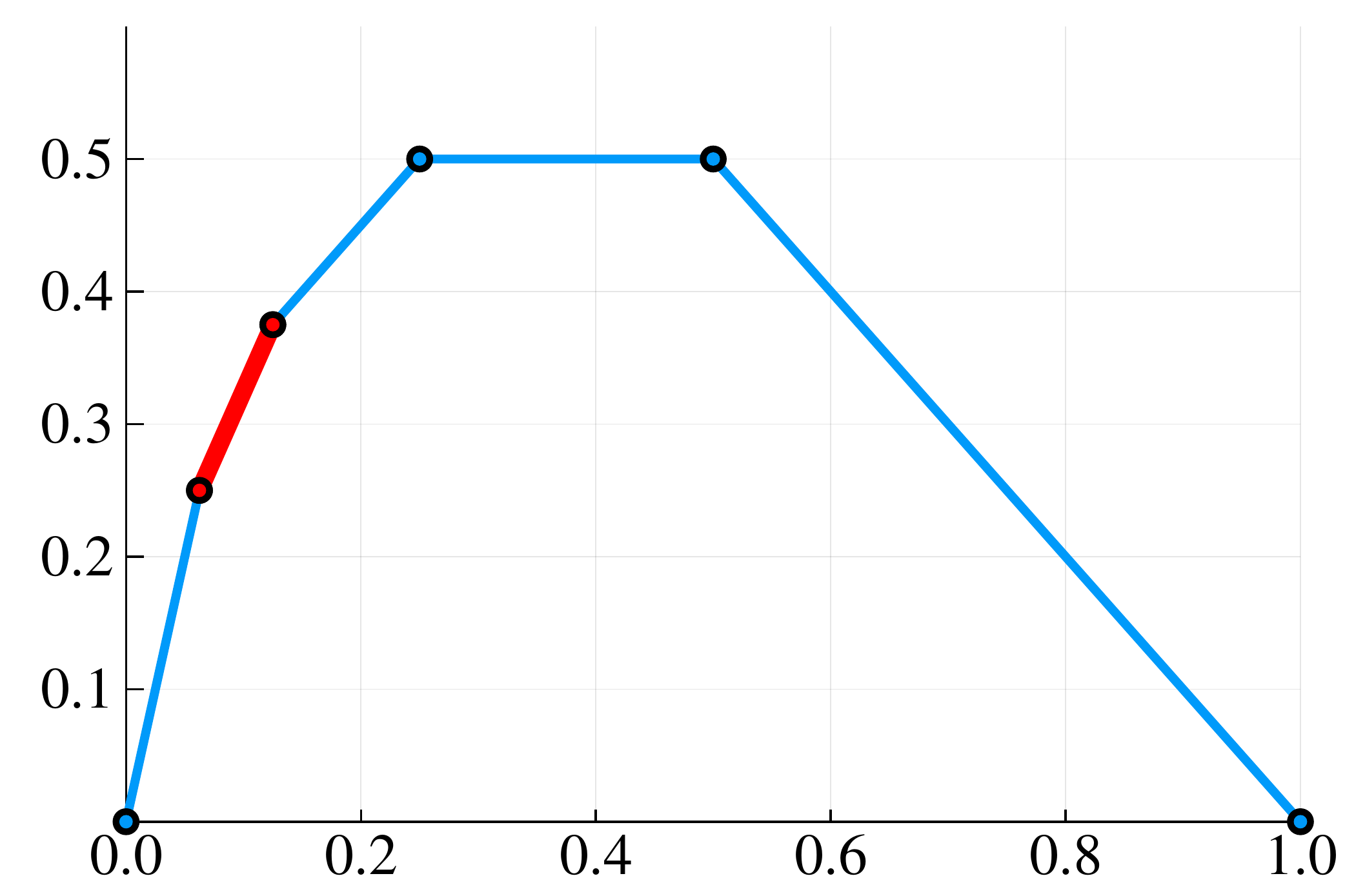}
		\caption{MHR distributions with $n = 8$, and $\epsilon_0 = 0.2$.}
		\label{fig:mhr}
	\end{subfigure}

	\caption{
		The revenue curves, in the quantile space, of the distributions that we use to prove the sample complexity lower bounds for different families of distributions. 
		In each sub-figure, the left-hand-size and the right-hand-side are the revenue curves of $D^h$ and $D^\ell$ respectively for the corresponding family.
		The critical value interval between $v_2$ and $v_1$ is plotted in bold (red).
	}
	\label{fig:hardness_revenue_curves}
\end{figure}

\subsubsection*{$[1,H]$-Bounded Distributions}

Let $v_0 = 1$, $v_1 = H$, $v_2 = \frac{H}{2} + \frac{1}{2}$, $p = \frac{2}{nH}$, and $\Delta = \epsilon H$ be the parameters.
Let $D^b$ be a singleton at $v_0 = 1$.
Define $D^\ell$ and $ D^h$ with the following probability mass:
\begin{align*}
f_{D^\ell}(v) & = 
\begin{cases}
1-\frac{2}{nH}& v = 1 + \frac{1}{n} - \frac{1}{nH} \\
\frac{1-\epsilon}{nH}& v = \frac{H}{2} + \frac{1}{2} \\
\frac{1+\epsilon}{nH}& v = H
\end{cases} \\
f_{D^h}(v) & = 
\begin{cases}
1-\frac{2}{nH}& v = 1 + \frac{1}{n} - \frac{1}{nH} \\
\frac{1+\epsilon}{nH}& v = \frac{H}{2} + \frac{1}{2}\\
\frac{1-\epsilon}{nH}& v = H
\end{cases}
\end{align*}

Conditions \ref{property:base_pointmass}, \ref{property:prob_high_values}, \ref{property:prob_critical_interval}, \ref{property:density_gap}, and \ref{property:v1_point_mass} hold trivially by the construction.
Condition \ref{property:kl} follows by the construction and \Cref{lem:dptrick}, with $\Omega_1 = \{ 1 \}$, $\epsilon_1 = 0$, and $\Omega_2 = \{ \frac{H}{2} + \frac{1}{2}, H \}$, $\epsilon_2 = \epsilon$.
Further, conditions \ref{property:virtual_value_gap}, \ref{property:virtual_value_bound_small_values}, and \ref{property:regular_high_distribution} can be verified from the virtual values of $D^h$ and $D^\ell$ below, which follow from straightforward calculations:
\begin{align*}
\phi^\ell(v) & =
\begin{cases}
1 & v = 1 + \frac{1}{n} - \frac{1}{nH} \\
1 - \Theta(\epsilon H) & v=\frac{H}{2}+\frac{1}{2}\\
H & v = H
\end{cases} \\
\phi^h(v) & =
\begin{cases}
1 & v = 1 + \frac{1}{n} - \frac{1}{nH} \\
1 + \Theta(\epsilon H) & v=\frac{H}{2}+\frac{1}{2}\\
H & v = H
\end{cases}    
\end{align*}


To show condition \ref{property:revenue_gap}, by our choice of $p$ and $\Delta$ and, thus, $n p \Delta = \Theta(\epsilon)$, it remains to show that $\opt(\mathbf{D}) \le O(1)$ for all $\mathbf{D} \in \mathcal{H}$.
This follows by that the optimal revenue is upper bounded by the expected optimal social welfare, which is upper bounded by:
\begin{align*}
& \left( 1 + \frac{1}{n} - \frac{1}{nH} \right) + H \cdot \Pr \left[ \exists i : v_i > 1 + \frac{1}{n} - \frac{1}{nH} \right] \\
& \qquad\qquad \le 1 + \frac{1}{n} - \frac{1}{nH} + H \cdot\sum_{i = 2}^n \Pr \left[ v_i > 1 + \frac{1}{n} - \frac{1}{nH} \right] \\
& \qquad\qquad \le 1 + \frac{1}{n} - \frac{1}{nH} + H \cdot (n - 2) \cdot \frac{2}{nH} < 3
~.
\end{align*}

Now that we have verified all conditions, we get the sample complexity lower bound as stated in \Cref{thm:lower_bounds} for $[1, H]$-bounded support distributions as a corollary of \Cref{lem:meta_hardness}. 

\subsubsection*{$[0, 1]$-Bounded Support Distributions}

In this case, we can simply take the construction of the $[1, H]$-bounded case, with $H$ being a large constant, say $H = 10$, and scale all values down by a multiplicative factor of $\frac{1}{H}$. 
Then, we have the sample complexity lower bound as stated in \Cref{thm:lower_bounds}.

\subsubsection*{Regular Distributions}

Let $v_0 = \frac{3}{2}$, $v_1 = +\infty$, $v_2 = 1 + \frac{1}{\epsilon}$, $p = \frac{\epsilon}{n}$, and $\Delta = \frac{1}{2}$ be the parameters.
Let $D^b$ be a singleton at $v_0 = \frac{3}{2}$, and define $D^\ell$ and $D^h$ with support $[1 + \frac{1}{n}, +\infty)$ and the following pdf:
%
\begin{align*}
&f_{D^h}(v)=\begin{cases}\frac{1}{n(v-1)^2} & 1 + \frac{1}{n} \le v < 1+\frac{1}{\epsilon}\\ \frac{1-\epsilon}{n(v-2)^2}&v \ge 1+\frac{1}{\epsilon}\end{cases}\\
&f_{D^\ell}(v) = \frac{1}{n(v-1)^2} 
\end{align*}

The corresponding complementary cdf's are as follows, via simple calculations:
\begin{align*}
1 - F_{D^h}(v) & =
\begin{cases}
\frac{1}{n(v-1)} & 1 + \frac{1}{n} \le v < 1 + \frac{1}{\epsilon} \\ 
\frac{1-\epsilon}{n(v-2)} & v \ge 1+\frac{1}{\epsilon}
\end{cases} \\
1 - F_{D^\ell}(v) & = \frac{1}{n(v-1)}
\end{align*}

Conditions \ref{property:base_pointmass}, \ref{property:prob_high_values}, \ref{property:prob_critical_interval}, \ref{property:density_gap}, and \ref{property:v1_point_mass} hold trivially by the construction.
Condition \ref{property:kl} follows by the construction and \Cref{lem:dptrick}, with $\Omega_1 = [1+\frac{1}{n}, 1 + \frac{1}{\epsilon})$, $\epsilon_1 = 0$, and $\Omega_2 = [1 + \frac{1}{\epsilon}, +\infty)$, $\epsilon_2 = \epsilon$.
Further, conditions \ref{property:virtual_value_gap}, \ref{property:virtual_value_bound_small_values}, and \ref{property:regular_high_distribution} can be verified from the virtual values of $D^h$ and $D^\ell$ below, which follow from straightforward calculations:
\begin{align*}
\phi_{D^h}(v) & =
\begin{cases}
1 & 1 \le v < 1+\frac{1}{\epsilon} \\ 
2 & v \ge 1+\frac{1}{\epsilon}
\end{cases} \\
\phi_{D^\ell}(v) & = 1
\end{align*}

To show condition \ref{property:revenue_gap}, by our choice of $p$ and $\Delta$ and, thus, $n p \Delta = \Theta(\epsilon)$, it remains to show that $\opt(\mathbf{D}) \le O(1)$ for all $\mathbf{D} \in \mathcal{H}$.
This holds because the virtual values in the three distributions $D^b$, $D^h$, and $D^\ell$ are at most $2$ everywhere as discussed above.

Now that we have verified all conditions, we get the sample complexity lower bound as stated in \Cref{thm:lower_bounds} for regular distributions as a corollary of \Cref{lem:meta_hardness}. 

\subsubsection*{MHR Distributions}

Assume for simplicity of notations that $\log n$ is an integer.
Let $\epsilon_0 = \epsilon \log n$.
Let $v_0 = \log n - 1 + \epsilon_0$, $v_1 = \log n + 1$, $v_2 = \log n$, $p = \frac{1}{n}$, and $\Delta = \epsilon_0$ be the parameters.
Let $D^b$ be a singleton at $v_0 = \log n - 1 + \epsilon_0$.
Define $D^h$ and $D^\ell$ with support $\{0, 1, 2, \dots, \log n + 1\}$ and the following probability mass:
\begin{align*}
f_{D^h}(v) & = 
\begin{cases}
2^{-v-1} & 0 \le v < \log n \\ 
\frac{1 + \epsilon_0}{2n} & v = \log n \\
\frac{1 - \epsilon_0}{2n} & v = \log n + 1
\end{cases} \\
f_{D^\ell}(v) & =
\begin{cases}
2^{-v-1} & 0 \le v \le \log n \\ 
\frac{1}{2n} & v = \log n \\
\frac{1}{2n} & v = \log n + 1
\end{cases}
\end{align*}
%

The corresponding complementary cdf are as follows, via simple calculations:
\begin{align*}
1- F_{D^h}(v) & = 
\begin{cases}
2^{-v-1} & 0 \le v < \log n \\ 
\frac{1 - \epsilon_0}{2n} & v = \log n \\
0 & v = \log n + 1
\end{cases} \\
1 - F_{D^\ell}(v) & =
\begin{cases}
2^{-v-1} & 0 \le v < \log n \\ 
\frac{1}{2n} & v = \log n \\
0 & v = \log n + 1
\end{cases}
\end{align*}

Conditions \ref{property:base_pointmass}, \ref{property:prob_high_values}, \ref{property:prob_critical_interval}, \ref{property:density_gap}, and \ref{property:v1_point_mass} hold trivially by the construction.
Condition \ref{property:kl} follows by the construction and \Cref{lem:dptrick}, with $\Omega_1 = \{0, 1, 2, \dots, \log n - 1\}$, $\epsilon_1 = 0$, and $\Omega_2 = \{ \log n, \log n + 1\}$, $\epsilon_2 = \epsilon_0$.
Further, conditions \ref{property:virtual_value_gap}, \ref{property:virtual_value_bound_small_values}, and \ref{property:regular_high_distribution} can be verified from the virtual values of $D^h$ and $D^\ell$ below, which follow from straightforward calculations:
\begin{align*}
\phi_{D^h}(v) & = 
\begin{cases}
v - 1 & 0 \le v < \log n \\ 
\log n - 1 + 2\epsilon_0 - O(\epsilon_0^2) & v = \log n \\
\log n + 1 & v = \log n + 1
\end{cases} \\
\phi_{D^\ell}(v) & =
\begin{cases}
v - 1 \qquad\qquad\qquad\quad~~ & 0 \le v < \log n \\ 
\log n - 1 & v = \log n \\
\log n + 1 & v = \log n + 1
\end{cases}
\end{align*}

To show condition \ref{property:revenue_gap}, by our choice of $p$ and $\Delta$ and, thus, $n p \Delta = \Theta(\epsilon \log n)$, it remains to show that $\opt(\mathbf{D}) \le O(\log n)$ for all $\mathbf{D} \in \mathcal{H}$.
This holds trivially because the values are upper bounded by $\log n + 1$ in all three distributions $D^b$, $D^h$, and $D^\ell$.

Now that we have verified all conditions, we get the sample complexity lower bound as stated in \Cref{thm:lower_bounds} for MHR distributions as a corollary of \Cref{lem:meta_hardness}. 

Finally, we include in \Cref{app:continuous_mhr_lb} a discussion on the sample complexity lower bounds under the same framework if we insist on using only continuous MHR distributions.

	\newpage
	
	\bibliographystyle{plainnat}
	\bibliography{learning-auction-refs}

\begin{thebibliography}{26}
\providecommand{\natexlab}[1]{#1}
\providecommand{\url}[1]{\texttt{#1}}
\expandafter\ifx\csname urlstyle\endcsname\relax
  \providecommand{\doi}[1]{doi: #1}\else
  \providecommand{\doi}{doi: \begingroup \urlstyle{rm}\Url}\fi

\bibitem[Balcan et~al.(2018)Balcan, Sandholm, and Vitercik]{BalcanSV/2018/EC}
Maria-Florina Balcan, Tuomas Sandholm, and Ellen Vitercik.
\newblock A general theory of sample complexity for multi-item profit
  maximization.
\newblock In \emph{Proceedings of the 19th ACM Conference on Economics and
  Computation}, pages 173--174. ACM, 2018.

\bibitem[Balcan et~al.(2016)Balcan, Sandholm, and Vitercik]{BalcanSV/2016/NIPS}
Maria-Florina~F Balcan, Tuomas Sandholm, and Ellen Vitercik.
\newblock Sample complexity of automated mechanism design.
\newblock In \emph{Advances in Neural Information Processing Systems}, pages
  2083--2091, 2016.

\bibitem[Barlow et~al.(1963)Barlow, Marshall, Proschan,
  et~al.]{BarlowMP/1963/AoMS}
Richard~E Barlow, Albert~W Marshall, Frank Proschan, et~al.
\newblock Properties of probability distributions with monotone hazard rate.
\newblock \emph{The Annals of Mathematical Statistics}, 34\penalty0
  (2):\penalty0 375--389, 1963.

\bibitem[Bernstein(1924)]{Bernstein/1924}
Sergei Bernstein.
\newblock On a modification of {C}hebyshev's inequality and of the error
  formula of {L}aplace.
\newblock \emph{Ann. Sci. Inst. Sav. Ukraine, Sect. Math}, 1\penalty0
  (4):\penalty0 38--49, 1924.

\bibitem[Blum and Hartline(2005)]{BlumH/2005/SODA}
Avrim Blum and Jason~D Hartline.
\newblock Near-optimal online auctions.
\newblock In \emph{Proceedings of the 16th Annual ACM-SIAM Symposium on
  Discrete Algorithms}, pages 1156--1163. Society for Industrial and Applied
  Mathematics, 2005.

\bibitem[Bubeck et~al.(2017)Bubeck, Devanur, Huang, and
  Niazadeh]{BubeckDHN/2017/EC}
Sebastien Bubeck, Nikhil~R Devanur, Zhiyi Huang, and Rad Niazadeh.
\newblock Online auctions and multi-scale online learning.
\newblock In \emph{Proceedings of the 18th ACM Conference on Economics and
  Computation}, pages 497--514. ACM, 2017.

\bibitem[Cai(2018)]{Cai/2018/communication}
Yang Cai.
\newblock Personal communication, 2018.

\bibitem[Cai and Daskalakis(2011)]{CaiD/2011/FOCS}
Yang Cai and Constantinos Daskalakis.
\newblock Extreme-value theorems for optimal multidimensional pricing.
\newblock In \emph{Proceedings of the 52nd IEEE Annual Symposium on Foundations
  of Computer Science}, pages 522--531. IEEE, 2011.

\bibitem[Cai and Daskalakis(2017)]{CaiD/2017/FOCS}
Yang Cai and Constantinos Daskalakis.
\newblock Learning multi-item auctions with (or without) samples.
\newblock In \emph{Proceedings of the 58th Annual IEEE Symposium on Foundations
  of Computer Science}, 2017.

\bibitem[Cole and Roughgarden(2014)]{ColeR/2014/STOC}
Richard Cole and Tim Roughgarden.
\newblock The sample complexity of revenue maximization.
\newblock In \emph{Proceedings of the 46th Annual ACM Symposium on Theory of
  Computing}, pages 243--252. ACM, 2014.

\bibitem[Devanur et~al.(2016)Devanur, Huang, and Psomas]{DevanurHP/2016/STOC}
Nikhil~R Devanur, Zhiyi Huang, and Christos-Alexandros Psomas.
\newblock The sample complexity of auctions with side information.
\newblock In \emph{Proceedings of the 48th annual ACM symposium on Theory of
  Computing}, pages 426--439. ACM, 2016.

\bibitem[Dhangwatnotai et~al.(2015)Dhangwatnotai, Roughgarden, and
  Yan]{DhangwatnotaiRY/2015/GEB}
Peerapong Dhangwatnotai, Tim Roughgarden, and Qiqi Yan.
\newblock Revenue maximization with a single sample.
\newblock \emph{Games and Economic Behavior}, 91:\penalty0 318--333, 2015.

\bibitem[Dughmi et~al.(2014)Dughmi, Han, and Nisan]{DughmiHN/2014/WINE}
Shaddin Dughmi, Li~Han, and Noam Nisan.
\newblock Sampling and representation complexity of revenue maximization.
\newblock In \emph{International Conference on Web and Internet Economics},
  pages 277--291. Springer, 2014.

\bibitem[Dvoretzky et~al.(1956)Dvoretzky, Kiefer, and
  Wolfowitz]{DvoretzkyKW/1956/AoMS}
Aryeh Dvoretzky, Jack Kiefer, and Jacob Wolfowitz.
\newblock Asymptotic minimax character of the sample distribution function and
  of the classical multinomial estimator.
\newblock \emph{The Annals of Mathematical Statistics}, pages 642--669, 1956.

\bibitem[Elkind(2007)]{Elkind/2007/SODA}
Edith Elkind.
\newblock Designing and learning optimal finite support auctions.
\newblock In \emph{Proceedings of the 18th Annual ACM-SIAM Symposium on
  Discrete Algorithms}, pages 736--745. Society for Industrial and Applied
  Mathematics, 2007.

\bibitem[Gonczarowski and Nisan(2017)]{GonczarowskiN/2017/STOC}
Yannai~A Gonczarowski and Noam Nisan.
\newblock Efficient empirical revenue maximization in single-parameter auction
  environments.
\newblock In \emph{Proceedings of the 49th Annual ACM Symposium on Theory of
  Computing}, pages 856--868. ACM, 2017.

\bibitem[Gonczarowski and Weinberg(2018)]{GonczarowskiW/2018/FOCS}
Yannai~A Gonczarowski and S~Matthew Weinberg.
\newblock The sample complexity of up-to-$varepsilon$ multi-dimensional revenue
  maximization.
\newblock In \emph{Proceedings of the 59th Annual IEEE Symposium on Foundations
  of Computer Science}, 2018.

\bibitem[Hart and Reny(2015)]{HartR/2015/TE}
Sergiu Hart and Philip~J Reny.
\newblock Maximal revenue with multiple goods: Nonmonotonicity and other
  observations.
\newblock \emph{Theoretical Economics}, 10\penalty0 (3):\penalty0 893--922,
  2015.

\bibitem[Huang et~al.(2015)Huang, Mansour, and Roughgarden]{HuangMR/2015/EC}
Zhiyi Huang, Yishay Mansour, and Tim Roughgarden.
\newblock Making the most of your samples.
\newblock In \emph{Proceedings of the 16th ACM Conference on Economics and
  Computation}, pages 45--60. ACM, 2015.

\bibitem[Morgenstern and Roughgarden(2016)]{MorgensternR/2016/COLT}
Jamie Morgenstern and Tim Roughgarden.
\newblock Learning simple auctions.
\newblock In \emph{Conference on Learning Theory}, 2016.

\bibitem[Morgenstern and Roughgarden(2015)]{MorgensternR/2015/NIPS}
Jamie~H Morgenstern and Tim Roughgarden.
\newblock The pseudo-dimension of nearly optimal auctions.
\newblock In \emph{Advances in Neural Information Processing Systems}, pages
  136--144, 2015.

\bibitem[Myerson(1981)]{Myerson/1981/MOR}
Roger~B. Myerson.
\newblock Optimal auction design.
\newblock \emph{Mathematics of Operations Research}, 6\penalty0 (1):\penalty0
  58--73, 1981.

\bibitem[Roughgarden and Schrijvers(2016)]{RoughgardenS/2016/EC}
Tim Roughgarden and Okke Schrijvers.
\newblock Ironing in the dark.
\newblock In \emph{Proceedings of the 17th ACM Conference on Economics and
  Computation}, pages 1--18. ACM, 2016.

\bibitem[Rubinstein and Weinberg(2015)]{RubinsteinW/2015/EC}
Aviad Rubinstein and S~Matthew Weinberg.
\newblock Simple mechanisms for a subadditive buyer and applications to revenue
  monotonicity.
\newblock In \emph{Proceedings of the 16th ACM Conference on Economics and
  Computation}, pages 377--394. ACM, 2015.

\bibitem[Syrgkanis(2017)]{Syrgkanis/2017/NIPS}
Vasilis Syrgkanis.
\newblock A sample complexity measure with applications to learning optimal
  auctions.
\newblock In \emph{Advances in Neural Information Processing Systems}, pages
  5352--5359, 2017.

\bibitem[Yao(2018)]{Yao/2018/SAGT}
Andrew Chi-Chih Yao.
\newblock On revenue monotonicity in combinatorial auctions.
\newblock In \emph{Proceedings of the 11th International Symposium on
  Algorithmic Game Theory}, pages 1--11. Springer, 2018.

\end{thebibliography}
	
	\newpage
		
	\appendix
	
	\section{Previous Approaches Rely on Prior Knowledge}
\label{app:prior_knowledge}

In this section we demonstrate through an example how the previous approaches crucially rely on knowing the family of distributions upfront, even in the special case of a single bidder.
In this case, any truthful auction is effectively posting a take-it-or-leave-it price.
What the algorithm needs to do is to learn an approximately optimal price from the samples.

Consider a value distribution $D$ with the following cdf:
\[
F_D(v) = 
\begin{cases}
\frac{v}{v+2} & 0 \le v \le 2 \\
\frac{2v-3}{2v-2} & v > 2
\end{cases}
\]
This is a regular distribution.
Its revenue curve in the quantile space is piece-wise linear with two segments, as shown in \Cref{fig:prior_knowledge_revenue_curve}.
The $x$-axis is the quantile; the $y$-axis is the expected revenue for a price with the corresponding quantile in $D$.
The optimal price is $2$.
At this price, the bidder will buy the item with probability $0.5$.
The corresponding optimal revenue is $1$.
Importantly, the distribution has a heavy tail in the sense that the expected revenue does not diminish as the price tends to infinity; instead, it tends to $0.5$.

We evaluated the performance of the best known previous algorithms for regular, MHR, and $[1, H]$-bounded support (with $H = 100$) distributions on $D$.
The results are shown in \Cref{fig:prior_knowledge}.
The $x$-axis is the number of samples.
The $y$-axis is the approximation ratio in comparison with the optimal revenue as the number of samples increases.

First, suppose we have an incorrect belief that the underlying distribution is an MHR distribution and simply use the empirical price, i.e., the optimal price w.r.t.\ the uniform distribution over the samples~\cite{DhangwatnotaiRY/2015/GEB, HuangMR/2015/EC}.
Then, the approximation ratio does not converge at all, as shown by the dashed line in \Cref{fig:prior_knowledge}.
To see why, consider any large integer $m$ and let there be $m$ samples.
Consider the value $\frac{m}{2}$ whose quantile is roughly $\frac{1}{m}$.
In expectation, there shall be $1$ sample, out of all $m$ of them, that has value at least $\frac{m}{2}$.
However, there is a non-negligible probability that there are at least $3$ such samples, in which case it may seem superior than the actual optimal price $2$, based on the empirical distribution.
As a result, the approximation ratio fails to converge to $1$.

Next, suppose we have the correct prior knowledge and choose the best know previous algorithm for regular distributions~\cite{DhangwatnotaiRY/2015/GEB, HuangMR/2015/EC}, which picks the best price according to the empirical distribution \emph{subject to having a sale probability at least $\delta = m^{-1/3}$}.
This is called the $\delta$-guarded empirical price.
The choice of $\delta = m^{-1/3}$ is due to the choice of $\delta = \epsilon$ in previous works when the goal is a multiplicative $(1 - \epsilon)$-approximation, and the optimal sample complexity of $\tilde{\Theta}(\epsilon^{-3})$ for regular distributions.
Then, the approximation ratio quickly converges to almost $1$ with less than a thousand samples, as shown by the solid line in \Cref{fig:prior_knowledge}.

Further, suppose we have an incorrect belief that the underlying distribution is a $[1, H]$-bounded distribution, with $H = 100$. 
The best known previous algorithm for this case is again the $\delta$-guarded empirical price, but with a different choice of $\delta = \frac{1}{H}$~\cite{HuangMR/2015/EC}.
Then, even though the approximation ratio does converge to $1$, it is much slower, as shown by the dotted line in \Cref{fig:prior_knowledge}.

Last but not least, the performance of the algorithm proposed in this paper, which picks the optimal price w.r.t.\ the dominated empirical distribution with no prior knowledge on the family of distributions at all, is almost indistinguishable with that of the best known previous algorithm tailored specifically for regular distributions.
This is shown by the dash-dotted line in \Cref{fig:prior_knowledge}.

\begin{figure}
	\centering
	\includegraphics[width=.8\textwidth]{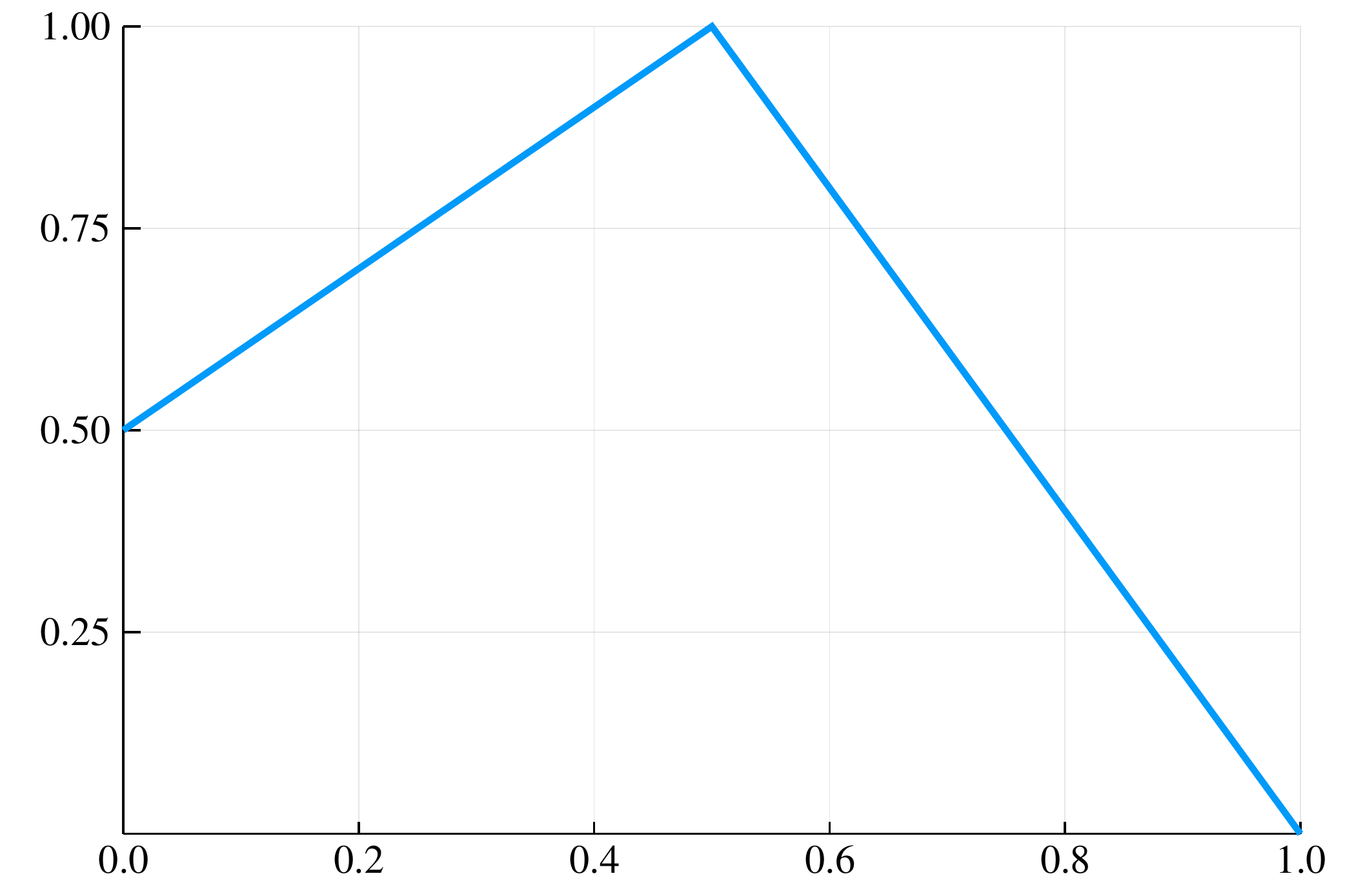}
	\caption{The revenue curve, in the quantile space, of a heavy-tail regular distribution that demonstrates how previous approaches crucially rely on knowing the family of distributions}
	\label{fig:prior_knowledge_revenue_curve}
\end{figure}

\begin{figure}
	\centering
	\includegraphics[width=.8\textwidth]{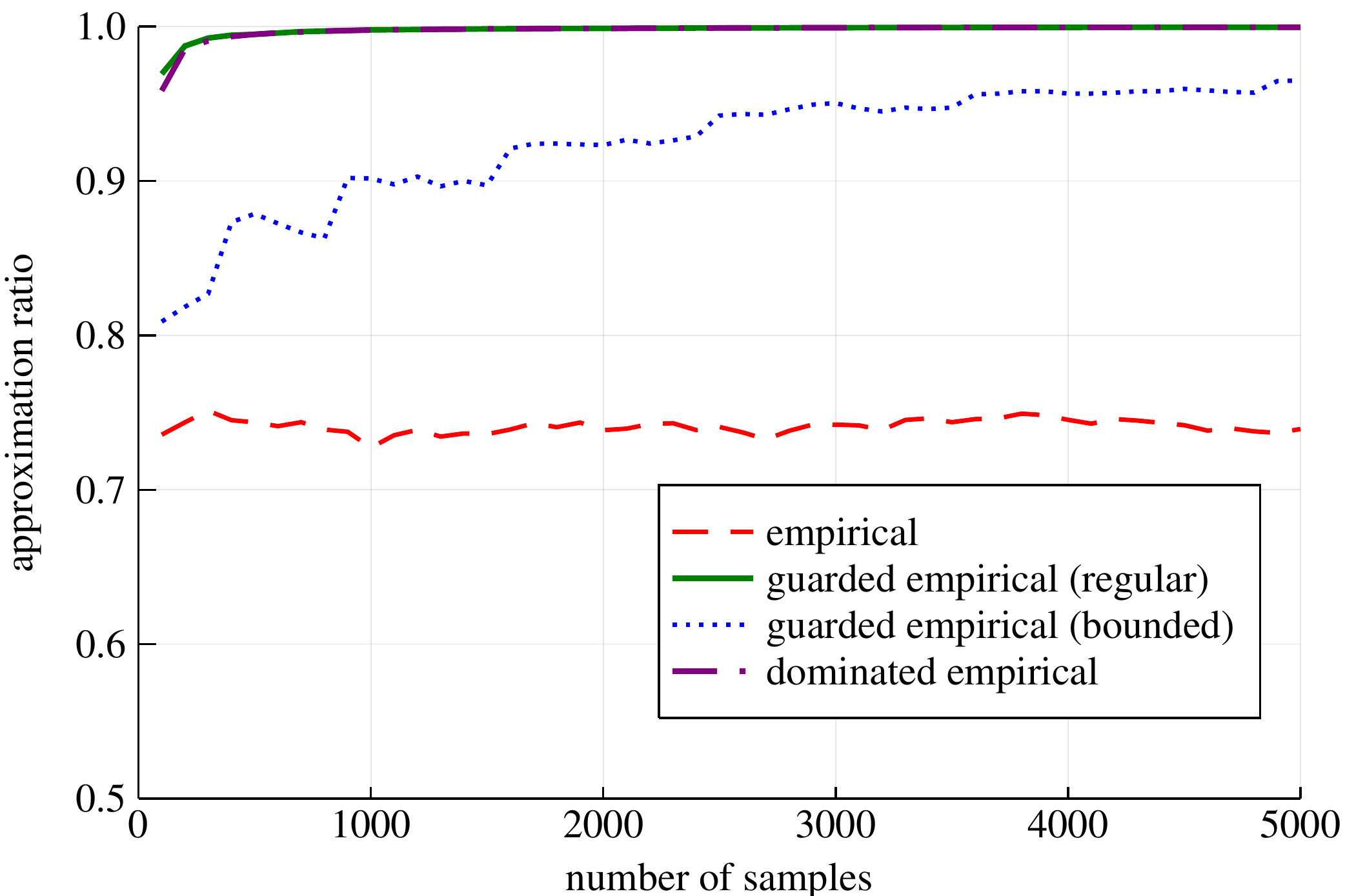}
	\caption{The convergences of best known previous algorithms for different families of distributions when there is a single bidder whose value distribution is defined by the revenue curve in \Cref{fig:prior_knowledge_revenue_curve}}
	\label{fig:prior_knowledge}
\end{figure}

	\section{Missing Proofs in \texorpdfstring{\Cref{sec:upper_bound_multi_bidder}}{Section~\ref{sec:upper_bound_multi_bidder}}}
\label{app:missing_proofs_upper_bounds}

\subsection{Proof of \texorpdfstring{\Cref{lem:empirical_error_bound}}{Lemma~\ref{lem:empirical_error_bound}}}

Consider quantiles that are multiples of $\frac{1}{m}$ and the corresponding values.
There are at most $m$ of them.
We will show that a stronger bound holds for each of such value $v$:
\begin{equation}
\label{eqn:empirical_error_bound}
\left| q^{E_i}(v) - q^{D_i}(v) \right| \le \sqrt{\frac{2 q^{D_i}(v) \big( 1-q^{D_i}(v) \big) \ln(2mn\delta^{-1})}{m} } + \frac{2\ln(2mn\delta^{-1})}{3m}
~.
\end{equation}
Then, the lemma follows because the bound on difference in the quantiles of any value in the two distributions bares an extra additive factor of at most $\frac{1}{m} < \frac{\ln(2mn\delta^{-1})}{3m}$.

Now fix any such value $v$.
Consider random variables $X_1, X_2, \dots, X_m$ such that $X_i = 1$ if the $i$-th sample is at least $v$ and $X_i = 0$ otherwise.
Then, we have $q^{E_i}(v) = \frac{1}{m} \sum_{i = 1}^m X_i$.
By Bernstein's inequality (\Cref{lem:bernstein}), with $\sigma^2 = q^{D_i}(v) \big( 1 - q^{D_i}(v) \big)$, $M = 1$, and 
\[
t = \sqrt{2 q^{D_i}(v) \big( 1-q^{D_i}(v) \big) m \ln(2mn\delta^{-1})} + \frac{2}{3}\ln(2mn\delta^{-1})
~,
\]
the bound stated in \Cref{eqn:empirical_error_bound} holds for $v$ fails with probability at most $\frac{\delta}{mn}$.

By union bound we get that, with probability at least $1 - \delta$, \Cref{eqn:empirical_error_bound} holds for all $m$ values that correspond to quantiles that are multiples of $\frac{1}{m}$, and all $n$ bidders.

\subsection{Proof of \texorpdfstring{\Cref{lem:D_and_tilde_E}}{Lemma~\ref{lem:D_and_tilde_E}}}
	
Fix any value $v > 0$.
We would like to show that $q^{\tilde{E}_i}(v) \le q^{D_i}(v)$.
By \Cref{lem:empirical_error_bound}, we have:
\[
q^{E_i}(v) \le q^{D_i}(v) + \sqrt{\frac{2 q^{D_i}(v) \big( 1-q^{D_i}(v) \big) \ln(2mn\delta^{-1})}{m} } + \frac{\ln(2mn\delta^{-1})}{m}
\]
Of course, we also have that $q^{E_i}(v) \le 1$.

Further, by the definition of $\mathbf{\tilde{E}}$, we have $q^{\tilde{E}_i}(v) = \shading \big( q^{E_i}(v) \big)$.
Noting that $\shading$ is monotone, it suffices to show $q^{\tilde{E}_i}(v) \le q^{D_i}(v)$ when the above equation holds with equality.

To simply notation, we write $q^D = q^{D_i}(v)$, $q^E = q^{E_i}(v)$,  and $q^{\tilde{E}} = q^{\tilde{E}_i}(v)$, and let $c = \frac{\ln(2mn\delta^{-1})}{m}$ in the rest of the proof.
We would like to show that $q^{\tilde{E}} \le q^D$ if 
\begin{equation}
\label{eqn:D_and_E}
q^E = q^D + \sqrt{2cq^D(1-q^D)} + c \le 1
~,
\end{equation}
and 
\begin{equation}
\label{eqn:E_and_tilde_E}
q^{\tilde{E}} = \shading(q^E) = q^E - \sqrt{2cq^E(1-q^E)} - 4c
~.
\end{equation}

Rearranging terms $q^D$ and $c$ to the left-hand-side of \Cref{eqn:D_and_E} and squaring both sides, we get that it is equivalent to:
\[
\left( q^E-q^D-c \right)^2 = 2c q^D(1 - q^D)
~.
\]

Next, organize the terms as a quadratic equation of $q^D$:
\[
(1+2c) \left( q^D \right)^2 - 2 q^E q^D + \left( q^E - c \right)^2 = 0
~.
\]

Solving it and noting that $q^D \le q^E$, we get that:
\begin{align*}
(1 + 2c) q^D & = q^E - \sqrt{2c q^E (1 - q^E)^2 + 4c^2 q^E - (1+2c)c^2} \\
& \ge q^E - \sqrt{2c q^E (1 - q^E)^2 + 4c^2} && \text{($q^E \le 1$)} \\
& \ge q^E - \sqrt{2c q^E (1 - q^E)^2} - 2c 
~.
\end{align*}

Further, the left-hand-side is at most $q^D + 2c$ because $q^D \le 1$.
Putting together with \Cref{eqn:E_and_tilde_E}, we have that:
\[
q^D \ge q^E - \sqrt{2c q^E (1 - q^E)^2} - 4c = q^{\tilde{E}}
~. 
\]

So the lemma follows.

\subsection{Proof of \texorpdfstring{\Cref{lem:tilde_E_and_tilde_D}}{Lemma~\ref{lem:tilde_E_and_tilde_D}}}

Fix any value $v > 0$.
We would like to show that $q^{\tilde{E}_i}(v) \ge q^{\tilde{D}_i}(v)$.
By \Cref{lem:empirical_error_bound}, we have:
\[
q^{E_i}(v) \ge q^{D_i}(v) - \sqrt{\frac{2 q^{D_i}(v) \big( 1-q^{D_i}(v) \big) \ln(2mn\delta^{-1})}{m} } - \frac{\ln(2mn\delta^{-1})}{m}
\]
Of course, we also have that $q^{E_i}(v) \ge 0$.

Further, by the definition of $\mathbf{\tilde{E}}$, we have $q^{\tilde{E}_i}(v) = \shading \big( q^{E_i}(v) \big)$.
Noting that $\shading$ is monotone, it suffices to show $q^{\tilde{E}_i}(v) \ge q^{\tilde{D}_i}(v) = \doubleshading \big( q^{D_i}(v) \big)$ when the above equation holds with equality.

To simply notation, we write $q^D = q^{D_i}(v)$, $q^{\tilde{D}} = q^{\tilde{D}_i}(v)$, $q^E = q^{E_i}(v)$, and $q^{\tilde{E}} = q^{\tilde{E}_i}(v)$, and let $c = \frac{\ln(2mn\delta^{-1})}{m}$ in the rest of the proof.
We would like to show that $q^{\tilde{E}} \ge q^{\tilde{D}}$ if 
\begin{equation}
\label{eqn:D_and_E_2}
q^E = q^D - \sqrt{2cq^D(1-q^D)} - c \le 1
~,
\end{equation}
and 
\begin{equation}
\label{eqn:E_and_tilde_E_2}
q^{\tilde{E}} = \shading(q^E) = q^E - \sqrt{2cq^E(1-q^E)} - 4c
~,
\end{equation}
and 
\begin{equation}
\label{eqn:D_and_tilde_D}
q^{\tilde{D}} = \doubleshading(q^D) = q^D - \sqrt{8cq^D(1-q^D)} - 7c
~.
\end{equation}	

By \Cref{eqn:D_and_E_2} and \Cref{eqn:E_and_tilde_E_2}, we have:
\[
q^{\tilde{E}} = q^D - \sqrt{2cq^D(1-q^D)} - \sqrt{2cq^E(1-q^E)} - 5c 
~.
\]

Comparing with the right-hand-side of \Cref{eqn:D_and_tilde_D}, it suffices to show that:
\[
\sqrt{2cq^D(1-q^D)} + 2c \ge \sqrt{2cq^E(1-q^E)}
\]

This holds due to the following sequence of inequalities:
\begin{align*}
2cq^E(1-q^E) & = 2c q^D(1-q^D) + 2c (q^D-q^E)(2q^D-1) - 2c (q^D-q^E)^2 \\
& \le 2c q^D(1-q^D) + 2c (q^D-q^E) - 2c (q^D-q^E)^2 && \text{($q^D \le 1$)} \\
& \le 2c q^D(1-q^D) + 2c (q^D-q^E) \\
& = 2c q^D(1-q^D) + 2c \sqrt{2cq^D(1-q^D)} + 2c^2 && \text{(\Cref{eqn:D_and_E_2})} \\ 
& < \big( \sqrt{2cq^D(1-q^D)} + 2c \big)^2 
~.
\end{align*}

\subsection{Proof of \texorpdfstring{\Cref{lem:truncate_bottom}}{Lemma~\ref{lem:truncate_bottom}}}

Recall that $M_{\mathbf{D}}$ denotes Myerson's optimal auction w.r.t.\ $\mathbf{D}$. 
Consider the following mechanism $M$ for distribution $\truncatebottom_{\epsilon}(\mathbf{D})$:
\begin{enumerate}
	\item Given a value profile $\mathbf{v}$, define an auxiliary value profile $\mathbf{v'}$ such that for every $i \in [n]$, let $v'_i = v_i$ if $v_i > 0$, and let $v'_i$ be a random sample from $D_i$ \emph{conditioned on its quantile is between $1 - \epsilon$ and $1$}.
	\item Let $i^*$ and $p^*$ be the winner and her payment given by $M_{\mathbf{D}}$ on the auxiliary value profile $\mathbf{v'}$.
	\item Let $i^*$ be the winner of running $M$ on value profile $\mathbf{v}$.
	\item Let $i^*$ pay $p^*$ if $v_i > 0$, and $0$ otherwise.
\end{enumerate}

Let us compare the expected revenue from running $M$ on the true value profile $\mathbf{v}$ drawn from $\truncatebottom_{\epsilon}(\mathbf{D})$, and that from running $M_{\mathbf{D}}$ on the corresponding auxiliary value profile $\mathbf{v'}$, which follows $\mathbf{D}$.
We will do so by comparing the contribution from each bidder $i$ in these two cases.

Fix any bidder $i$ and any value profile $v_{-i}$ and, thus, the corresponding auxiliary value profile $v'_{-i}$ of the bidders other than $i$.
Since the auxiliary values of the other bidders are fixed, from $i$'s viewpoint, $M_{\mathbf{D}}$ sets a price $p^*$ and $i$ wins the item so long as her value is at least $p^*$.
Therefore, the expected payment by $i$ from running $M_{\mathbf{D}}$ on the auxiliary value profile $\mathbf{v'}$ is $q^{D_i}(p^*) \cdot p^*$;
that from running $M$ on the true value profile is $\mathbf{v}$ is:
$q^{\truncatebottom_{\epsilon}(D_i)}(p^*) p^* = \min \{1 - \epsilon, q^{D_i}(p^*)\} \cdot p^*$.
Note that they differ by at most a $1 - \epsilon$ multiplicative factor for any value of $q^{D_i}(p^*)$.

Hence, bidder $i$'s contribution to the expected revenue of $M$ on $\truncatebottom_{\epsilon}(\mathbf{D})$ is at least a $1 - \epsilon$ of her contribution to the optimal revenue of $\mathbf{D}$.
Summing over all bidders proves the lemma.

	\section{Upper Bounds without Information Theory}
\label{app:without_information_theory}

In this section, we will demonstrate how to prove weaker sample complexity upper bounds building only on revenue monotonicity, without the information theoretic argument.
Recall that what we need to prove is the following inequality when the number of samples $m$ is sufficiently large:
\[
\opt(\mathbf{\tilde{D}}) \ge (1 - \epsilon) \opt(\mathbf{D})
\]

To simplify the discussion, let us introduce a dummy bidder $n+1$ whose value is a point mass at $0$.
Then, we have $\tilde{D}_{n+1} = D_{n+1}$, and the virtual value of the bidder is always $0$.
As a result, we can assume without loss of generality that the item is always allocated to some bidder.

Further, recall that we assume for simplicity that the optimal mechanism $M_{\mathbf{D}}$ breaks ties over bidders with the same ironed virtual value in the lexicographical order.
Therefore, if two bidders $i < j$ have the same virtual value, we will still write $\bar{\phi}_i(v_i) > \bar{\phi}_j(v_j)$ in the following discussions.

The main technical lemma is the following, which states that $\mathbf{\tilde{D}}$ approximately preserves the density of $\mathbf{D}$ almost point-wise, except for the values with very small quantiles.

\begin{lemma}\label{lem:density_approximation}
	For any product distribution $\mathbf{D}$, suppose $\tilde{\mathbf{D}}=d_{m,n,\delta}(\mathbf{D})$, and $m$ is at least $\frac{4 \ln (2mn\delta^{-1})}{\epsilon^2}$.
	Then, for any $i \in [n]$, and any $v$ such that $q^{D_i}(v) \ge \Omega \big(\frac{\ln (2mn\delta^{-1})}{m\epsilon^2} \big)$, we have:
	\[
	1 + \frac{\epsilon}{2} \ge \frac{d\tilde{D}_i}{dD_i}(v) \ge 1 - \frac{\epsilon}{2} 
	~.
	\]
\end{lemma}

\begin{proof}
	Let $q = q^{D_i}(v)$ for simplicity in this proof.
	By $q \ge \Omega \big(\frac{\ln (2mn\delta^{-1})}{m\epsilon^2} \big)$ and $m \ge \frac{4 \ln (2mn\delta^{-1})}{\epsilon^2}$, we get that $\doubleshading(q)$ is continuous at $q$.
	Therefore, if $v$ is not a point mass, we have:
	\begin{align*}
	\frac{d\tilde{D}_i}{dD_i}(v) & = \doubleshading'(q) \\
	& = 1-\sqrt{\frac{8\ln (2mn\delta^{-1})}{m}}\frac{1-2q}{2\sqrt{(1-q)q}} 
	~.
	\end{align*}
	
	On the other hand, suppose $v$ is a point mass.
	Denote its probability mass as $p_v$.
	We have:
	\begin{align*}
	\frac{d\tilde{D}_i}{dD_i}(v)&= \frac{\doubleshading \big( q \big) - \doubleshading \big( q - p_v \big)}{p_v} \\
	& = \frac{\int_{q-p_v}^{q} \doubleshading'(x) dx}{p_v} \\[1ex]
	& \ge \doubleshading'(q) && \text{(monotonicity of $\doubleshading'$)} \\[1ex]
	& = 1-\sqrt{\frac{8\ln (2mn\delta^{-1})}{m}}\frac{1-2q}{2\sqrt{(1-q)q}} 
	~.
	\end{align*}

	Then, the lemma follows by $q \ge \Omega \big(\frac{\ln (2mn\delta^{-1})}{m\epsilon^2} \big)$.
\end{proof}

Next, we present in the next lemma a meta analysis that will be further specialized for each family of distributions.

\begin{lemma}\label{lem:quantile_truncate_opt_bound}
	Suppose $\mathbf{D}$ have a bounded support in $[0, \bar{v}_1] \times [0, \bar{v}_2] \times \dots \times [0, \bar{v}_n]$ such that the probability mass of $\bar{v}_i$ is at least $\Omega \big(\frac{\ln (2mn\delta^{-1})}{m\epsilon^2} \big)$ in $D_i$ for all $i \in [n]$, and $m$ is at least $\frac{4 \ln (2mn\delta^{-1})}{\epsilon^2}$.
	Then, we have:
	\[
	\opt (\Tilde{\mathbf{D}})\ge (1-\epsilon)\opt(\mathbf{D})
	~.
	\]
\end{lemma}

\begin{proof}
	%
%
	Consider the optimal mechanism $M_{\mathbf{D}}$ w.r.t.\ $\mathbf{D}$.
	We can write the expected revenue by enumerating all possible configuration of the highest bidder and the second highest bidder in terms of ironed virtual values.
	\begin{align*}
	\opt(\mathbf{D}) =& \sum_i \sum_{j \ne i} \iint_{\bar{\phi}_i(v_i) > \bar{\phi}_j(v_j)} \bar{\phi}_i^{-1} \big( \bar{\phi}_j(v_j) \big) \cdot \prod_{k \ne i,j} \Pr_{v_k \sim D_k} \left[ \bar{\phi}_k(v_k) < \bar{\phi}_j(v_j) \right] dD_idD_j
	~.
	\end{align*}
	
	Next, suppose we run the same mechanism $M_{\mathbf{D}}$ on the auxiliary distribution $\tilde{\mathbf{D}}$.
	Its expected revenue can be written in a similar fashion as:
	\begin{align*}
		\rev(M_{\mathbf{D}}, \mathbf{\tilde{D}})
		 = &
		\sum_i \sum_{j \ne i} \iint_{\bar{\phi}_i(v_i) > \bar{\phi}_j(v_j)} \bar{\phi}_i^{-1} \big( \bar{\phi}_j(v_j) \big) \cdot \prod_{k \ne i,j} \Pr_{v_k \sim \Tilde{D}_k} \left[ \bar{\phi}_k(v_k) < \bar{\phi}_j(v_j) \right] d \tilde{D}_i d\tilde{D}_j \\
		\ge &  
		\sum_i \sum_{j \ne i} \iint_{\bar{\phi}_i(v_i) > \bar{\phi}_j(v_j)} \bar{\phi}_i^{-1} \big( \bar{\phi}_j(v_j) \big) \cdot \prod_{k \ne i,j} \Pr_{v_k \sim D_k} \left[ \bar{\phi}_k(v_k) < \bar{\phi}_j(v_j) \right] d \tilde{D}_i d\tilde{D}_j \\
		\ge &
		\sum_i \sum_{j \ne i} \iint_{\bar{\phi}_i(v_i) > \bar{\phi}_j(v_j)} \bar{\phi}_i^{-1} \big( \bar{\phi}_j(v_j) \big) \cdot \prod_{k \ne i,j} \Pr_{v_k \sim D_k} \left[ \bar{\phi}_k(v_k) < \bar{\phi}_j(v_j) \right] (1 - \epsilon) d D_i d D_j \\[1ex]
		= & (1-\epsilon) \opt(\mathbf{D})
		~.
	\end{align*}
	
	Here, the first inequality is due to stochastic dominance, i.e., $D_k \succeq \tilde{D}_k$;
	the second inequality follows by \Cref{lem:density_approximation}.
	The lemma then follows by $\opt(\mathbf{\tilde{D}}) \ge \rev(M_{\mathbf{D}}, \mathbf{\tilde{D}})$.
\end{proof}

To simplify notations in the rest of the section, we define a function $\truncatetop_{\epsilon}$ that truncate the top $\epsilon$-fraction of the distribution.
For every bidder $i \in [n]$, let $\truncatetop_{\epsilon}(D_i)$ be the distribution obtained by truncating top $\epsilon$ fraction of values to the value whose original quantile is equal to $\epsilon$.
In other words, the quantiles of the truncated distribution is defined as:
\[
q^{\truncatetop_{\epsilon}(D_i)}(v) \defeq 
\begin{cases}
q^{D_i}(v) & \text{if $\epsilon\le q^{D_i}(v)\le 1$} \\
0 & \text{if $0\le q^{D_i}(v)<\epsilon$}
\end{cases}
\]

Further, for any product value distribution, define:
\[
\truncatetop_{\epsilon}(\mathbf{D}) = \truncatetop_{\epsilon}(D_1) \times \truncatetop_{\epsilon}(D_2) \times \dots \times \truncatetop_{\epsilon}(D_n)
~.
\]

\subsubsection*{Regular Distributions}

We first explain how to prove the optimal $\tilde{O}(n \epsilon^{-3})$ sample complexity upper bound without using information theory.
The main technical lemma is the following.

\begin{lemma}\label{lem:regular_truncate}
	For any regular product distribution $\mathbf{D}$, we have:
	\[
	\opt \big( \truncatetop_{\frac{\epsilon}{2n}}(\mathbf{D}) \big) \ge (1-\epsilon) \cdot \opt \big( \mathbf{D} \big)
	~.
	\]
\end{lemma}

The proof is similar to that of Lemma 4.6 in \citet{DevanurHP/2016/STOC}.
The main difference is that \citet{DevanurHP/2016/STOC} truncate in the value space while here we truncate in the quantile space.

\begin{proof}
	For simplicity of notations, we extend the definition of virtual values to the quantile space and abuse notation in letting $\phi_i(q_i)$ and $\hat{\phi}_i(q_i)$ denote the virtual value of the $i$-th bidder when her value has quantile $q_i$ in $D_i$ and $\truncatetop_{\frac{\epsilon}{2n}}(D_i)$ respectively.
	Let the truncation point with quantile $\frac{\epsilon}{2n}$ in distribution $D_i$ be $\bar{v}_i$. 
	Given any $n$-dimensional quantile vector $\mathbf{q}$, let $H(\mathbf{q})=\{i\in[n],q_i\le\frac{\epsilon}{2n}\}$, and $L(\mathbf{q})=\{i\in[n],q_i>\frac{\epsilon}{2n}\}$. 
	
	First, consider a quantile vector $\mathbf{q}\in [0,1]^n \setminus [\frac{\epsilon}{2n},1]^n$. 
	We claim that:
	\begin{equation}
	\label{eqn:quantile_truncate_regular_1}
	\int_{[0,1]^n \setminus [\frac{\epsilon}{2n},1]^n} \max_i \hat{\phi}_i(q_i)d\mathbf{q} \ge ( 1 - \epsilon ) \int_{[0,1]^n}\max_{i\in H(\mathbf{q})}\phi_i(q_i)d \mathbf{q} 
	\end{equation}
	
	On one hand, left-hand-side of \eqref{eqn:quantile_truncate_regular_1} is lower bounded by:
	\begin{align*}
	\int_{[0,1]^n \setminus [\frac{\epsilon}{2n},1]^n} \max_i\bar{\phi}_i(q_i) d\mathbf{q}
	& \ge \sum_j \int_{q_j \in [0, \frac{\epsilon}{2n}], \mathbf{q}_{-j} \in [\frac{\epsilon}{2n},1]^{n-1}} \max_i\bar{\phi}_i(q_i) d\mathbf{q} \\
	& \ge \sum_j \int_{q_j \in [0, \frac{\epsilon}{2n}], \mathbf{q}_{-j} \in [\frac{\epsilon}{2n},1]^{n-1}} v_j d\mathbf{q} & \text{($\hat{\phi}_j(q_j) = v_j$)} \\
	& = \frac{\epsilon}{2n} \left( 1 - \frac{\epsilon}{2n} \right)^{n-1} \sum_j v_j \\
	& \ge \frac{\epsilon}{2n} \left( 1 - \frac{\epsilon}{2} \right) \sum_j v_j 
	~.
	\end{align*}
	
	On the other hand, the right-hand-side of \eqref{eqn:quantile_truncate_regular_1}, omitting the $1 - \frac{\epsilon}{2}$ factor, is upper bounded by:
	\begin{align*}
	\int_{[0,1]^n} \max_{i\in H(\mathbf{q})}\phi_i(q_i)d \mathbf{q}	
	& \le 
	\int_{[0,1]^n} \sum_{i\in H(\mathbf{q})} \phi_i(q_i)d \mathbf{q} \\
	& = 
	\sum_{i \in [n]} \int_{[0,\frac{\epsilon}{n}]} \max \{\phi_i(q_i), 0\} dq_i
	~.
	\end{align*}

	Thus, it suffices to show that $\forall i\in[n]$, we have:
	\[
	\frac{\epsilon}{2n} v_i \ge \left( 1 - \frac{\epsilon}{2} \right) \int_{[0,\frac{\epsilon}{2n}]} \max\{\phi_i(q_i),0\}dq_i
	~.
	\]
	
	The left-hand-side, $\frac{\epsilon}{2n} v_i$, is exactly the expected revenue of price $ v_i$, which has quantile $\frac{\epsilon}{2n}$.
	That is, we have:
	\[
	\frac{\epsilon}{2n} v_i = R^{D_i} \left( \frac{\epsilon}{2n} \right)
	~.
	\]
	
	The right-hand-side without the $1 - \frac{\epsilon}{2}$ factor, $\int_{[0,\frac{\epsilon}{2n}]} \max\{\phi_i(q_i),0\}dq_i$, is the maximum expected revenue subject to having a sale probability at most $\frac{\epsilon}{2n}$.
	That is, we have:
	\[
	\int_{[0,\frac{\epsilon}{2n}]} \max\{\phi_i(q_i),0\} dq_i = \max_{0 \le q \le \frac{\epsilon}{2n}} R^{D_i}(q)
	~.
	\]
		 
	Suppose $q_i^*$ is the quantile of the monopoly price if there is a single bidder with value distribution $D_i$.
	Then, if $q_i^* \ge \frac{\epsilon}{2n}$, by concavity of revenue curve of regular distributions, it must be increasing from $0$ to $q_i^*$.
	In particular, we have $\max_{0 \le q \le \frac{\epsilon}{2n}} R^{D_i}(q) = R^{D_i}(\frac{\epsilon}{2n})$.
	
	On the other hand, suppose $q_i^* \le \frac{\epsilon}{2n}$.
	Then, by the concavity of revenue curve of regular distributions, we have:
	\begin{align*}
	R^{D_i} \left( \frac{\epsilon}{2n} \right) & \ge \frac{1 - \frac{\epsilon}{2n}}{1 - q_i^*} R^{D_i}(q_i^*) + \frac{\frac{\epsilon}{2n} - q^*}{1 - q_i^*} R^{D_i}(1) \\
	& \ge \left(1 - \frac{\epsilon}{2n}\right) R^{D_i}(q_i^*) \\
	& = \left(1 - \frac{\epsilon}{2n}\right)\max_{0 \le q \le \frac{\epsilon}{2n}} R^{D_i}(q) 
	~.
	\end{align*}
	
	Putting together proves \eqref{eqn:quantile_truncate_regular_1}.
	
	\bigskip

	Next, we will show that:
	\begin{equation}\label{eqn:quantile_truncate_regular_2}
	\int_{[\frac{\epsilon}{2n},1]^n}\max_i \hat{\phi}_i(q_i)d\mathbf{q}
	\ge 
	\left( 1 - \epsilon \right) \int_{[0,1]^n}\max_{i\in L(\mathbf{q})}\phi_i(q_i)d \mathbf{q}
	\end{equation}
	
	The left-hand-side of \eqref{eqn:quantile_truncate_regular_2} can be rewritten as follows because $D_i$ and $\truncatetop_{\frac{\epsilon}{2n}}(D_i)$ are identical, and $L(\mathbf{q}) = [n]$, for quantiles $\mathbf{q}$ greater than $\frac{\epsilon}{2n}$:
	\[
	\int_{[\frac{\epsilon}{2n},1]^n} \max_{i} \phi_i(q_i) d\mathbf{q}
	=
	\int_{[\frac{\epsilon}{2n},1]^n} \max_{i \in L(\mathbf{q})} \phi_i(q_i) d\mathbf{q}
	~.
	\]
	 
	Further, note that for any $i \in [n]$, any $\mathbf{q} = (q_i,\mathbf{q}_{-i})$, and any $\mathbf{q'} = (q_i',\mathbf{q}_{-i})$ where $q_i < \frac{\epsilon}{2n} \le q_i'$, we have $L(\mathbf{q})\subseteq L(\mathbf{q'})$. 
	Thus, $\max_{i \in L(\mathbf{q})}\phi_i(q_i)\le \max_{i\in L(\mathbf{q'})}\phi_i(q_i)$. 
	Therefore, Eqn.~\eqref{eqn:quantile_truncate_regular_2} follows by a hybrid argument as follows:
	\begin{align*}
		\int_{[\frac{\epsilon}{2n},1]^n}\max_{i\in L(\mathbf{q})}\phi_i(q_i)d\mathbf{q}
		& \ge \left(1-\frac{\epsilon}{2n}\right) \int_{[0,1]\times[\frac{\epsilon}{2n},1]^{n-1}}\max_{i\in L(\mathbf{q})}\phi_i(q_i) d\mathbf{q} \\
		& \ge \left(1-\frac{\epsilon}{2n}\right)^2 \int_{[0,1]^2\times[\frac{\epsilon}{2n},1]^{n-2}}\max_{i\in L(\mathbf{q})}\phi_i(q_i) d\mathbf{q} \\[1ex]
		& \ge \dots \\[1ex]
		& \ge \left( 1-\frac{\epsilon}{2n} \right)^n \int_{[0,1]^n}\max_{i\in L(\mathbf{q})}\phi_i(q_i) d\mathbf{q} \\
		& \ge \left( 1 - \frac{\epsilon}{2} \right)\int_{[0,1]^n}\max_{i\in L(\mathbf{q})}\phi_i(q_i) d\mathbf{q}
		~.
	\end{align*}
	
	Summing \eqref{eqn:quantile_truncate_regular_1} and \eqref{eqn:quantile_truncate_regular_2} give:
	\begin{align*}
	\int_{[0,1]^n} \max_i \hat{\phi}_i(q_i)d\mathbf{q} 
	& \ge 
	( 1 - \epsilon ) \left( \int_{[0,1]^n} \max_{i\in H(\mathbf{q})} \phi_i(q_i) + \max_{i \in L(\mathbf{q}))} \phi_i(q_i) \right) d \mathbf{q} \\
	& \ge 
	( 1 - \epsilon ) \int_{[0,1]^n} \max_{i} \phi_i(q_i) d \mathbf{q} 
	~.
	\end{align*}
	
	The left-hand-side and the right-hand-side are precisely the optimal revenue of distributions $\truncatetop_{\frac{\epsilon}{2n}}(\mathbf{D}) $ and $\mathbf{D}$ respectively.
\end{proof}

	We now finish the analysis for regular distributions with the following sequence of inequalities:
	\begin{align*}
        \rev(M_{\Tilde{\mathbf{E}}},\mathbf{D}) & \ge \opt(\mathbf{\tilde{D}}) && \text{(\Cref{lem:tilde_E_and_tilde_D})} \\[1.5ex]
        & = \opt \left( \doubleshading(\mathbf{D}) \right) \\[0.5ex]
	& \ge \opt \left( \doubleshading \circ \truncatetop_{\frac{\epsilon}{2n}} (\mathbf{D}) \right) && \text{(weak revenue monotonicity, i.e., \Cref{lem:weak_revenue_monotonicity})} \\ 
	& \ge (1 - \epsilon) \opt \left( \truncatetop_{\frac{\epsilon}{2n}} (\mathbf{D}) \right) && \text{(\Cref{lem:quantile_truncate_opt_bound}, and $m \ge \tilde{O}(n \epsilon^{-3})$)} \\
	& \ge (1 - 2 \epsilon) \opt(\mathbf{D}) && \text{(\Cref{lem:regular_truncate})}
	\end{align*}

\subsubsection*{MHR Distributions, Weaker Upper Bound}

Note that the upper bound of $\tilde{O}(n \epsilon^{-3})$ also holds for MHR distributions since all MHR distributions are regular.

\subsubsection*{[1, H]-Bounded Support Distributions, Weaker Upper Bound}

Rather than truncating at quantile $\frac{\epsilon}{2n}$, we will pick the truncation point to be $\frac{\epsilon}{Hn}$ to bound the contribution by large values up to $H$.
We first show the following lemma.

\begin{lemma}\label{lem:one_H_truncate}
	For any $[1, H]$-bounded support distribution $\mathbf{D}$, we have:
	\[
	\opt \left( \truncatetop_{\frac{\epsilon}{Hn}}(\mathbf{D}) \right)\ge (1-\epsilon)\opt(D)
	~.
	\]
\end{lemma}

\begin{proof}
	Since the values are upper bounded by $H$, we have:
	\begin{align*}
	\opt(\mathbf{D}) - \opt \left( \truncatetop_{\frac{\epsilon}{Hn}}(\mathbf{D}) \right) & \le H \cdot \Pr \left[ \exists i,\ q(v_i)\le \frac{\epsilon}{Hn}\right] \\
	& \le H \cdot \sum_{i = 1}^n\Pr \left[ q(v_i)\le \frac{\epsilon}{Hn}\right] \\
	& = H \cdot \sum_{i = 1}^n \frac{\epsilon}{Hn} \\
	& = \epsilon 
	~.
	\end{align*}
	
	The lemma now follows by $\opt(\mathbf{D}) \ge 1$ because the values are lower bounded by $1$.
\end{proof}

The rest of the analysis is similar to the regular case.
If the number of samples $m$ is at least $\tilde{O}(n H \epsilon^{-3})$, we have the following sequence of inequalities:
\begin{align*}
    \rev(M_{\Tilde{\mathbf{E}}},\mathbf{D}) & \ge \opt(\mathbf{\tilde{D}}) && \text{(\Cref{lem:tilde_E_and_tilde_D})} \\[1.5ex]
    & = \opt \left( \doubleshading(\mathbf{D}) \right) \\[0.5ex]
& \ge \opt \left( \doubleshading \circ \truncatetop_{\frac{\epsilon}{Hn}} (\mathbf{D}) \right) && \text{(weak revenue monotonicity, i.e., \Cref{lem:weak_revenue_monotonicity})} \\ 
& \ge (1 - \epsilon) \opt \left( \truncatetop_{\frac{\epsilon}{Hn}} (\mathbf{D}) \right) && \text{(\Cref{lem:quantile_truncate_opt_bound}, and $m \ge \tilde{O}(n H \epsilon^{-3})$)} \\
& \ge (1 - 2 \epsilon) \opt(\mathbf{D}) && \text{(\Cref{lem:one_H_truncate})}
\end{align*}
	

\subsubsection*{$[0, 1]$-Bounded Support Distributions, Weaker Upper Bound}

Similar to the previous cases, we first establish a lemma that bound the revenue loss due to the truncation of small quantiles.

\begin{lemma}\label{lem:zero_one_truncate}
	For any $[0, 1]$-bounded support distribution $\mathbf{D}$, we have:
	\[
	\opt \left( \truncatetop_{\frac{\epsilon}{n}}(\mathbf{D}) \right) \ge \opt(D) - \epsilon
	~.
	\]
\end{lemma}

The proof is almost verbatim to that of \Cref{lem:one_H_truncate}.
We include it for completeness.

\begin{proof}
	Since the values are upper bounded by $1$, we have:
	\begin{align*}
	\opt(\mathbf{D}) - \opt \left( \truncatetop_{\frac{\epsilon}{n}}(\mathbf{D}) \right) & \le \Pr\left[ \exists i, q^{D_i}(v_i)\le \frac{\epsilon}{n}\right] \\
	& \le \sum_{i = 1}^n \Pr \left[ q^{D_i}(v_i)\le \frac{\epsilon}{n} \right]  \\
	& = \sum_{i = 1}^n \frac{\epsilon}{n} = \epsilon 
	~.
	\end{align*}
\end{proof}

The rest of the analysis is also almost verbatim to the previous case, changing the parameters appropriately and replacing \Cref{lem:one_H_truncate} with \Cref{lem:zero_one_truncate}.
Concretely, if the number of samples $m$ is at least $\tilde{O}(n \epsilon^{-3})$, we have the following sequence of inequalities:
\begin{align*}
    \rev(M_{\Tilde{\mathbf{E}}},\mathbf{D}) & \ge \opt(\mathbf{\tilde{D}}) && \text{(\Cref{lem:tilde_E_and_tilde_D})} \\[1.5ex]
    & = \opt \left( \doubleshading(\mathbf{D}) \right) \\[0.5ex]
& \ge \opt \left( \doubleshading \circ \truncatetop_{\frac{\epsilon}{n}} (\mathbf{D}) \right) && \text{(weak revenue monotonicity, i.e., \Cref{lem:weak_revenue_monotonicity})} \\ 
& \ge (1 - \epsilon) \opt \left( \truncatetop_{\frac{\epsilon}{n}} (\mathbf{D}) \right) && \text{(\Cref{lem:quantile_truncate_opt_bound}, and $m \ge \tilde{O}(n \epsilon^{-3})$)} \\
& \ge (1 - 2 \epsilon) \opt(\mathbf{D}) && \text{(\Cref{lem:zero_one_truncate})}
\end{align*}

	\section{Optimality in the Single-bidder Case}
\label{app:single_ub}

Here we present a brief discussion on the optimality of the dominated empirical algorithm proposed in this paper in the special case with only one bidder.
In this case, any truthful mechanism is effectively offering a take-it-or-leave-it price.
All four families of distributions have been studied in this special case and the tight sample complexity is known for all of them up to a logarithmic factor.
We summarize the tight sample complexity from previous works in \Cref{tab:single-bidder}.

\begin{table}[h]
	\centering
	\begin{tabular}{|c|c|}
		\hline
		Setting & Sample Complexity \\
		\hline
		Regular & $\tilde{\Theta}(\epsilon^{-3})$ \cite{DhangwatnotaiRY/2015/GEB} \\
		\hline
		(Continuous) MHR & $\tilde{\Theta}(\epsilon^{-1.5})$ \cite{HuangMR/2015/EC} \\
		\hline
		$[1,H]$ & $\tilde{\Theta}(H \epsilon^{-2})$ \cite{HuangMR/2015/EC} \\
		\hline
		$[0,1]$ & $\tilde{\Theta}(\epsilon^{-2})$ \cite{HuangMR/2015/EC}\footnotemark[1] \\
		\hline 
		\multicolumn{2}{p{3.5in}}
		{\footnotesize\footnotemark[1] This bound is not explicitly stated in \cite{HuangMR/2015/EC} but follows straightforwardly from the analysis framework therein.}
	\end{tabular}
	\caption{Summary of Sample Complexity with a Single Bidder}
	\label{tab:single-bidder}
\end{table}

Let $n = 1$ in \Cref{tab:multi-bidder-results} compare with the bounds in \Cref{tab:single-bidder}.
Our analysis matches the tight sample complexity for $[1, H]$-bounded support distributions, $[0,1]$-bounded support distributions, and regular distributions.
For MHR distributions, our upper bound in \Cref{tab:multi-bidder-results} applies to both continuous and discrete MHR distributions while that by \citet{HuangMR/2015/EC} considers only continuous ones. 
Next, we will show in \Cref{app:discrete_MHR_lb} that our bound is in fact tight up to a logarithmic factor for discrete MHR distributions. 
Finally, we will present in \Cref{app:continuous_MHR_ub} a more specialized analysis for continuous MHR distributions, incorporating the techniques by \citet{HuangMR/2015/EC}, to show that our algorithm also achieves the optimal sample complexity in \Cref{tab:single-bidder} for continuous MHR distributions.

\subsection{Lower Bound for Discrete MHR Distributions}
\label{app:discrete_MHR_lb}

We will show the following sample complexity lower bound for discrete MHR distributions when there is a single bidder, using the framework by \citet{HuangMR/2015/EC}.

\begin{lemma}[\citet{HuangMR/2015/EC}, Theorem 4.2]
	\label{lem:single_bidder_lb}
	If two distributions $D_1$ and $D_2$ have disjoint $(1 - 3
	\epsilon)$-approximate price sets, and there is a pricing algorithm
	that is $(1 - \epsilon)$-approximate for both $D_1$ and $D_2$, then
	the algorithm uses at least $\Omega(\skl( D_1, D_2 )^{-1})$ samples.
\end{lemma}

\begin{theorem}
	\label{thm:discrete_MHR_lb}
	Suppose an algorithm guarantees a $(1 - \epsilon)$-approximation with high probability for all discrete MHR distributions using $m$ samples.
	Then, $m$ must be at least $\Omega(\epsilon^{-2})$.
\end{theorem}

\begin{proof}
	Let $D_1$ and $D_2$ be two distributions with support $\{ 1, 2 \}$.
	$D_1$ takes value $1$ with probability $\tfrac{1 + 4 \epsilon}{2}$, and $2$ with probability $\tfrac{1 - 4 \epsilon}{2}$.
	$D_2$ takes value $1$ with probability $\tfrac{1 - 4 \epsilon}{2}$, and $2$ with probability $\tfrac{1 + 4 \epsilon}{2}$.
	The KL divergence of the distributions is bounded by:
	\[ 
	\skl(D_1, D_2) 
	= 
	2 \cdot \left( \tfrac{1+4\epsilon}{2} \ln
	\tfrac{1+4\epsilon}{1-4\epsilon} + \tfrac{1-4\epsilon}{2} \ln
	\tfrac{1-4\epsilon}{1+4\epsilon} \right) = 8\epsilon \ln
	\tfrac{1+4\epsilon}{1-4\epsilon} = \Theta(\epsilon^2).
	\]
	
	On one hand, any price that is at least a $(1 - 3\epsilon)$-approximation for $D_1$ must be at most $1$ because the optimal is $1$, which is achieved when the price is $1$, but setting a price larger than $1$ gets at most $1 - 4 \epsilon$, which is achieved when the price is $2$.
	On the other hand, any price that is at least a $(1 - 3\epsilon)$-approximation for $D_2$ must be greater than $1$ because the optimal is $1 + 4\epsilon$, which is achieved when the price is $2$, but setting a price at most $1$ gets at most $1$, which is achieved when the price is $1$. 
	Hence, $D_1$ and $D_2$ have disjoint $(1-3\epsilon)$-approximate
	price sets.
	
	Therefore, the claim follows from \Cref{lem:single_bidder_lb}.
\end{proof}

\subsection{Improved Upper Bound Continuous MHR Distributions}
\label{app:continuous_MHR_ub} 

We will show the following improved sample complexity upper bound for \Cref{alg:dominated_empirical_myerson}.

\begin{theorem}
	\label{thm:single_MHR_upper_bound}
	For any $0 < \epsilon < 1$ and any continuous MHR distribution $D$, suppose $m$ is at least $\tilde{O}(\epsilon^{-1.5})$.
	Then, \Cref{alg:dominated_empirical_myerson} returns a mechanism with an expected revenue at least $(1 - \epsilon) \opt(D)$, with high probability.
\end{theorem}

For any MHR distribution $D$, let $R(q)=(q^{D})^{-1}(q)\cdot q$ denote the revenue of setting a reserve price with quantile $q$. Let $v^*$ denote the optimal reserve price and let $q^*$ denote the corresponding quantile.
The following lemma by \citet{HuangMR/2015/EC} is crucial for the getting the improved sample complexity upper bound.

\begin{lemma}[\citet{HuangMR/2015/EC}, Lemma 3.3]
	\label{lem:MHR_single_quadratic_R}
	For any MHR distribution $D$, any $0 \le q \le 1$:
	\[
	R(q^*)-R(q)\ge \frac{1}{4}(q^* - q)^2R(q^*)
	~.
	\]
\end{lemma}

We will also use the following well known fact about continuous MHR distributions.

\begin{lemma}[e.g., \citet{HuangMR/2015/EC}, Lemma 2.1]
	\label{lem:MHR_optimal_quantile_bound}
	For any MHR distribution $D$, we have:
	\[
	q^* \ge \frac{1}{e}
	~.
	\]
\end{lemma}

The lemmas show two properties of single-bidder auction/pricing when the value distribution is MHR.
\Cref{lem:MHR_single_quadratic_R} shows that the gap between the expected revenue of a quantile and the optimal revenue is lower bounded by a quadratic function of the corresponding quantile gap, which should be maintained with high probability with enough samples. 
\Cref{lem:MHR_optimal_quantile_bound} asserts that small quantiles cannot be optimal. 

To prove \Cref{thm:single_MHR_upper_bound}, we need to show that with high probability the algorithm will not pick any price with expected revenue less than $(1 - \epsilon)\opt(D)$.
We will divide into two subcases and analyze separately.
First, suppose the quantile of the price is at least $\frac{1}{9e}$.
Then, we will make use of \Cref{lem:MHR_single_quadratic_R} and present in \Cref{lem:single_bidder_MHR_large_quantile} an analysis similar to that by \citet{HuangMR/2015/EC}, incorporating the fact that we pick the price based on the dominated empirical distribution instead of the empirical distribution.
Otherwise, suppose the quantile of the price is less than $\frac{1}{9e}$, i.e., much smaller than the quantile of the optimal price according to \Cref{lem:MHR_optimal_quantile_bound}.
The reason why we need to handle this case differently is that the dominated empirical distribution may have large estimation errors for small quantiles, specifically, those that are at most $O(\frac{1}{m})$.
We analyze this case in \Cref{lem:single_bidder_MHR_small_quantile}.

\begin{lemma}
	\label{lem:single_bidder_MHR_large_quantile}
	Suppose a sample value $v$ has quantile $q \ge \frac{1}{9e}$ and $R(q) < (1 - \epsilon) R(q^*)$, and $m$ is at least $\tilde{O}(\epsilon^{-1.5})$.
	Then, the probability that \Cref{alg:dominated_empirical_myerson} picks $v$ over $v^*$ is at most $\frac{\delta}{m}$.
\end{lemma}

\begin{proof}
	We will prove the lemma for $q > q^*$.
	The other case when $\frac{1}{9e} \le q \le q^*$ is similar.

	Let $\delta_q = q - q^*$, and let $\Delta$ be such that $R(q) = (1-\Delta) R(q^*)$.	
	Then, by the assumption on $q$ in the lemma, we have:
	\begin{equation}
	\label{eqn:single_bidder_MHR_large_quantile_00}
	\Delta \ge \epsilon
	~.
	\end{equation}
	
	Further, by \Cref{lem:MHR_single_quadratic_R}, we have:
	\begin{equation}
	\label{eqn:single_bidder_MHR_large_quantile_0}
	\Delta \ge \frac{1}{4} \delta_q^2
	~.
	\end{equation}
	
	By the definition of the revenue curve $R(q)$, we also have:
	\begin{equation}
	\label{eqn:single_bidder_MHR_large_quantile_1}
	v \cdot q = (1 - \Delta) \cdot v^* \cdot q^*
	~.
	\end{equation}
	
	Let $\bar{q}^*$ and $\bar{q}$ denote the empirical quantiles of $v^*$ and $v$.
	Noting that $q^*$ and $q$ are both $\Theta(1)$, by Bernstein inequality (\Cref{lem:bernstein}), so are $\bar{q}^*$ and $\bar{q}$ and, subsequently, $\shading(\bar{q}^*)$ and $\shading(\bar{q})$, with high probability. 
	In particular, by the definition of $\shading$ with $n = 1$, we have:
	\begin{equation}
	\label{eqn:single_bidder_MHR_large_quantile_2}	
	\big| q - \shading(\bar{q}) \big| \le O \left( \sqrt{\frac{\ln(2 m \delta^{-1})}{m}} \right)
	~.
	\end{equation}
	
	Further, we have (with $n = 1$):
	\begin{align}
		\big( \bar{q} - \bar{q}^* \big) - \big( \shading(\bar{q}) - \shading(\bar{q}^*) \big) & 
		= \sqrt{\frac{2 \ln(2 m\delta^{-1})}{m}} \left( \sqrt{\bar{q}(1 - \bar{q})} - \sqrt{\bar{q}^*(1 - \bar{q}^*)} \right) 
		\notag \\
		& = \sqrt{\frac{2 \ln(2 m\delta^{-1})}{m}} \frac{(\bar{q} - \bar{q}^*)(1 - \bar{q} - \bar{q}^*)}{\sqrt{\bar{q}(1 - \bar{q})} + \sqrt{\bar{q}^*(1 - \bar{q}^*)}} 
		\notag \\
		& = O \left( \sqrt{\frac{\ln(2 m \delta^{-1})}{m}} \right) \cdot (\bar{q} - \bar{q}^*) && \text{($\bar{q}, \bar{q}^* = \Theta(1)$)} 
		\notag \\
		& = O \left( \sqrt{\frac{\ln(2 m \delta^{-1})}{m}} \right) \cdot \delta_q ~. \label{eqn:single_bidder_MHR_large_quantile_3}
	\end{align}
	
	Next, consider the event that the algorithm picks $v$ over $v^*$.
	By the definition of the algorithm, that means:
	\begin{equation}
	\label{eqn:single_bidder_MHR_large_quantile_4}
	v \cdot \shading(\bar{q}) \ge v^* \cdot \shading(\bar{q}^*)
	~.
	\end{equation}
	
	We claim that it means there must be a much bigger gap between dominated empirical quantiles $\bar{q}$ and $\bar{q}^*$ than that between the actual quantiles.
	This is formalized with the following sequence of inequalities:
	\begin{align*}
		\bar{q} - \bar{q}^* + O \left( \sqrt{\frac{\ln(2 m \delta^{-1})}{m}} \right) \cdot \delta_q & \ge \shading(\bar{q}) - \shading(\bar{q}^*) && \text{(Eqn.~\eqref{eqn:single_bidder_MHR_large_quantile_3})} \\
		& = \shading(\bar{q}) \left( 1 - \frac{\shading(\bar{q}^*)}{\shading(\bar{q})} \right) \\[1ex]
		& \ge \shading(\bar{q}) \left( 1 - \frac{v}{v^*} \right) && \text{(Eqn.~\eqref{eqn:single_bidder_MHR_large_quantile_4})} \\[1ex]
		& = \shading(\bar{q}) \left( 1 - (1 - \Delta) \cdot \frac{q^*}{q} \right) && \text{(Eqn.~\eqref{eqn:single_bidder_MHR_large_quantile_1})} \\
		& \ge \left( q - O \left( \sqrt{\frac{\ln(2 m \delta^{-1})}{m}} \right) \right) \left( 1 - (1 - \Delta) \cdot \frac{q^*}{q} \right) && \text{(Eqn.~\eqref{eqn:single_bidder_MHR_large_quantile_2})} \\[1ex]
		& \ge q - q^* - O \left( \sqrt{\frac{\ln(2 m \delta^{-1})}{m}} \right) \cdot \delta_q + \Omega(\Delta) && \text{($q^*, q = \Theta(1)$)}
		~.
	\end{align*}
	
	Further note that by Eqn.~\eqref{eqn:single_bidder_MHR_large_quantile_00}, Eqn.~\eqref{eqn:single_bidder_MHR_large_quantile_0}, and that $m \ge \tilde{O}(\epsilon^{-1.5})$, we have:
	\[
	\Delta \ge O \left( \sqrt{\frac{\ln(2 m \delta^{-1})}{m}} \right) \cdot \delta_q
	\]
	
	Hence, we get that:
	\[
	\bar{q} - \bar{q}^* \ge q - q^* + \Omega(\Delta)
	~.
	\]
	
	It remains to upper bound the probability that the above inequality holds.
	Note that it means that the fraction of samples that fall betwen $q^*$ and $q$ is larger than its expectation $\delta_q$ by at least $\Omega(\Delta)$. 
	By Bernstein inequality (\Cref{lem:bernstein}), we have:
	\[
	\Pr \left[ \bar{q} - \bar{q}^* \ge q - q^* + \Omega(\Delta) \right] \le  e^{- \Omega (\Delta^2 m \delta_q^{-1})} \le \frac{\delta}{m}
	~,
	\]
	where the last inequality holds because $\Delta^2 \delta_q^{-1} \ge \Delta^{-1.5} \ge \epsilon^{1.5}$ (due to Eqn.~\eqref{eqn:single_bidder_MHR_large_quantile_00} and Eqn.~\eqref{eqn:single_bidder_MHR_large_quantile_0})
	and that $m \ge \tilde{O}(\epsilon^{-1.5})$.
\end{proof}

\begin{lemma}
	\label{lem:single_bidder_MHR_small_quantile}
	Suppose a sample value $v$ has quantile $q < \frac{1}{9e}$, and $m$ is at least $\tilde{O}(\epsilon^{-1.5})$.
	Then, the probability that \Cref{alg:dominated_empirical_myerson} picks $v$ over $v^*$ is at most $\frac{\delta}{m}$.
\end{lemma}

\begin{proof}
	Note that in this lemma we no longer need the assumption that $R(q) < (1 - \epsilon) R(q^*)$ because it holds for all value with quantile less than $\frac{1}{9e}$ due to \Cref{lem:MHR_single_quadratic_R} and that $q < \frac{1}{9e}$ and $q^* \ge \frac{1}{e}$ (\Cref{lem:MHR_optimal_quantile_bound}).
	In fact, the lemma asserts something stronger.
	For some constant $c < 1$, we have:
	\[
	v \cdot q = R(q) \le c R(q^*) = c \cdot v^* \cdot q^*
	\]
	
	Let $\bar{q}^*$ and $\bar{q}$ denote the empirical quantiles of $v^*$ and $v$.
	By Bernstein's inequality, with probability at least $1 - \frac{\delta}{2m}$, we have:
	\[
	\bar{q}^* \ge q - \tilde{O} \left( \frac{1}{\sqrt{m}} \right)
	~.
	\]
	
	Further, by the definition of $\shading$ with $n = 1$, we have:
	\[
	\shading(\bar{q}^*) \ge \bar{q}^* - \tilde{O} \left( \frac{1}{\sqrt{m}} \right)
	\]	
	
	Putting together, we get that:
	\[
	\shading(\bar{q}^*) > c \cdot q^*
	\]
	
	On the other hand, by \Cref{lem:D_and_tilde_E} we get that with probability, $\frac{\delta}{2m}$ we have $D \succeq \tilde{E}$ and, in particular, we have:
	\[
	\shading(\bar{q}) \le q
	~.
	\]
	
	Therefore, taking a union bound over the above two events, the revenue of $v$ and $v^*$ in w.r.t.\ the dominated empirical distribution satisfy the following with probability at least $1 - \delta$:
	\[
	v \cdot \shading(\bar{q}) \le v \cdot q < c \cdot v^* \cdot q^* < v^* \cdot \shading(\bar{q}^*)
	~.
	\]
	
	This means that the algorithm will not pick $v$ over $v^*$. 
	So the lemma follows.
\end{proof}

Given the lemmas, we can now get the sample complexity bound in \Cref{thm:single_MHR_upper_bound} by taking a union bound over all $m$ samples.

	\section{Lower Bound for Continuous MHR Distributions}
\label{app:continuous_mhr_lb}

The proof of the sample complexity lower bound for MHR distributions in \Cref{sec:lower_bound_multi_bidder} makes use of discrete MHR distributions in the construction of the hard instance.
If we insist on using MHR distributions whose supports are continuous intervals, our meta lower bound framework still gives the following weaker lower bound of $\Omega(n \epsilon^{3/2})$, which nonetheless improves the previous best known bound by \citet{ColeR/2014/STOC}.
We stress that our sample complexity upper bound hold for both discrete and continuous MHR distributions.

We will follow the meta analysis in \Cref{sec:lower_bound_multi_bidder} and proceed by constructing three distribution $D^b$, $D^h$, and $D^\ell$ and verify the conditions listed in \Cref{sec:lower_bound_multi_bidder}. 
Let $\epsilon_0 = \epsilon \ln n$.
Let $v_0 = \ln n - 1 + \sqrt{\epsilon_0}$, $v_1 = \ln n$, $v_2 = \ln \frac{n}{1+\sqrt{\epsilon_0}} \approx \ln n - \sqrt{\epsilon_0}$, $p = \frac{2 \sqrt{\epsilon_0}}{n}$, and $\Delta = \sqrt{\epsilon_0}$ be the parameters.
Let $D^b$ be a singleton at $v_0 = \ln n - 1 + \sqrt{\epsilon_0}$.
Let the probability mass at $v_1$ be $\frac{1}{n} e^{-3\sqrt{\epsilon_0}(v_1 - v_2)} \approx \frac{1 - 3 \epsilon_0}{n}$ in $D^h$ and $\frac{1}{n}$ in $D^\ell$.
Further, define $D^h$ and $D^\ell$ with the following pdf for $0 \le v < v_1$:
\begin{align*}
f_{D^h}(v) & =
\begin{cases}
e^{-v} & 0 \le v < v_2 \\ 
(1+3\sqrt{\epsilon_0}) e^{-(1+3\sqrt{\epsilon_0})v + 3\sqrt{\epsilon_0} v_2} & v_2 \le v < v_1 
\end{cases} \\
f_{D^\ell}(v) & = e^{-v} 
\end{align*}
%

The corresponding complementary cdf are as follows, via simple calculations:
\begin{align*}
1- F_{D^h}(v) & = 
\begin{cases}
e^{-v} & 0 \le v < v_2 \\ 
e^{-(1+3\sqrt{\epsilon_0})v + 3\sqrt{\epsilon_0} v_2} & v_2 \le v < v_1 \\
0 & v = v_1
\end{cases} \\
1 - F_{D^\ell}(v) & = 
\begin{cases}
e^{-v} & 0 \le v < v_2 \\
0 & v = v_1
\end{cases}
\end{align*}

Conditions \ref{property:base_pointmass}, \ref{property:prob_high_values}, \ref{property:prob_critical_interval}, \ref{property:density_gap}, and \ref{property:v1_point_mass} hold trivially by the construction.\
Condition \ref{property:kl} follows by the construction and \Cref{lem:dptrick}, with $\Omega_1 = [0, v_2)$, $\epsilon_1 = 0$, $\Omega_2 = [v_2, v_1)$, $\epsilon_2 = 3 \sqrt{\epsilon_0}$, and $\Omega_3 = \{ v_1 \}$, $\epsilon_3 = 3\epsilon_0$.
Further, conditions \ref{property:virtual_value_gap}, \ref{property:virtual_value_bound_small_values}, and \ref{property:regular_high_distribution} can be verified from the virtual values of $D^h$ and $D^\ell$ below, which follows from straightforward calculations:
\begin{align*}
\phi_{D^h}(v) & = 
\begin{cases}
v - 1 & 0 \le v < v_1 \\ 
v - \frac{1}{1 + 3\sqrt{\epsilon_0}} \approx v - 1 + 3\sqrt{\epsilon_0} & v_2 \le v < v_1 \\
\ln n & v = v_1
\end{cases} \\
\phi_{D^\ell}(v) & = 
\begin{cases}
v - 1 \qquad\qquad\qquad\qquad\quad~ & 0 \le v < v_2 \\
\ln n & v = v_1
\end{cases}
\end{align*}

To show condition \ref{property:revenue_gap}, by our choice of $p$ and $\Delta$, it remains to show that $\opt(\mathbf{D}) \le O(\ln n)$ for all $\mathbf{D} \in \mathcal{H}$.
This holds trivially because the values are upper bounded by $\ln n$ in all three distributions $D^b$, $D^h$, and $D^\ell$.
Putting together shows the weaker lower bound of $\tilde{\Omega}(n \epsilon^{-3/2})$.

\begin{figure}
	\centering
	\includegraphics[width=0.48\textwidth]{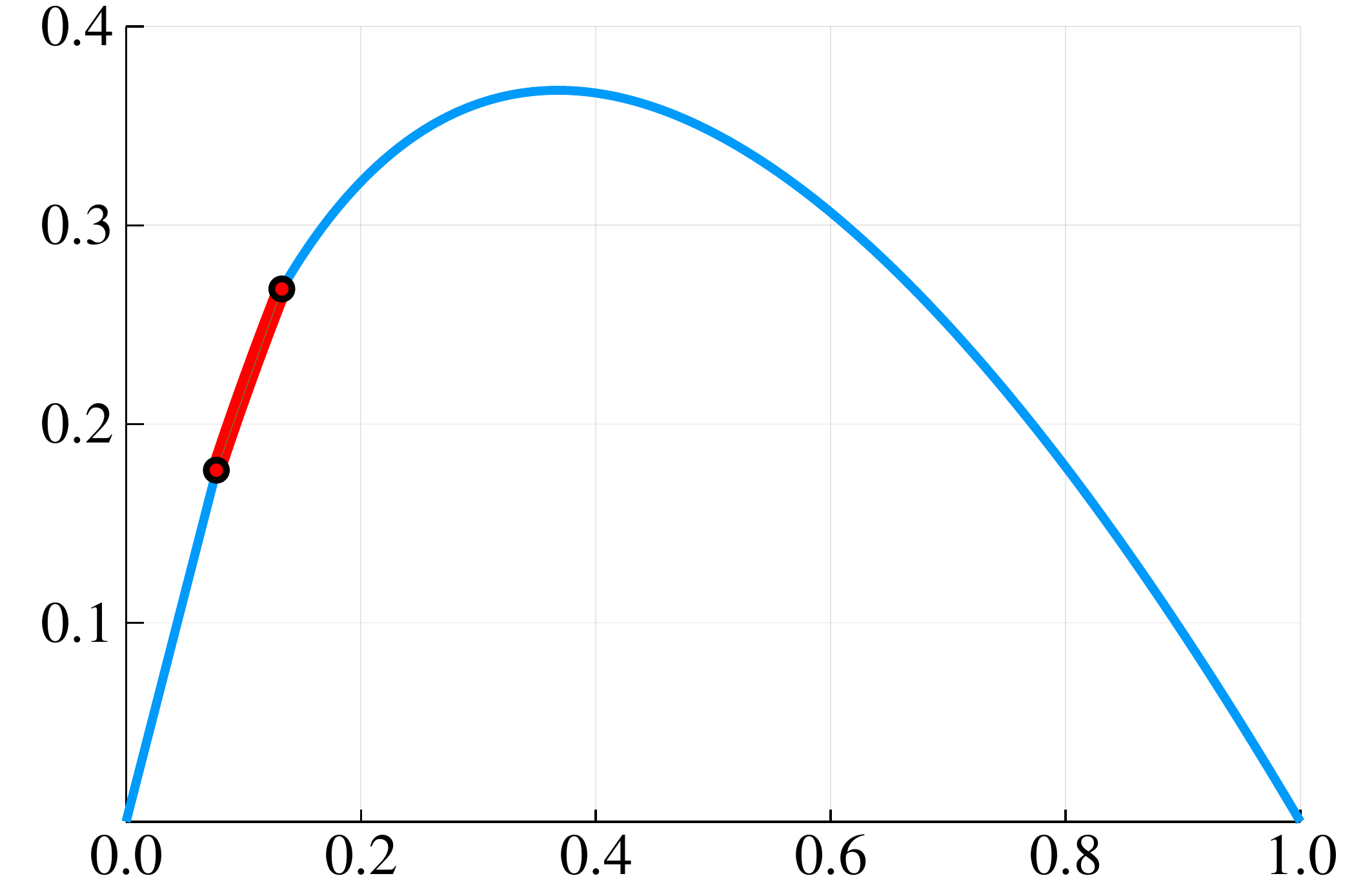}
	\includegraphics[width=0.48\textwidth]{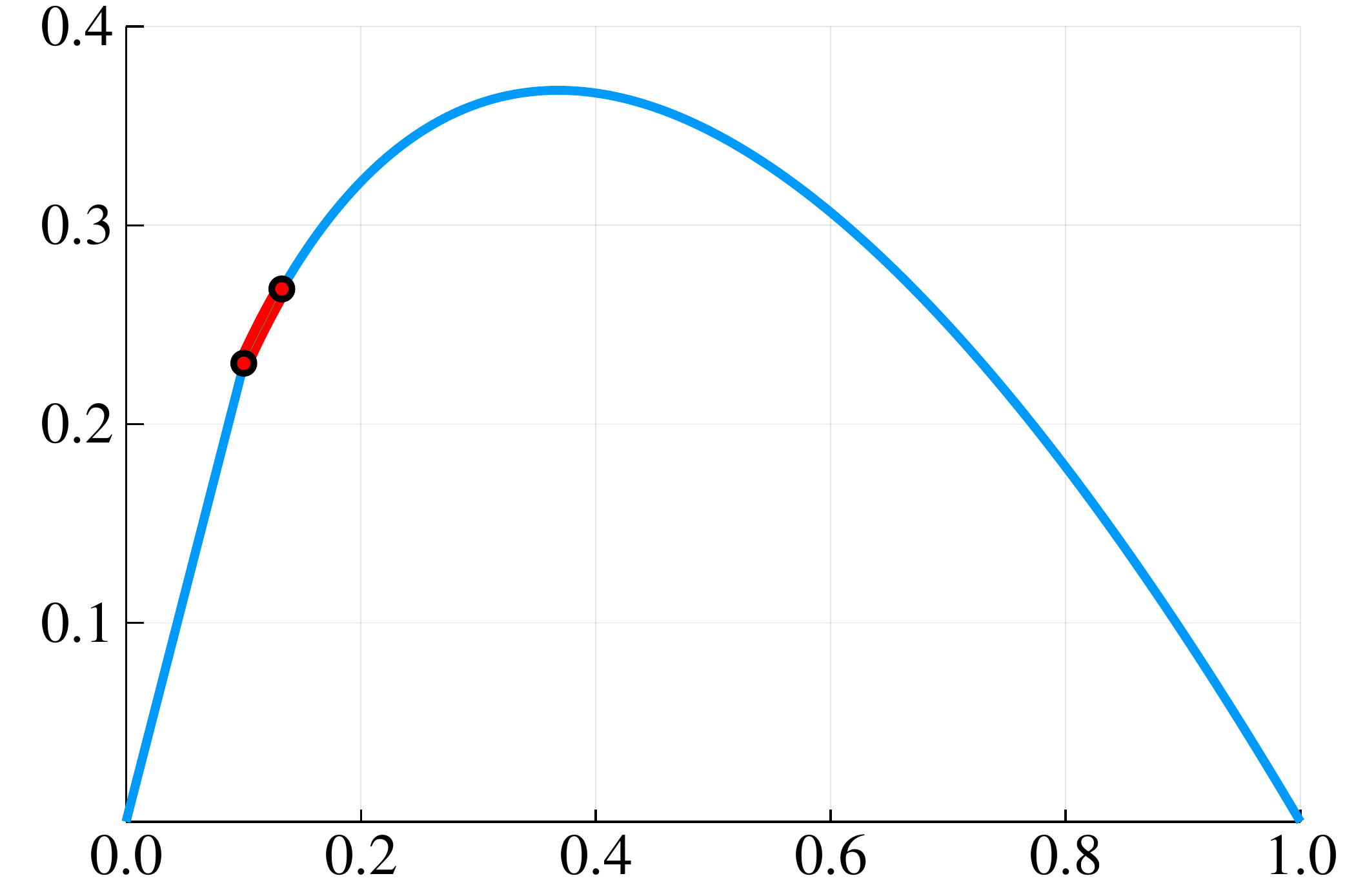}
	\caption{
		The left-hand-size and the right-hand-side are the revenue curves, in the quantile space, of $D^h$ and $D^\ell$ respectively used in the lower bound for continuous MHR distributions, with $n = 10$, and $\epsilon_0 = 0.1$.
		The critical value interval between $v_2$ and $v_1$ is plotted in bold (red).
	}
\label{fig:continuous_mhr}
\end{figure}

	\section{Matroid Constraint: Upper Bound}
\label{app:matroid_ub}

In this section, we generalize our analysis of sample complexity upper bound to single-parameter auctions where each bidder is unit-demanded and the allocation has to meet a matroid constraint $\mathcal{M}=([n],\mathcal{I})$ with rank $k$. 
Readers may think of the $k$-uniform matroid, i.e., there are $k$ copies of the item and therefore any feasible allocation can allocate to up to $k$ bidders, as a running example.

Formally, we define auction with matroid constraint as follows:
\begin{definition}[Auction with matroid constraint]
An auction with matroid constraint of rank $k$ is specified by an integer $k\in [n]$ a matroid $\mathcal{M}=([n],\mathcal{I})$ with rank $k$, where the allocation vector $\mathbf{x}$ must be probabilistic combination of indicator function of independent sets of $\mathcal{M}$. 
\end{definition}

For simplicity, we only consider the case of $[0,1]$-bounded support distributions with additive error. We will show that our algorithm can guarantee an up to $\epsilon$ error to the optimal expected revenue with $\tilde{O}(\frac{kn}{\epsilon^2})$ samples.

\begin{theorem}\label{thm:matroid-ub}
	For any $0 < \epsilon < \frac{1}{2}$, any $n$-bidder product value distribution $\mathbf{D}$ bounded in $[0,1]^n$, any $k\in [n]$ and any matroid constraint $\mathcal{M}=([n],\mathcal{I})$ with rank $k$, \Cref{alg:dominated_empirical_myerson} returns a mechanism with an expected revenue at least $\opt(\mathbf{D}) - O(\epsilon)$ with high probability, if $m$ is at least $\tilde{\Omega}(kn\epsilon^{-2})$.
\end{theorem}

In particular, this bound holds for the $k$-uniform matroid, which corresponds to having $k$ identical copies of the item.
Further, we will establish for this case a matching lower bound, up to a polylogarithmic factor, in Appendix~\ref{app:matroid_lb}

\begin{corollary}
	For any $0 < \epsilon < \frac{1}{2}$, any $n$-bidder product value distribution $\mathbf{D}$ bounded in $[0,1]^n$, any $k\in [n]$ such that the total number of identical items is $k$ and each bidder has unit-demand, \Cref{alg:dominated_empirical_myerson} returns a mechanism with an expected revenue at least $\opt(\mathbf{D}) - O(\epsilon)$, with probability at least $1-o(1)$, if $m$ is at least $\tilde{\Omega}(kn\epsilon^{-2})$.
\end{corollary}


Similar to the single-item case, we use a two-step argument to show the sample complexity upper bound: the first step uses revenue monotonicity and the second step uses information theory. Our result for the second step can be generated into a broader class of auctions, which is stated as follows:

To show this lemma, we use 
a hybrid argument: 
First partition $[n]$ into $k$ groups of size $\frac{n}{k}$ as $B_1,\cdots,B_k$, then define $k+1$ \emph{hybrid distributions} $\mathbf{D_0},\cdots,\mathbf{D_k}$ where $\mathbf{D_0}=\mathbf{\tilde{D}'}$, $\mathbf{D_k}=\mathbf{D'}$ (here $\mathbf{D'}$ is the auxiliary distribution defined in Lemma \ref{lem:auxiliary_distributions}), and $\forall i\in [k]$, $\mathbf{D_i}$ is the product of marginal auxiliary distributions for bidders in $\cup_{j=1}^iB_j$ and their doubly shaded version for bidders in $\cup_{j=i+1}^k B_j$. Therefore $\forall i\in [k+1]$, $\mathbf{D_{i}}\succ \mathbf{D_{i-1}}$. We then show an upper bound of $\opt(\mathbf{D_i}) - \opt(\mathbf{D_{i-1}})$ for all $i$ using information theory.

\begin{definition}[Hybrid Distributions]
    Let $B_1, \dots, B_k$ be a partition of $[n]$ and each group has size $\frac{n}{k}$. 
    Define $\mathbf{D_i}$ to be shaded surrogate only on the coordinates that belong to $\bigcup_{\ell=1+1}^k B_\ell$, i.e., $\mathbf{D_i}\defeq\prod_{j\in \cup_{\ell=1}^iB_\ell}D_j'\times\prod_{j\in \cup_{\ell=i+1}^kB_\ell}\tilde{D_j'}$.
\end{definition}

\begin{lemma}\label{lem:hybrid-KL-ub}
	For any partition $B_1,\cdots,B_k$ of $[n]$ where each group has size $\frac{n}{k}$, and $\forall 2\le i\le k$, $\skl(\mathbf{D_i},\mathbf{D_{i-1}})=\tilde{O}\left(\frac{n}{km}\right)$.
\end{lemma}
\begin{proof}
    According to Lemma~\ref{lem:KL_upper_bound} in Section~\ref{sec:upper_bound_multi_bidder}, $\forall j\in B_i$:
    \begin{equation*}
        \skl(D_j',\tilde{D}_j')\le O(\frac{\ln (mn\delta^{-1}) \ln (nk\epsilon^{-1})}{m}) = \tilde{O}(\frac{1}{m})
    \end{equation*}
    Therefore, 
    \begin{equation*}
        \skl(\mathbf{D_i},\mathbf{D_{i-1}}) \le \sum_{j\in B_i}\skl(D_j',\tilde{D}_j') = \tilde{O}(\frac{n}{km})
    \end{equation*}
\end{proof}

\begin{lemma}[Bernstein's inequality for sampling without replacement, e.g., \cite{bardenet2015concentration}]\label{lem:Bernstein-without-replacement}
    Let $\mathcal{X}=(x_1,\cdots,x_N)$ be a finite population of $N$ points and $(X_1,\cdots,X_n)$ be a random sample drawn without replacement from $\mathcal{X}$. Suppose $\forall i\in[N],x_i\in[a,b]$, let $\mu=\frac{1}{N}\sum_{i=1}^N x_i$ be the mean of $\mathcal{X}$, and $\sigma^2 = \frac{1}{N}\sum_{i=1}^N(x_i-\mu)^2$ be the variance of $\mathcal{X}$, then
    \begin{equation*}
        \Pr\left[\frac{1}{n} \sum_{i=1}^n (X_i-\mu) \ge \epsilon \right] \le exp(-\frac{n\epsilon^2}{2\sigma^2 + (2/3)(b-a)\epsilon}) ~.
    \end{equation*}
\end{lemma}

\begin{lemma}\label{lem:matroid_ub_partition}
    For any distribution $\mathbf{D}$ that has support in $[0,1]^n$ and any deterministic allocation rule $\mathbf{x}(\mathbf{v})$, there exists a partition $P=\{B_1\cdots B_k\}$ that separate $[n]$ into $k$ groups of size $\frac{n}{k}$, such that, 
    \begin{equation*}
        \Pr_{\mathbf{v}\in\mathbf{D}}\left[\forall i\in [k], \sum_{j\in B_i}x_j(\mathbf{v})\le O(\log \frac{n}{\epsilon})\right] \ge 1-\frac{\epsilon}{4nk^2} ~.
    \end{equation*} 
\end{lemma}
\begin{proof}
    Suppose $\mathcal{P}$ is the set of all partitions that separate $[n]$ into $k$ groups of size $\frac{n}{k}$. Let $U(\mathcal{P})$ be the uniform distribution of $\mathcal{P}$, if for all $\mathbf{v}\in [0,1]^n$,
    \begin{equation}\label{eqn:matroid-random-partition}
        \Pr_{P\sim U(\mathcal{P})}\left[\exists i\in[k] \text{, such that }\sum_{j\in B_i}x_j(\mathbf{v}) \ge 3\log\frac{n}{\epsilon}\right] \le \frac{\epsilon}{4nk^2} ~,
    \end{equation}
    Then
    \begin{equation*}
        \Pr_{\mathbf{v}\sim \mathbf{D}, P\sim U(\mathcal{P})}\left[\exists i\in[k] \text{, such that }\sum_{j\in B_i}x_j(\mathbf{v}) \ge 3\log\frac{n}{\epsilon}\right] \le\frac{\epsilon}{4nk^2}
    \end{equation*}
    Therefore $\exists P\in \mathcal{P}$,
    \begin{equation*}
        \Pr_{\mathbf{v}\sim \mathbf{D}}\left[\exists i\in[k] \text{, such that }\sum_{j\in B_i}x_j(\mathbf{v}) \ge 3\log\frac{n}{\epsilon}\right] \le \frac{\epsilon}{4nk^2} ~.
    \end{equation*}
    So it suffices to show (\ref{eqn:matroid-random-partition}). Fix $\mathbf{v}\in[0,1]^n$, for all $i\in [k]$, $\{x_j(\mathbf{v})\}_{j\in B_i}$ is a random sample (without replacement) of $k$ elements in $\{x_j(\mathbf{v})\}_{j\in [n]}$, which has exactly $k$ 1's and $n-k$ 0's. Then according to Bernstein's inequality for sampling without replacement (Lemma \ref{lem:Bernstein-without-replacement}), if $\epsilon<\frac{n}{2k}$,
    \begin{equation*}
        \Pr_{P\sim U(\mathcal{P})}\left[\sum_{j\in B_i}x_j(\mathbf{v})-1>4\log(\frac{n}{\epsilon}) \right] \le exp(-\frac{16\log^2(n/\epsilon)}{2(n-k)/n+8/3\cdot \log(n/\epsilon)})\le \frac{\epsilon^4}{n^4}\le \frac{\epsilon}{8nk^3}
    \end{equation*}
    Therefore from union bound,
    \begin{equation*}
        \Pr_{P\sim U(\mathcal{P})}\left[\exists i\in[k] \text{, such that }\sum_{j\in B_i}x_j(\mathbf{v}) \ge 4\log\frac{n}{\epsilon}\right] \le \frac{\epsilon}{4nk^2} ~.
    \end{equation*}
\end{proof}

\begin{corollary}\label{cor:matroid_ub_low_allocation_inside_group}
    There exists a partition $\mathcal{P}=\{B_1,B_2\cdots,B_k\}$ that partitions $[n]$ into $k$ groups of size $\frac{n}{k}$ such that for any $i\in [k]$,
    \begin{equation*}
        \Pr_{\mathbf{v}\in\mathbf{D_i}}\left[\forall l\in [k], \sum_{j\in B_l}x_j(\mathbf{v})\le O(\log \frac{n}{\epsilon})\right] \ge 1-\frac{\epsilon}{4kn} ~,
    \end{equation*}
\end{corollary}
\begin{proof}
    We can choose the $\mathbf{D}$ in Lemma~\ref{lem:matroid_ub_partition} to be $\frac{1}{k}\sum_{i=1}^k\mathbf{D_i}$, then we have
    \begin{equation*}
    \frac{1}{k}\sum_{i=1}^k\Pr_{\mathbf{v}\in\mathbf{D_i}}\left[\forall l\in [k], \sum_{j\in B_l}x_j(\mathbf{v})\le O(\log \frac{n}{\epsilon})\right]\ge 1-\frac{\epsilon}{4nk^2}
    ~.    
    \end{equation*}
    
\end{proof}

\begin{lemma}{(\cite{schrijver2003combinatorial}, Corollary 39.12a)}\label{lem:maxindeset-perfect-matching}
If both $I$ and $I'$ are independent sets of matroid $\mathcal{M}=([n],\mathcal{I})$ and $|I|=|I'|$, then there exists a perfect matching $M$ between $I \backslash I'$ and $I' \backslash I$ such that for each $i\in I \backslash I'$ and $j\in I' \backslash I$ such that $(i,j)\in M$, we have $I\cup \{j\} \backslash \{i\}\in \mathcal{I}$ and $I'\cup \{i\} \backslash \{j\}\in \mathcal{I}$.
\end{lemma}

\begin{lemma}\label{lem:matroid-change-value}
For any matroid $\mathcal{M}=([n],\mathcal{I})$,  any weight function $v:[n]\to \mathbb{R}^n$, and any fixed tie-breaking rule, suppose $I$ is the maximum weighted independent set of $\mathcal{I}$ regarding $v$ (assume elements with 0 weight are always not included), then
$\forall i\in [n]$, $\forall$ weight function $v':[n]\to \mathbb{R}^n$ such that $v'(i)\le v(i)$ and $v'(j)=v(j)$ for any $j\not=i$, the maximal weighted independent set $J$ of $(\mathcal{M},v')$ satisfy $I \backslash J \subseteq \{i\}$.
\end{lemma}

\begin{proof}
 Suppose $j\in I \backslash J$ and $j\neq i$.

First, $|I|\ge |J|$ otherwise another element from $J$ could be added to $I$ to form a new independent set, but any element in $J$ is positive, which contradicts with the definition of $I$. 
For the same reason, if $v'(i)>0$ then $|J|\ge |I|$, if $v'(i)\le 0$ then $|J|\ge |I|-1$.

For the case when $|J|=|I|$, then from Lemma \ref{lem:maxindeset-perfect-matching}, $\exists k\in J \backslash I$, such that $I\cup \{k\} \backslash \{j\}\in \mathcal{I}$ and $I'\cup \{k\} \backslash \{j\}\in \mathcal{I}$. If $v(j)<v(k)$, then the total value of $I\cup \{k\} \backslash \{j\}$ should be larger than that of $I$, which contradicts to its definition; if $v(j)>v(k)$, then $J\cup \{j\} \backslash \{k\}$ should be larger than that of $J$, which also leads to a contradiction.

For the case when $|J|=|I|-1$, we have $v'(i)\le 0$. From Lemma \ref{lem:maxindeset-perfect-matching}, $\exists k\in J \backslash I$, such that $I\cup \{k\} \backslash \{j,i\}\in \mathcal{I}$ and $J\cup \{j\} \backslash \{k\}\in \mathcal{I}$. Because $|I\cup \{k\}\backslash\{i,j\}|<|I|$, we have $I\cup \{k\}\backslash\{i\}\in \mathcal{I}$ or $I\cup \{k\}\backslash\{j\}\in\mathcal{I}$. But if $I\cup \{k\}\backslash\{i\}\in \mathcal{I}$, $|I\cup \{k\}\backslash\{i\}|>|J|$ and each of its element is positive, contradicting with the definition of $J$, so $I\cup \{k\}\backslash\{j\}\in\mathcal{I}$. If $v(j)<v(k)$, then the total value of $I\cup \{k\} \backslash \{j\}$ should be larger than that of $I$, which contradicts to its definition; if $v(j)>v(k)$, then $J\cup \{j\} \backslash \{k\}$ should be larger than that of $J$, which also leads to a contradiction.
\end{proof}

\begin{lemma}\label{lem:matroid_opt_dif_and_group_dif}
    For any $j\in [k]$, let $p_l$ be the payment function of $M_{\mathbf{D_{j}}}$
    $$\opt(\mathbf{D_j})-\opt(\mathbf{D_{j-1}})\le\E_{\mathbf{v}\sim \mathbf{D_j}}\sum_{l\in B_j} p_l(\mathbf{v})-\E_{\mathbf{v}\sim \mathbf{D_{j-1}}}\sum_{l\in B_j} p_l(\mathbf{v})$$
\end{lemma}

\begin{proof}
    Let $v_i^{(j)}(q)$ be the inverse of quantile function for bidder $i$ on distribution $\mathbf{D_j}$ and $\mathbf{v}^{(j)}(\mathbf{q})$ be the vector of quantile functions. First
    $\opt(\mathbf{D_j})-\opt(\mathbf{D_{j-1}})$ 
    is bounded by $\rev(M_{\mathbf{D_j}},\mathbf{D_j})-\rev(M_{\mathbf{D_j}},\mathbf{D_{j-1}})$, 
    and because for $i=j$ or $i=j-1$,
    $$
    \rev(M_{\mathbf{D_j}},\mathbf{D_i})= \int_{\mathbf{q}\in [0,1]^n} \sum_{l\in [n]\backslash B_j} x_l(\mathbf{v}^{(i)}(\mathbf{q}))\phi_l(v_l^{(i)}(q_l))d\mathbf{q} + \sum_{l\in B_j}\E_{\mathbf{v}\sim \mathbf{D_i}}p_l(\mathbf{v})~,
    $$
    so it suffices to show that for any $\mathbf{q}\in [0,1]^n$ and any $l\in [n]\backslash B_j$,
    $$
     x_l(\mathbf{v}^{(j)}(\mathbf{q}))
    \ge  x_l(\mathbf{v}^{(j-1)}(\mathbf{q}))
    ~.
    $$
    If we view the ironed virtual value of $\mathbf{D_j}$ as a weight of the matroid $\mathbf{M}$, $M_{\mathbf{D_j}}$ picks the maximum weighted independent set as winners. Suppose the set of winners for $\mathbf{D_j}$ and $\mathbf{D_{j-1}}$ are respectively $I_j$ and $I_{j-1}$ Because for any $l\in [n]\backslash B_j$, $v_l^{(j)}(q_l)=v_l^{(j-1)}(q_l)$, and for any $l\in B_j$, $v_l^{(j)}(q_l)\ge v_l^{(j-1)}(q_l)$, we can apply Lemma~\ref{lem:matroid-change-value} to every elements in $B_j$ sequentially and get $I_{j}\backslash I_{j-1}\subseteq B_j$.
\end{proof}

\begin{lemma} \label{lem:matroid-shade-opt-dif}
For any $0 < \epsilon < \frac{1}{2}$, any $n$-bidder product value distribution $\mathbf{D}$ bounded in $[0,1]^n$, any $k\in [n]$ and any matroid constraint $\mathcal{M}=([n],\mathcal{I})$ with rank $k$,
it holds when $m=\tilde{\Omega}(kn\epsilon^{-2})$, 
	\begin{equation*}
	    \opt(\mathbf{D'}) - \opt (\mathbf{\tilde{D}'}) \le \epsilon ~,
	\end{equation*}
where 
    \begin{equation*}
    D_i' = \truncatebottom_{\epsilon/k^2} \circ \truncatetop_{\mathbf{\bar{v}}} (D_i) ~,~
    \Bar{v}_i=\sup \{v: q^{D_i}(v)\ge \frac{\epsilon^2}{nk^2}\} ~,
    \end{equation*}
and
    \begin{equation*}
    \tilde{D_i'} = \doubleshading(D_i') ~.
\end{equation*}
\end{lemma}

\begin{proof}
We will use a proof similar to Lemma~\ref{lem:difference_ori_shaded} to show that for any $1\le j\le k$, $\opt(\mathbf{D_j})-\opt(\mathbf{D_{j-1}})\le \epsilon/k$. But instead of considering the difference between $\opt(\mathbf{D_j})$ and $\opt(\mathbf{D_{j-1}})\le \epsilon/k$, we consider the payment of bidders in $B_j$, i.e., 
\begin{equation*}
    \E_{\mathbf{v}\sim \mathbf{D_j}}\sum_{l\in B_j} p_l(\mathbf{v})-\E_{\mathbf{v}\sim \mathbf{D_{j-1}}}\sum_{l\in B_j} p_l(\mathbf{v})
\end{equation*}
According to Lemma~\ref{lem:matroid_opt_dif_and_group_dif},
the difference between $\opt(\mathbf{D_j})$ and $\opt(\mathbf{D_{j-1}})$ is upper bounded by that. Choose $N$ in Lemma~\ref{lem:difference_ori_shaded} to be $\tilde{O}(\frac{\epsilon^2}{k^2})$.

On the one hand, $\skl(\mathbf{D_j},\mathbf{D_{j-1}})=\tilde{O}\left(\frac{n}{km}\right)=\tilde{O}(\frac{\epsilon^2}{k^2})$ by Lemma~\ref{lem:hybrid-KL-ub}.

On the other hand, consider the following algorithm: 
Take $N$ samples $\mathbf{v_1},\cdots,\mathbf{v_N}$, calculate $\frac{1}{N}\sum_{s=1}^N\sum_{l\in B_j} \min(p_l(\mathbf{v_s}),3\log\frac{n}{\epsilon})$ .

According to Corollary~\ref{cor:matroid_ub_low_allocation_inside_group}, for $i=j$ or $j-1$, $\E_{\mathbf{v}\sim \mathbf{D_i}}\sum_{l\in B_i} \min(p_l(\mathbf{v}),3\log\frac{n}{\epsilon})$ and $\E_{\mathbf{v}\sim \mathbf{D_i}}\sum_{l\in B_i} p_l(\mathbf{v})$ differ by at most $\frac{\epsilon}{4kn}*\lceil\frac{n}{k}\rceil\le \frac{\epsilon}{4k}$. And use Bernstein inequality (Lemma~\ref{lem:bernstein}), we have 
\begin{equation*}
    \Pr\left[\left|\frac{1}{N}\sum_{s=1}^N\sum_{l\in B_i} \min(p_l(\mathbf{v_s}),3\log\frac{n}{\epsilon})-\E_{\mathbf{v}\sim \mathbf{D_i}}\sum_{l\in B_i} \min(p_l(\mathbf{v}),3\log\frac{n}{\epsilon})\right|>\frac{\epsilon}{2k}\right]\le e^{-\frac{(\frac{\epsilon N}{k})^2}{18N\log^2 \frac{n}{\epsilon}+2\log(\frac{n}{\epsilon})\frac{\epsilon N}{k}}}=o(1)
\end{equation*}
Therefore, with probability $1-o(1)$, the algorithm approximate $\E_{\mathbf{v}\sim \mathbf{D_i}}\sum_{l\in B_i} p_l(\mathbf{v})$ with an additive factor of $\frac{\epsilon}{4k}+\frac{\epsilon}{2k}\le \frac{\epsilon}{k}$.
\end{proof}

\begin{proof}{(of Theorem \ref{thm:matroid-ub})}
We need to show that when $m=\tilde{\Omega}(kn\epsilon^{-2})$
\begin{equation*}
    \opt(\mathbf{D}) - \opt(\mathbf{\tilde{D}}) \le 3\epsilon
\end{equation*}
Then, by Lemma \ref{lem:analysis_revenue_monotonicity}, 
\begin{equation*}
    \rev(M_{\mathbf{\tilde{E}}},\mathbf{D})\ge \opt(\mathbf{\tilde{D}}) \ge \opt(\mathbf{D}) - 3\epsilon
\end{equation*}

Since for all $[0,1]$-bounded support product distribution $\mathbf{D}$,
\begin{align*}
\opt(\mathbf{D}) - \opt(\truncatetop_{\mathbf{\bar{v}}}(\mathbf{D})) & \le k\cdot \Pr \big[ \exists i \in [n] : v_i > \bar{v}_i \big] && \text{(at most $k$ items, values bounded by $1$)} \\[1.5ex]
& \le k\cdot \sum_{i = 1}^n \Pr \big[ v_i > \bar{v}_i \big] && \text{(union bound)} \\
& \le k\cdot \sum_{i = 1}^n \frac{\epsilon^2}{nk} && \text{(definition of $\bar{v}_i$'s)} \\[1.5ex]
& = \epsilon^2 
~.
\end{align*}
and since if we run $M_{\mathbf{D}}$ on $\truncatebottom_{\epsilon/k^2}(\mathbf{D})$, then with probability at least $(1-\frac{\epsilon}{k^2})^k$, allocation is the same as running $M_{\mathbf{D}}$ on $\mathbf{D}$, we have
\begin{align*}
\opt(\mathbf{D}) - \opt(\truncatebottom_{\epsilon/k^2}(\mathbf{D}))
&\le \opt(\mathbf{D}) - \rev(M_{\mathbf{D}},\truncatebottom_{\epsilon/k^2}(\mathbf{D}))\\
&\le (1-(1-\frac{\epsilon}{k^2})^k)\cdot k \\
&\le \frac{\epsilon}{k} \cdot k = \epsilon
\end{align*}
we have
\begin{align}
    \opt(\mathbf{D'}) &= \opt(\truncatebottom_{\epsilon/k^2} \circ \truncatetop_{\mathbf{\bar{v}}}(\mathbf{D})) \nonumber \\
    &\ge \opt(\truncatetop_{\mathbf{\bar{v}}}(\mathbf{D})) - \epsilon \nonumber \\
    &\ge \opt(\mathbf{D}) - 2\epsilon \label{eqn:matroid-ancillary-opt}
\end{align}
Then from $\mathbf{\tilde{D}}\succ \mathbf{\tilde{D}'}$, we have:
\begin{align*}
    \opt(\mathbf{D}) - \opt(\mathbf{\tilde{D}}) &\le \opt(\mathbf{D}) - \opt(\mathbf{\tilde{D}'}) & \text{(weak revenue monotoinicity)} \\
    & = (\opt(\mathbf{D}) - \opt(\mathbf{D'})) + (\opt(\mathbf{D'}) - \opt(\mathbf{\tilde{D}'}))\\
    & \le 2\epsilon + \epsilon = 3\epsilon & \text{(\eqref{eqn:matroid-ancillary-opt} and Lemma~\ref{lem:matroid-shade-opt-dif})}
\end{align*}
Combine this and Lemma~\ref{lem:analysis_revenue_monotonicity}, we get $\rev(M_{\mathbf{\tilde{E}}}, \mathbf{D}) \ge \opt(\mathbf{D})-3\epsilon$.

\end{proof}

	\section{Matroid Constraint: Lower Bound}
\label{app:matroid_lb}

In this section always assume $n\ge 2k$, and all bidders' valuation are bounded in $[0,1]$. We will show that the sample complexity lower bound of $k$-unit demand auction coincides with its upper bound, which is $\tilde{\Omega}(nk\epsilon^{-2})$. We first define a family of distribution $\mathcal{H}$. Let 
\[
\mathcal{H} = \big\{ \mathbf{D} : D_1 = D_2 = \cdots D_k= D^b \text{, and } D_i = D^h \text{ or } D^\ell \text{ for all } k < i \le n \}
~.
\]
where $D^b$, $D^{\ell}$ and $D^h$ satisfy
    $$
        D_b \text{ is a point distribution at } 1/2
    $$
        \begin{align*}
        f_{D^\ell}(v) & = 
        \begin{cases}
        1-\frac{k}{2n}& v= v_3 \defeq 1/2 + \frac{k}{8n} \\
        \frac{k-\epsilon}{4n}& v= v_2 \defeq \frac{3}{4} \\
        \frac{k+\epsilon}{4n}& v= v_1 \defeq 1
        \end{cases} \\
        f_{D^h}(v) & = 
        \begin{cases}
        1-\frac{k}{2n}& v= v_3 = 1/2 + \frac{k}{8n} \\
        \frac{k+\epsilon}{4n}& v= v_2 = \frac{3}{4}  \\
        \frac{k-\epsilon}{4n}& v= v_1 = 1.
        \end{cases} \\
        \end{align*}
We can verify that 
\begin{equation}\label{eqn:k-unit-lb-virtual-value-gap}
    \phi^\ell(v_2) + \epsilon/2k \le 1/2 \le \phi^h(v_2) - \epsilon/2k
\end{equation}
and 
\begin{equation}\label{eqn:k-unit-lb-derivate-bound}
    \sqrt{2} \ge \frac{dD^\ell}{dD^h}(v_2) \ge \frac{1}{\sqrt{2}}.
\end{equation}
Define 
$$\mathcal{V}_{i}=\bigg\{ \mathbf{b} = (b_1, b_2, \dots, b_n), b_1=\cdots=b_k=\frac{1}{2},  b_{i} = v_2, \text{ and }|\{j|b_j\ge v_2\}|\le k \bigg\} ~.$$
\begin{theorem}
	\label{thm:matroid_hardness}
	If an algorithm A takes $m$ samples from an arbitrary product value distribution $\mathbf{D} \in \mathcal{H}'$ and returns, with probability at least $0.99$, a mechanism whose expected revenue is at least:
	\[
	\opt(\mathbf{D}) - O(\epsilon) 
	~.
	\]
	Then, the number of samples $m$ is at least:
	\[
	\Omega(nk\epsilon^{-2})
	~.
	\]
\end{theorem}

Fix any $i$, and any $\mathbf{D}_{-i}= \times_{j \neq i} D_j$ such that $D_1=\cdots = D_k=D^b$ and $D_j \in \{D^h,D^{\ell}\}$ for all $j \ne 1, i$.
Let $\mathbf{D}^1 = (\mathbf{D}_{-i}, D_i = D^h) \in \mathcal{H}$ and $\mathbf{D}^2 = (\mathbf{D}_{-i}, D_i = D^\ell) \in \mathcal{H}$ be a pair of distributions that differ only in the $i$-th coordinate. 
Then, we have:
\[
\skl \big( \mathbf{D}^1, \mathbf{D}^2 \big) = \skl(D^h, D^\ell) = O(\frac{\epsilon^2}{nk})
~.
\]
The second inequality follows from Lemma~\ref{lem:dptrick} by choosing $\Omega_1=\{v_3\}$, $\epsilon_1=0$ and $\Omega_2=\{v_1,v_2\}$, $\epsilon_2=\frac{\epsilon}{2k}$.
Then, since algorithm $A$ takes $m < c \cdot \skl(D^h, D^\ell)^{-1} = c \cdot \skl(\mathbf{D}^1, \mathbf{D}^2)^{-1}$ samples for some sufficiently small constant $c$, by \Cref{lem:kl_divergence}, it cannot distinguish whether the underlying distribution is $\mathbf{D}^1$ or $\mathbf{D}^2$ correctly, and as a result will choose a mechanism from essentially the same distribution in both cases.

\begin{lemma}
	\label{lem:matroid_hardness_similar_decisions}
	For any mechanism $M$, the probability that $M$ picks bidder $i$ ($i\in I$) as the winner, conditioned on the value vector $\mathbf{v}$ is in $\mathcal{V}_{I,J}$, differs by at most a factor of $2$ whether $\mathbf{v}$ is drawn from $\mathbf{D}^1$ or $\mathbf{D}^2$.
\end{lemma}

\begin{proof}
Same as Lemma~\ref{lem:meta_hardness_similar_decisions} using Equation \ref{eqn:k-unit-lb-derivate-bound}.
\end{proof}

Define
\begin{align*}
\mathcal{M}^1 & = \left\{ M : \Pr_{\mathbf{v} \sim \mathbf{D}^1 : \mathbf{v} \in \mathcal{V}_{i}} \big[ \text{$M$ picks $i$ as the winner} \big] \ge \frac{2}{3} \right\} 
~, \\
\mathcal{M}^2 & = \left\{ M : \Pr_{\mathbf{v} \sim \mathbf{D}^1 : \mathbf{v} \in \mathcal{V}_{i}} \big[ \text{$M$ picks $i$ as the winner} \big] < \frac{2}{3} \right\}
~.
\end{align*}

\begin{corollary}
	\label{cor:k-unit-hardness-similar-decisions}
	For any $M \in \mathcal{M}^1$, we have that:
	\[
	\Pr_{\mathbf{v} \sim \mathbf{D}^2 : \mathbf{v} \in \mathcal{V}_i} \big[ \text{\rm $M$ picks $i$ as the winner} \big] \ge \frac{1}{3}
	~.
	\]
\end{corollary}

\begin{lemma}
	\label{lem:k-unit-hardness-mistakes}
	For either $j = 1$ or $j = 2$ (or both), we have:
	\[
	\Pr \big[ A(\mathbf{D}^j) \in \mathcal{M}^{3-j} \big] > \frac{1}{3}
	~.
	\]
\end{lemma}
\begin{proof}
    Same as Lemma \ref{lem:meta_hardness_mistakes}.
\end{proof}

\begin{lemma}
	\label{lem:matroid_hardness_virtual_value_maximizer}
	For any value distribution $\mathbf{D} \in \mathcal{H}$, the optimal mechanism w.r.t.\ $\mathbf{D}$ always chooses the bidder with $k$ highest virtual values as winners. 
\end{lemma}
\begin{proof}
    The virtual value and ironed virtual value of $D^b$, $D^h$ and $D^{\ell}$ are as follows:
    \begin{equation*}
        \bar{\phi}^b(v) = \phi^b(v) =  \frac{1}{2}
    \end{equation*}
    \begin{align*}
        \bar{\phi}^h(v) = \phi^h(v) = 
        \begin{cases}
            \frac{1}{2}  & v_3=1/2 + \frac{k}{8n}\\
            \frac{k+2\epsilon}{2(k + \epsilon)} & v_2 = \frac{3}{4}\\
            1 & v_1 = 1
        \end{cases}
    \end{align*}
    \begin{align*}
        \phi^{\ell}(v) = 
        \begin{cases}
            \frac{1}{2}  & v_3=1/2 + \frac{k}{8n}\\
            \frac{k-2\epsilon}{2(k - \epsilon)} & v_2 = \frac{3}{4}\\
            1 & v_1 = 1
        \end{cases}
    \end{align*}
    
    \begin{align*}
        \bar{\phi}^{\ell}(v) = 
        \begin{cases}
            \frac{4n-k-2\epsilon}{2(4n-k-\epsilon)}  & v_3=1/2 + \frac{k}{8n}\\
            \frac{4n-k-2\epsilon}{2(4n-k-\epsilon)} & v_2 = \frac{3}{4}\\
            1 & v_1 = 1
        \end{cases}
    \end{align*}
Since $ \frac{4n-k-2\epsilon}{2(4n-k-\epsilon)}<1/2$ and the first $k$ bidders all have a constant ironed virtual value of $1/2$, a bidder with $v\le v_2$ would never be chosen if his distribution is $D^{\ell}$. So the k highest virtual values always coincide with $k$ highest ironed virtual values.
\end{proof}

In the following discussion, Let $I_{A(\mathbf{D}^j)}$ denote the set of winners chosen by $A(\mathbf{D}^j)$. 

The revenue of optimal mechanism is
$$
\int_{\mathbf{v} }  \max_{|J|=k} \sum_{j \in J} \phi_j(v_j)  d \mathbf{D}
~,
$$
Now we only consider the case when $\mathbf{v}\in \cup_i \mathcal{V}_i$. According to the definition of $\mathcal{V}_i$, the number of bidders whose value are over $v_2$ is at most $k$. Therefore all bidders with value $v_1$, all bidders with value $v_2$ and having distribution $D^h$, and some bidders among $b_1$ to $b_k$ will be selected as winners in the optimal mechanism. Therefore when the valuation vector is in $\mathcal{V}_i$, selecting bidder $i$ with distribution $D^\ell$ would cause a loss of at least $1/2-\phi_\ell(v_2)\ge \frac{\epsilon}{2k}$, and not selecting bidder $i$ with distribution $D^h$ would cause a revenue loss of at least $\phi_l(v_2)-1/2\ge \frac{\epsilon}{2k}$. We define this loss to be $L_i(A(\mathbf{D}^j),\mathbf{v})$.
\begin{equation*}
    L_i(A(\mathbf{D}^j),\mathbf{v})=
    \begin{cases}
    \frac{\epsilon}{2k}\cdot \mathbb{I}(i\notin I_{A(\mathbf{D}^j)}(\mathbf{v})) & \text{if $j=1$ and $ \mathbf{v}\in \mathcal{V}_i$}\\
    \frac{\epsilon}{2k}\cdot \mathbb{I}(i\in I_{A(\mathbf{D}^j)}(\mathbf{v}))) &\text{if $j=2$ and $ \mathbf{v}\in \mathcal{V}_i$}\\
    0 & \text{otherwise}\\
\end{cases}
\end{equation*}

Then if we only consider the revenue loss caused by misclassifying bidder $i$'s type on valuation profile $\cup_i\mathcal{V}_i$, we can lower bound the revenue gap between $A(\mathbf{D}^j)$ and optimal mechanism:
\begin{equation}\label{eqn:matroid-virtual-value-loss}
    \begin{split}
     &\int_{\mathbf{v} \in \cup_i \mathcal{V}_i} \left( \max_{|J|=k} \sum_{j \in J} \phi_j(v_j) - \sum_{j \in I_{A(\mathbf{D}^j)}(\mathbf{v})} \phi_j(v_j) \right) d \mathbf{D}\\
    \ge& \int_{\mathbf{v} \in \cup_i \mathcal{V}_i} \sum_{i>k} L_i(A(\mathbf{D}^j),\mathbf{v}) d \mathbf{D}\\
    =& \sum_{i>k} \int_{\mathbf{v} \in \mathcal{V}_i} L_i(A(\mathbf{D}^j),\mathbf{v}) d \mathbf{D}
    \end{split}
\end{equation}
We will consider this quantity in the following discussion:
$$
\int_{\mathbf{v} \in \mathcal{V}_i} L_i(A(\mathbf{D}^j),\mathbf{v}) d \mathbf{D}
$$

\begin{lemma}
	For either $j = 1$ or $j = 2$ (or both), we have:
	\[
	\Pr_{A(\mathbf{D}^j)} \left[ \E_{\mathbf{v} \sim \mathbf{D}^j : \mathbf{v} \in \mathcal{V}_i} \left[ L_i(A(\mathbf{D}^j),\mathbf{v})  \right] \ge \frac{\epsilon}{6k} \right] \ge \frac{1}{3}
	~.
	\]
\end{lemma}

\begin{proof}
    According to Lemma \ref{lem:k-unit-hardness-mistakes}, we have either 
    \begin{equation*}
        \Pr \big[ A(\mathbf{D}^j) \in \mathcal{M}^{3-j} \big] > \frac{1}{3}
    \end{equation*}
    for $j=1,2$.We will consider the two cases separately,
    
    \textbf{Case 1: }$j=1$. From definition of $\mathcal{M}^{2}$, we have
    \begin{equation*}
        \Pr_{\mathbf{v} \sim \mathbf{D}^1 : \mathbf{v} \in \mathcal{V}_i} \big[ \text{$M$ picks $i$ as the winnner} \big] < \frac{2}{3} ~.
    \end{equation*}
    From Equation \ref{eqn:k-unit-lb-virtual-value-gap}, we have
    \begin{equation*}
         \E_{\mathbf{v} \sim \mathbf{D}^1 : \mathbf{v} \in \mathcal{V}_i} \left[ L_i(A(\mathbf{D}^1),\mathbf{v})  \right] \ge \frac{1}{3}\cdot \frac{\epsilon}{2k} = \frac{\epsilon}{6k} ~.
    \end{equation*}
    
    \textbf{Case 2: }$j=2$. From $A(\mathbf{D^2})\in M^1$ and Corollary \ref{cor:k-unit-hardness-similar-decisions}, we know that
    \begin{equation*}
        \Pr_{\mathbf{v} \sim \mathbf{D}^2 : \mathbf{v} \in \mathcal{V}_i} \big[ \text{\rm $M$ picks $i$ as the winner} \big] \ge \frac{1}{3}
    \end{equation*}
    Therefore,
    \begin{equation*}
         \E_{\mathbf{v} \sim \mathbf{D}^2 : \mathbf{v} \in \mathcal{V}_i} \left[ L_i(A(\mathbf{D}^2),\mathbf{v})  \right] \ge \frac{1}{3}\cdot \frac{\epsilon}{2k} = \frac{\epsilon}{6k} ~.
    \end{equation*}
    Thus the lemma holds.
\end{proof}

Let $\mathcal{B}_{\mathbf{D}} $ be
$$
\mathcal{B}_{\mathbf{D}} 
= 
\left\{ 
i : \Pr_{A(\mathbf{D})} \left[ \E_{\mathbf{v} \sim \mathbf{D} : \mathbf{v} \in \mathcal{V}_i} \left[ L_i(A(\mathbf{D}^j),\mathbf{v})  \right] \ge \frac{\epsilon}{6k} \right] \ge \frac{1}{3}
\right\}
,
$$
then Lemma~\ref{lem:meta_hardness_counting}, Corollary~\ref{cor:meta_hardness_bad_instance} and Lemma~\ref{lem:meta_hardness_bad_instance_count_const_prob} still holds.

\begin{lemma}\label{lem:matroid_probmass_vi}
    For any $\mathbf{D}\in \mathcal{H}$ and $i>k$,
    $$\Pr_{\mathbf{v}\sim \mathbf{D}}[\mathbf{v}\in \mathcal{V}_i]=\Theta(k/n).$$
\end{lemma}

\begin{proof}
    $\mathbf{v}\in \mathcal{V}_i$ is equivalent to say that $b_i=v_2$ and there are no more than $k-1$ other bidders has value over $v_2$. The two events are independent and probability of the first part is at between  $\frac{k-\epsilon}{4n}$ and $\frac{k+\epsilon}{4n}$. 
    
    Let $Z_j$ be the indicator of whether $b_j\ge v_2$, since the expected number of bidder that has value over $v_2$ is $(n-k)\frac{k}{2n}\le \frac{k}{2}$, we can bound the probability of second part by Bernstein's inequality:
    $$
        \Pr\left[\sum_{j\not=i,j>k}Z_i\ge k-1\right]\le e^{-\frac{(k/2-1)^2/2}{(n-k-1)(k/n)+(k/2-1)/3}}=o(1)
        ~.
    $$
    Therefore the total probability is between
    $\frac{k-\epsilon}{4n}\cdot (1-o(1))$ and $\frac{k+\epsilon}{4n}$.
\end{proof}

Now we can prove Theorem~\ref{thm:matroid_hardness}:
\begin{align*}
& \opt(\mathbf{D}) -  \rev(A(\mathbf{D}), \mathbf{D})  \\[1.5ex]
& \qquad =  \int_{\mathbf{v}} \left( \max_{|J|=k} \sum_{j \in J} \phi_j(v_j) - \sum_{j \in I_{A(\mathbf{D}^j)}(\mathbf{v})} \phi_j(v_j) \right) d \mathbf{D} && \\
& \qquad \ge \sum_{i \in \mathcal{B}_{\mathbf{D},A(\mathbf{D})}} \int_{\mathbf{v} \in \mathcal{V}_i} L_i(A(\mathbf{D}^j),\mathbf{v}) d \mathbf{D}  &&(\text{Equation~\ref{eqn:matroid-virtual-value-loss}}) \\
& \qquad = \sum_{i \in \mathcal{B}_{\mathbf{D},A(\mathbf{D})}} \E_{\mathbf{v} \sim \mathbf{D} : \mathbf{v} \in \mathcal{V}_i} \left[ L_i(A(\mathbf{D}^j),\mathbf{v})  \right] \cdot \Pr_{\mathbf{v} \sim \mathbf{D}} \big[ \mathbf{v} \in \mathcal{V}_i \big] \\
& \qquad \ge \sum_{i \in \mathcal{B}_{\mathbf{D},A(\mathbf{D})}} \frac{\epsilon}{3k} \cdot \Pr_{\mathbf{v} \sim \mathbf{D}} \big[ \mathbf{v} \in \mathcal{V}_i \big] && \text{(definition of $\mathcal{B}_{\mathbf{D},A(\mathbf{D})}$)} \\
& \qquad = \Theta(\epsilon) &&\text{(Lemma~\ref{lem:meta_hardness_bad_instance_count_const_prob} and Lemma~\ref{lem:matroid_probmass_vi})}
~.
\end{align*}

\end{document}